%% file: draft.tex
\documentclass{lmcs}
\pdfoutput=1

\usepackage{lastpage}
\lmcsdoi{20}{3}{14}
\lmcsheading{}{\pageref{LastPage}}{}{}%
{Jun.~30,~2023}{Aug.~07,~2024}{}

\usepackage{stmaryrd}  
\usepackage{tikz}
\usetikzlibrary{arrows,positioning,fit,calc}
\usepackage[utf8]{inputenc}
\usepackage{amsmath}

\usepackage{amsbsy}
\usepackage{bussproofs}
\usepackage{cite} 
\usepackage{caption}
\usepackage{subcaption}
\usepackage{xspace,centernot}
\usepackage[normalem]{ulem}
\newcommand*{\nat}{\ensuremath{\mathbb{N}}}

\newcommand{\logicize}[1]{{\ensuremath{\mathbf{#1}}}\xspace}
\newcommand{\klogic}{\logicize{K}}
\newcommand{\dlogic}{\logicize{D}}
\newcommand{\tlogic}{\logicize{T}}
\newcommand{\blogic}{\logicize{B}}
\newcommand{\kf}{\logicize{K4}}
\newcommand{\df}{\logicize{D4}}
\newcommand{\Bf}{\logicize{B4}}
\newcommand{\sr}{\logicize{S4}}
\newcommand{\sv}{\logicize{S5}}
\newcommand{\kv}{\logicize{K5}}

\newcommand{\Inv}{\mathsf{Inv}}
\newcommand{\cl}{\mathsf{cl}}
\newcommand\changeJ[1]{#1}
\newcommand\change[1]{#1}
\newcommand{\LG}{\logicize{L}}

\usepackage{complexity}
\renewclass{\EEXP}{2EXP}
\renewcommand{\M}{\mathcal{M}}

\newcommand{\G}{\logicize{G}}

\usepackage{amsmath}
\usepackage{amsthm}
\usepackage{amsfonts}
\usepackage{amssymb}

\usepackage{multicol}
\usepackage{thmtools}

\theoremstyle{remark}
\newtheorem{remark}[thm]{Remark}
\newtheorem{example}[thm]{Example}

\theoremstyle{definition}
\newtheorem{translation}[thm]{Translation}

\newcommand{\mycap}[1]{\text{\sc#1}}

\newcommand{\trueset}[1]{{\left\llbracket #1 \right\rrbracket}}
\newcommand{\diam}[1]{\langle #1 \rangle}

\newcommand{\stbox}[1]{[ #1 ]}
\newcommand{\al}{\alpha}
\newcommand{\alb}{\beta}

\newcommand{\proc}{\ensuremath{W}}  

\newcommand{\impl}{\rightarrow}
\newcommand{\tran}{\ensuremath{R}}
\newcommand{\trana}{\ensuremath{R_\al}}

\newcommand{\false}{\mathtt{ff}}
\newcommand{\true}{\mathtt{tt}}

\newcommand{\act}{\mycap{Ag}}

\renewcommand{\R}{{\mathcal{R}}}

\newcommand{\ML}{modal logic}
\newcommand{\NI}{Negative Introspection}

\newcommand{\mn}{\ensuremath{\mu}}
\newcommand{\mx}{\ensuremath{\nu}}

\newcommand{\transl}[3]{\ensuremath{\mathsf{F}_{#1}^{#2}({#3})}}

\newcommand{\nStt}[2]{.#1\langle#2\rangle}
\newcommand{\fix}{\ensuremath{\mathsf{fx}}}
\newcommand{\form}{\ensuremath{\Phi}}

\newcommand{\closed}{\texttt{x}}

\newcommand{\resp}{\emph{resp.}~}

\newcommand{\sub}{\ensuremath{\mathsf{sub}}}
\newcommand{\subc}{\ensuremath{\overline{\mathsf{sub}}}}

\newcommand{\ax}{\ensuremath{\mathsf{ax}}}

\newcommand{\finb}{\mathsf{fp}}
\newcommand{\infb}{\mathsf{ip}}

\begin{document}

\title[Complexity results for modal logic with recursion]{Complexity results for modal logic with recursion via translations and tableaux}

\author[L. Aceto]{Luca Aceto\lmcsorcid{0000-0002-2197-3018}}[a,b]

\author[A. Achilleos]{Antonis Achilleos\lmcsorcid{0000-0002-1314-333X}}[a]

\author[E. Anastasiadi]{Elli Anastasiadi\lmcsorcid{0000-0001-7526-9256}}[a,c]

\author[A. Francalanza]{Adrian Francalanza\lmcsorcid{0000-0003-3829-7391}}[d]

\author[A. Ing\'{o}lfsd\'{o}ttir]{Anna Ing\'{o}lfsd\'{o}ttir\lmcsorcid{0000-0001-8362-3075}}[a]

\address{Dept. of Computer Science, Reykjavik University,
	Reykjavik, Iceland}	
\email{luca@ru.is, antonios@ru.is, annai@ru.is}  

\address{Gran Sasso Science Institute
L’Aquila, Italy}
\address{Dept. of Information Technology, Uppsala University, Uppsala, Sweden}
\email{elli.anastasiadi@it.uu.se}

\address{Dept. of Computer Science, ICT, University of Malta, Msida, Malta}
\email{adrian.francalanza@um.edu.mt}

\thanks{ 
	This work has been funded by the projects ``Open Problems in the Equational Logic of Processes (OPEL)'' (grant no.~196050), ``Epistemic Logic for Distributed Runtime Monitoring'' (grant no.~184940), ``Mode(l)s of Verification and Monitorability'' (MoVeMent) (grant no~217987) of the Icelandic Research Fund, and the project ``Runtime and Equational Verification of Concurrent Programs'' (ReVoCoP) (grant 222021) of the Reykjavik University Research Fund.}

	
	
\begin{abstract} 

This paper studies the complexity of classical modal logics and of their extension with fixed-point operators, using 
translations to transfer results across logics. In particular, we show several complexity results
for 
multi-agent logics via translations to and from the $\mu$-calculus and modal logic, which allow us to transfer known upper and lower bounds. We also use these translations to introduce terminating and non-terminating tableau systems for the logics we study, based on Kozen's tableau for the $\mu$-calculus and the one of Fitting and Massacci for modal logic.
Finally, we describe these tableaux with $\mu$-calculus formulas, thus reducing the satisfiability of each of these logics to the satisfiability of the $\mu$-calculus, resulting in a general \EEXP\ upper bound for satisfiability testing.
\end{abstract}

\maketitle

\section{Introduction}\label{intro_multi_ml}
\input{introduction.tex}

\section{Definitions and Background}\label{sec:multi_ml_background} 
\input{background.tex}

\subsection{Background on Fixed Points}\label{subsec:fixpoints}

In this subsection, we mention useful facts about fixed-point formulas.

We fix a model $\M$, a formula $\mn X.\varphi$, and an environment $\rho$.
Let $S_X^0 = \emptyset$; for every successor ordinal, let 
$S_X^{\zeta + 1} = \trueset{\varphi,\rho[X \mapsto S_X^\zeta]}_\M$, and for every limit ordinal 
$\zeta$, let 
$S_X^{\zeta} = \bigcup_{\eta < \zeta} S_X^\eta$.

\begin{lem}[Iterative Characterization of the Least Fixed Point]\label{lem:iterativeTarski}
    For every model $\M$, environment $\rho$, and $\mn X.\varphi \in \LG$, there exists some ordinal $\xi$ such that 
    $\trueset{\mn X\varphi,\rho}_\M = \trueset{\varphi,\rho[X \mapsto S_X^\xi]}_\M = S_X^{\xi }$. 
\end{lem}

\begin{proof}
    First we observe that the map $S \mapsto \trueset{\varphi,\rho[X \mapsto S]}_\M$ is monotone, and therefore by the Knaster-Tarski Theorem \cite{tarski1955lattice,knaster1928theoreme}, it has a least fixed-point, which is $\trueset{\mn X\varphi,\rho}_\M = \bigcap \{S \mid S \supseteq \trueset{\varphi,\rho[X \mapsto S]}_\M $. 
    Moreover, if $\xi \leq \zeta$, then $S^{\xi}$ is included in $S^{\zeta}$, which means that the approximation forms a non-decreasing chain w.r.t. set inclusion. 
    By transfinite induction, for every ordinal $\zeta$, $S_X^\zeta \subseteq \trueset{\mn X\varphi,\rho}_\M$, and by Hartogs' Theorem on well-ordered sets \cite{hartogs1915problem}, there is an ordinal $\xi$ such that $S_X^\xi = \trueset{\varphi,\rho[X \mapsto S_X^\xi]}_\M$. 
    As $\trueset{\mn X\varphi,\rho}_\M$ is a least fixed point, the lemma follows.
\end{proof}

Let $\M = (W,R,V)$ be a $\logicize{L}$-model and consider a subformula-closed set $\Phi = \sub(\varphi)$ for some $\varphi \in L$, such that for every $X \in \Phi$, there exists a unique fixed-point formula $\fix(X) \in \Phi$ that binds $X$. 
		Let the environment $\rho$ be such that $\rho(X) = \trueset{\fix(X),\rho}_\M$ for every variable $X$.
		We say that $\rightarrow \subseteq (W\times L)^2$ is a dependency relation on $\M$ and $\varphi$ when it satisfies the following conditions:
		\begin{itemize}
			\item if $u\models_\rho \psi_1 \land \psi_2$, then $(u,\psi_1 \land\psi_2) \rightarrow (u,\psi_1)$ and $(u,\psi_1 \land\psi_2) \rightarrow (u,\psi_2)$;
			\item if $u\models_\rho \psi_1 \lor \psi_2$, then $(u,\psi_1 \lor\psi_2) \rightarrow (u,\psi_1)$ and $u \models_\rho \psi_1$ or \changeJ{$(u,\psi_1 \lor\psi_2) \rightarrow (u,\psi_2)$} and $u \models_\rho \psi_2$;
			\item if $u\models_\rho [\al]\psi$, then $(u,[\al]\psi) \rightarrow (v,\psi)$ for all $v \in W$, such that $u R_\al v$;
			\item if $u\models_\rho \diam{\al}\psi$, then $(u,\diam{\al}\psi) \rightarrow (v,\psi)$ for some $v \in W$, such that $u R_\al v\models \psi$;
			\item if $u\models_\rho \mn X.\psi$ or $u\models_\rho \mx X.\psi$, then $(u,\fix(X)) \rightarrow (u,\psi)$; and
			\item if $u\models_\rho X$, then $(u,X) \rightarrow (u,\fix(X))$.
		\end{itemize}
		For each variable $X$ and dependency relation $\rightarrow$, we also define $\xrightarrow{X}$, such that $(w_1,\psi_1) \xrightarrow{X} (w_2,\psi_2)$ whenever $(w_1,\psi_1) \rightarrow (w_2,\psi_2)$ and $\psi_1 \neq Y$ for all variables $Y$ where $\fix(Y)$ is not a subformula of $\fix(X)$.
		We call a dependency relation $\rightarrow$ lfp-finite, if every least-fixed-point variable $X$ appears finitely many times  on every $\xrightarrow{X}$-sequence.
		
		\begin{thm}\label{thm:lfp-finite-dep}
            Let $\M = (W,R,V)$ be a $\logicize{L}$-model, $w\in W$, and $\varphi \in L$, such that $\M,w\models \varphi$. 
            Then, there exists an lfp-finite dependency relation on $\M$ and $\varphi$. 
        \end{thm}
        \begin{proof}
		The claim amounts to the memoryless determinacy of parity games and our proof is inspired by \cite{zielonka1998infinite}.
        We prove the existence of an lfp-finite dependency relation on $\M$ and $\varphi$ by induction on $\varphi$.
        To do so, and to handle the case of free recursion variables, 
        we can extend the definition of a dependency relation to formulas with free variables, such that the condition that if $u\models_\rho X$, then $(u,X) \rightarrow (u,\fix(X))$, is omitted when $X$ is free in $\varphi$.
        For the remainder of this proof, we maintain that if $(u,\psi)$ appears in a pair in a dependency relation, then $u \models_\rho \psi$.

        All cases for $\varphi$ are straightforward, except for the fixed points.
        If $\varphi = \mx X.\psi$, then we can construct a lfp-finite dependency relation $\rightarrow$ on $\M$ and $\psi$, and then add the pairs $(\varphi,\psi)$ and $(X,\varphi)$ in $\rightarrow$. If $Y$ is a least-fixed-point variable, then any new $\rightarrow$-path where $Y$ appears infinitely often would go through $X$ infinitely often, and therefore $\rightarrow$ is still lfp-finite.
        If $\varphi = \mn X.\psi$, then by Lemma \ref{lem:iterativeTarski}, there exists some ordinal $\xi$, with $\trueset{\varphi,\rho} = \trueset{\psi,\rho[X \mapsto S_X^\xi]}$.
        For each $\zeta \leq \xi$, we define $\rightarrow_\zeta$ and prove that it is an lfp-finite dependency relation on $\M$ and $\psi$ for $\rho[X \mapsto S_X^\xi]$, where $X$ does not appear infinitely often on any $\xrightarrow{X}_\zeta$ path:
        \begin{description}
            \item[Case $\zeta =0$] for each subformula $\chi$ of $\psi$, 
            let $V_\chi^{0} = \trueset{\chi,\rho[X \mapsto \emptyset]}$.
            Then, $\rightarrow_0$
            is an lfp-finite dependency relation on $\M$ and $\psi$, for $\rho[X \mapsto \emptyset]$ with the additional pairs $((u,\varphi),(u,\psi))$, for every $u \in V_\psi^0$.
            
            \item[Case $\zeta = \xi+ 1$ for some $\xi$] then for each subformula $\chi$ of $\psi$, 
            let 
            \[
                V_\chi^{\xi + 1} = \trueset{\chi,\rho[X \mapsto S_X^{\xi + 1}]}\setminus \trueset{\chi,\rho[X \mapsto S_X^\xi]}.  
            \]
            Let $\rightarrow_{\xi +1}^0$ be an lfp-finite dependency relation on $\M$ and $\psi$, for $\rho[X \mapsto S_X^{\xi + 1}]$.
            Then, we define 
            \[\rightarrow_{\xi+1} = \{((u,\chi),x) \mid 
            (u,\chi) \rightarrow_{\xi} x; ~~ \text{ or } 
            u \in V_\chi^{\xi +1 } \text{ and } (u,\chi) \rightarrow_{\xi +1}^0 x\} \]
            \[
            \cup \{((u,X),(u,\varphi)), ((v,\varphi),(v,\psi)) \mid u \in S_X^{\xi}, v \in V_\chi^{\xi + 1}
            .
            \]
            \item[if $\zeta$ is a limit ordinal] we define 
                $
                    \rightarrow_{\zeta} = \bigcup_{\eta < \zeta} \rightarrow_{\eta} 
                $.
        \end{description}
        It is straightforward to verify that $\rightarrow_{\zeta}$ is a dependency relation on $\M$ and $\varphi$ for $\rho$.
            Notice that if $(u,X) \rightarrow_{\zeta} (u,\varphi)$, then $u \in V_\psi^{\eta}$ for some $\eta < \zeta$. Therefore, no $\xrightarrow{X}_\zeta$-path can have infinite occurrences of $X$, and $\rightarrow_{\zeta}$ is an lfp-finite dependency relation on $\M$ and $\varphi$ for $\rho$, thus completing the proof.
        \end{proof}                                     
%

\section{Complexity Through Translations}\label{sec:translations}
\input{translations.tex}
\input{translations-lower.tex}

\input{translations-complexity.tex}


\section{Tableaux for \texorpdfstring{$\LG_k^\mu$}{L\textunderscore k\textasciicircum mu}}
\label{sec:multi}
\input{tableaux-multi.tex}
\section{The case of \texorpdfstring{$\kf^\mu$}{K4\textasciicircum mu}}
\label{sec:KfourMu} 
\input{k4mu_pspace.tex}

\section{Translations through Tableaux}
\label{sec:TtT} 
\input{translations-tableaux.tex}

\section{Conclusions}
\label{sec:multi_ml_conclusion}
\input{conclusions.tex}

\section*{Acknowledgements}
We thank the anonymous reviewers for their insightful comments and criticisms on the submitted manuscript. We are most grateful to 
Laura Bozzelli, whose detailed feedback and pointers led to substantial improvements in both the presentation and the technical contents in this article.

%
%
%
\bibliographystyle{alphaurl}
\bibliography{mybib}
%

%

\end{document}

%% file: introduction.tex

In this paper, we study a family of multi-modal logics with fixed-point operators that are interpreted over restricted classes of Kripke models.
One can consider these logics as extensions of the usual multi-agent logics of knowledge and belief  \cite{Fagin1995ReasoningAboutKnowledge} by adding recursion to their syntax or of the $\mu$-calculus \cite{Kozen1983} by interpreting formulas over different classes of frames and thus giving an epistemic interpretation to the modalities.
We define \emph{translations} between these logics, and we demonstrate how 
one can rely on these translations to
prove  finite-model theorems, complexity bounds, and tableau termination for each logic in the family,
in a systematic fashion. 

Modal logic comes in several variations \cite{blackburn2006handbook}. 
Some of these, such as multi-modal logics of knowledge and belief  \cite{Fagin1995ReasoningAboutKnowledge}, are of particular interest to Epistemology and other application areas. 
Semantically, the classical modal logics used in epistemic (but also other) contexts result from imposing certain restrictions on their models.
On the other hand, the modal $\mu$-calculus \cite{Kozen1983} can be seen as an extension of the smallest normal modal logic $\klogic$
 with greatest and least fixed-point operators, $\mx X$ and $\mn X$ respectively.

In this article, we explore the situation where one allows both recursion (\emph{i.e.} fixed-point) operators in a multi-modal language and imposes restrictions on the semantic models.
In particular, we are interested in the complexity of satisfiability for the resulting logics.
Satisfiability for the 
$\mu$-calculus is known to be \EXP-complete \cite{Kozen1983}, while for the modal logics between $\klogic$ and $\sv$ the problem is \PSPACE-complete or \NP-complete, depending on whether they have \NI\ \cite{ladnermodcomp,Halpern2007Characterizing}.
In the multi-modal case, satisfiability for those modal logics becomes \PSPACE-complete, and is \EXP-complete with the addition of a common knowledge operator \cite{Halpern1992}.

Our primary method for proving complexity results is through translations to and from the multi-modal $\mu$-calculus.
We show that we can use surprisingly simple translations from modal logics without recursion to the base modal logic $\klogic_n$, that is, K with $n$ agent-indexed modalities, reproving the \PSPACE\ upper bound for these logics (Theorem \ref{thm:no-fixpoint-translations-work} and Corollary \ref{cor:modalupper}).
These translations and our constructions to prove their correctness do not generally transfer to the corresponding logics with recursion.
We present translations from specific logics to the $\mu$-calculus and back, and we discuss the remaining open cases.
We discover, through the existence of our translations, that several restrictions induced on the frames do not affect the complexity of the satisfiability problem.
	As a result, we prove that all logics with axioms among $D$, $T$, and $4$, and the  least-fixed-point fragments of logics that also have $B$,  
have their satisfiability in \EXP, and a matching lower bound for the logics with axioms from $D,T$, and $B$ (Corollaries \ref{cor:mu-upper1} and \ref{cor:mu-upper2}).%

To address the limitations of the translation-based approach to proving complexity results, we then present tableaux for the discussed logics, based on the ones by Kozen for the $\mu$-calculus \cite{Kozen1983}, and by Fitting and Massacci for modal logic \cite{fitting1972tableau,Massacci1994}.
We give tableau-termination conditions for every logic with a finite-model property (Theorem \ref{thm:tableaux}).
Based on these tableaux, we give a polynomial-space decision procedure for the single-agent logics with axioms among $D$, $T$, and $4$, thus proving that satisfiability for these logics is \PSPACE-complete (Theorem \ref{thm:4-PSPACE}).
Finally, we give a general satisfiability-preserving translation from each logic to the $\mu$-calculus, using our tableaux. The translation describes a tableau branch with an exponentially larger $\mu$-calculus formula, establishing a \EEXP-upper bound for all our logics. 
The decidability of the satisfiability problem for logics with $B$ or $5$ was an open problem (see, for example, Remark \ref{remark:finitemodelB5}, based on \cite{Dagostino2013S5}).

The addition of recursive operators to \ML\ increases expressiveness.
An important example is that of \emph{common knowledge} or \emph{common belief}, which can be expressed with a greatest fixed-point thus:
 $\mx X.(\varphi \land \bigwedge_\al [\al] X)$, 
 where $\varphi$ describes the property that is common knowledge, and $\al$ ranges over a (finite) set of agents. 
However, the combination of epistemic logics and fixed-points
can potentially express more interesting epistemic concepts.
For instance, in the context of a belief interpretation,  the formula
$\mn X. \bigvee
_{\al 
}(
[\al]\varphi \vee [\al] X)$, 
can be viewed as claiming that there is a rumour of $\varphi$.
It would be interesting to see what other meaningful sentences of epistemic interest one can express using recursion. 
Furthermore, the family of logics we consider allows each agent to behave according to a different logic. 
This flexibility allows one to mix different interpretations of modalities, such as a temporal interpretation for one agent and an epistemic interpretation for another.
Such logics can even resemble hyper-logics \cite{Hyperproperties} if a set of agents represents different streams of computation.
Combinations of epistemic and temporal or hyper-logics have recently been used to express safety and privacy properties of systems \cite{EpistLogicRV}.

The paper is organized as follows. Section \ref{sec:multi_ml_background} gives the necessary background and an overview of current results.
Section \ref{sec:translations} defines a class of translations that provide us with several upper and lower bounds, and identifies conditions under which they can be composed.
In Section \ref{sec:multi} we give tableaux systems for our multi-modal logics with recursion. 
Section \ref{sec:KfourMu} proves a \PSPACE\ upper bound for the single-agent logics with axioms among $D$, $T$, and $4$; Section \ref{sec:TtT} presents the translation that establishes the \EEXP\ upper bounds for all the modal logics with recursion.
We conclude in Section \ref{sec:multi_ml_conclusion} with a set of open questions and directions.

This is an extended version of \cite{complexity-translations-gandalf}, with complete proofs for our results. The results from Sections \ref{sec:KfourMu} and \ref{sec:TtT} are also new in this version.

%% file: background.tex

This section introduces the logics that we study and the necessary background on the complexity of  \ML\ and the $\mu$-calculus.

\subsection{The Multi-Modal Logics with Recursion}
We start by defining the syntax of the logics.
\begin{defi}
	We consider formulas  constructed from the following grammar:
	\begin{align*}
	\varphi,\psi \in L
	::&=
	p &&|~~
	\neg p &&|~~
	\true                     &&|~~\false      &	&|~~X     
	&&|~~\varphi\land\psi            &&|~~\varphi\lor\psi    \\
	&|~~\langle\al\rangle\varphi &&|~~[\al]\varphi    
	&&|~~\mn X.\varphi              &&|~~\mx X.\varphi                                                     ,
	\end{align*}
	where $X$ comes from a countably infinite set of logical (or fixed-point) variables, \change{$\mycap{LVar}$}, $\al$ from a finite set of agents, $\act$, and $p$ from a finite 
	set of propositional variables\change{, $\mycap{PVar}$}. 
	When $\act = \{\al \}$, 
	$\Box \varphi$ 
	stands for
	$[\al]\varphi$, and $\Diamond \varphi$ 
	for
	$\diam{\al}\varphi$. 
	We also write $[A]\varphi$, with $A \subseteq \act$, to mean $\displaystyle\bigwedge_{\al \in A} [\al]\varphi$ and $\diam{A}\varphi$ 
	\change{for} 
	$\displaystyle\bigvee_{\al \in A} \diam{\al}\varphi$. As usual, 
	an empty conjunction stands for $\true$ and and empty disjunction stands for $\false$.
\end{defi}

A formula is closed when every occurrence of a variable $X$  is in the scope of some recursive operator $\mx X$ or $\mn X$. 
Henceforth we 
consider
only 
closed formulas, unless we specify otherwise.

\change{Moreover, for recursion-free closed formulas we associate the notion of \textit{modal depth}, which is the maximum nesting depth 
	of the modal operators%
	\footnote{\change{The modal depth of recursive formulas can be either zero, or infinite. However, this will not affect the results we present later \changeJ{on.}}}. The modal depth of $\varphi$ is defined inductively as:
\begin{itemize}
\item $md(p) = md(\neg p) = md(\true) = md(\false) = 0$, where $p \in \mycap{PVar}$,
\item $md(\varphi \lor \psi) = md(\varphi \land \psi) = \max(md(\varphi),md(\psi))$, and
\item $md(\stbox{\al}\varphi) = md(\diam{\al}\varphi)=  1 + md(\varphi)$, where $\al \in \act$.
\end{itemize}}
We assume that each recursion variable $X$ appears in a unique fixed-point formula $\fix(X)$, which is either of the form $\mn X.\varphi$ or $\mx X.\varphi$.
If $\fix(X)$ is a least-fixed-point (\resp greatest-fixed-point) formula, then $X$ is called a least-fixed-point (\resp greatest-fixed-point) variable.
We can define a partial order over fixed-point variables, such that 
$X\leq Y$ iff $\fix(X)$ is a subformula of $\fix(Y)$, and $X < Y$ when $X \leq Y$ and $X \neq Y$.
If $X$ is $\leq$-minimal among the free variables of $\varphi$, then 
we define the \emph{closure} of $\varphi$ to be 
\label{closure-formula}
\changeJ{$\cl(\varphi)=\cl(\varphi[\fix(X)/X])$},
 where $\varphi[\psi/X]$ is the usual substitution operation.
 Note that if $\varphi$ is closed, then 
 $\cl(\varphi)=\varphi$.
 
 We define $\sub(\varphi)$ as the set of subformulas of $\varphi$. The size of $\varphi$, denoted by $|\varphi|$, is the cardinality of $\sub(\varphi)$, which is bounded by the length of $\varphi$ as a string of symbols.
 Negation, $\neg \varphi$, and implication, $\varphi \impl \psi$, can be defined in the usual way.
 Then, we define $\subc(\varphi) = \sub(\varphi) \cup \{ \neg \psi \mid \psi \in \sub(\varphi) \}$.
 


\paragraph{Semantics}

We interpret formulas on 
the states of a
\emph{Kripke model}.
%
A Kripke model, or simply 
model, 
is
a triple
\changeJ{$\M=$} $( \proc,\tran,V)$
where $\proc$ is a nonempty set of states, 
${\tran}\subseteq\proc\times\act \times \proc$ is a transition relation, and $V:\proc\to 2$\change{$^{\mycap{PVar}}$} determines the propositional variables that are true at each state. 
$( \proc,\tran)$ is called a \emph{frame}.
We usually write $(u,v)\in \trana$ or $u \trana v$ instead of $(u,\al,v)\in \tran$, or $u \tran v$, \change{when $\act$ is a singleton $\{\al\}$}.

Formulas are evaluated in the context of 
an \emph{environment} $\rho:\mycap{LVar}\to 2^\proc$, which gives values to the logical variables. 
For an environment $\rho$, variable $X$, and set $S \subseteq \proc$, we write $\rho[X \mapsto S]$ for the environment that maps $X$ to $S$ and all $Y \neq X$ to $\rho(Y)$.
The semantics for 
our
formulas is given through a function $ \trueset{ -  }_\M$, defined in Table \ref{table:semantics}.
\begin{table}
\begin{align*}
\noalign{\noindent
$ \trueset{\true, \rho }= \proc, 
\hfill 
 \trueset{\false,\rho }=\emptyset ,
\hfill
 \trueset{ p, \rho } = \{s \mid p\in V(s) \}
,
\hfill
 \trueset{ \neg p, \rho } = \proc {\setminus}  \trueset{ p, \rho } 
,
$}
 \trueset{[\al]\varphi,\rho }&=\left\{s~\big|~ \forall t. \ sR_{\al}t\text{ implies } t\in \trueset{\varphi,\rho }\right\},
& \trueset{\varphi_1{\land}\varphi_2, \rho }&= \trueset{\varphi_1,\rho }\cap \trueset{\varphi_2,\rho }, 
\\
 \trueset{\langle\al\rangle\varphi,\rho }&=\left\{s~\big|~ \exists t. \ sR_{\al}t\text{ and } t\in \trueset{\varphi,\rho }\right\},
 & 
 \trueset{\varphi_1{\lor}\varphi_2, \rho }&= \trueset{\varphi_1,\rho }\cup \trueset{\varphi_2,\rho } , 
\\
 \trueset{\mn X.\varphi,\rho }&=\bigcap\left\{S~\big|~ S \supseteq  \trueset{\varphi,\rho[X\mapsto S] }\right\},
 &
 \trueset{ X, \rho }& = \rho(X),
\\
 \trueset{\mx X.\varphi,\rho }&=\bigcup\left\{S~\big|~ S \subseteq  \trueset{\varphi, \rho[X\mapsto S] }\right\} .
\end{align*}
\caption{Semantics of modal formulas on a model $\M = (W,R,V)$. We omit $\M$ from the 
	notation.}
\label{table:semantics}
\end{table}
The semantics of $\neg \varphi$ 
is constructed 
as usual, where $ \trueset{ \neg X, \rho }_\M = \proc {\setminus} \rho(X)$.

We call a pair $(\M,s)$ of a model and a state a \emph{pointed model}.
We sometimes use  
$(\M,s) \models_\rho \varphi$ for $s \in \trueset{\varphi,\rho}_\M$, and 
as the environment 
has no effect on the semantics of a closed formula $\varphi$, 
we 
often 
drop it from the notation 
 and 
write 
$(\M,s) \models \varphi$ or $s \in \trueset{\varphi}_\M$.
%
If $(\M,s) \models \varphi$, we say that $\varphi$ is true, or satisfied, in $s$. 
When the particular model does not matter, or is clear from the context, we may omit it. 


Depending on how we 
further 
restrict our 
syntax 
and 
the model, we can describe several logics.
Without further restrictions, the resulting logic 
is the $\mu$-calculus \cite{Kozen1983}. The max-fragment (resp. min-fragment) of the $\mu$-calculus is the fragment that only allows 
the $\mx X$ (resp. the $\mn X$) recursive operator.
If 
$|\act| = k \in \mathbb{N}^+$
and we allow no recursive operators (or recursion variables), then we have the basic modal logic $\klogic_k$ (or $\klogic$, if $k=1$), and further restrictions on the frames can result in a wide variety of modal logics (see, for instance, \cite{MLBlackburnRijkeVenema}). 
We give names to the following frame conditions, or frame constraints, for the case where $\act=\{\al\}$.
These conditions correspond to the usual axioms for normal modal logics --- see \cite{blackburn2006handbook,MLBlackburnRijkeVenema,Fagin1995ReasoningAboutKnowledge}, which we will revisit in Section \ref{sec:translations}.
\begin{multicols}{2}
\begin{description}
	\item[$D$] 
	$\tran$ is \textit{serial}: 
	$\forall s\exists t, ~s  \tran t $;
	\item[$T$] $\tran$ is \textit{reflexive}: 
	$\forall s, ~s  \tran s $;
	\item[$B$] $\tran$ is \textit{symmetric}: 
	$\forall s,t,~ {s \tran t \text{ implies } t \tran s}$;
	\item[$4$] $\tran$ is \textit{transitive}: 
	$\forall s,t,r,$ if $s\tran t$ and $t \tran r$ then $s \tran r$;
	\item[$5$] $\tran$ is \textit{euclidean}: 
	$\forall s,t,r,$ if $s \tran t$ and $s \tran r$,  then $t \tran r$.
\end{description}
\end{multicols}

We consider modal logics that are interpreted over models that satisfy a combination of these constraints for each agent. 
$T$, which we call Factivity, is a special case of $D$, which we call Consistency. Constraint $4$ is 
Positive Introspection and $5$ is called Negative Introspection.\footnote{These are names for properties or axioms of a logic. When we refer to these conditions as conditions on a frame or model, we may refer to them with the name of the corresponding condition on relations, namely: seriality, reflexivity, symmetry, transitivity, and euclidicity.}
Given a logic $\LG$ and constraint $c$, we write $\LG+c$ for the logic that is interpreted over all models with frames that satisfy  all the constraints of $\LG$ and $c$.
The name of a single-agent logic is a combination of the constraints that apply to its frames, including $K$, if the constraints are contained in $\{4,5\}$.
Therefore, logic $\dlogic$
 is $\klogic + D$, $\tlogic$ is $\klogic+T$, $\blogic$ is $\klogic+B$, $\kf = \klogic + 4$, 
$\df = \klogic + D + 4 = \dlogic + 4$, 
and so on.
We use the special names $\sr$ for $\logicize{T4}$ and $\sv$ for $\logicize{T45}$.
We define a (multi-agent) logic $\LG$ on $\act$ as a map from agents to single-agent logics.
$\LG$ is interpreted on Kripke models of the form $(\proc,\tran,V)$, where for every $\al\in\act$, $(\proc,\trana)$ is a frame for $\LG(\al)$.

For a logic $\LG$, 
we write $\LG^\mu$ for the logic that results from $\LG$ after we allow recursive operators in the syntax --- in case they were not allowed in \LG.
Furthermore, if for every $\al \in \act$, $\LG(\al)$ is the same single-agent logic ${\LG}$, we write $\LG$ as ${\LG}_k$, where $|\act|=k$.
Therefore, the $\mu$-calculus is $\klogic^\mu_k$.

From now on, unless we explicitly say otherwise, by a logic, we mean one of the logics we 
have 
defined above.
%
%
We call a formula satisfiable  for a logic $\LG$, if it is satisfied in some  state of a model  for $\LG$. 

\begin{example}\label{ex:inv_common_K}
	For a formula $\varphi$, we  define $\Inv(\varphi)= \mx X.(\varphi \land [\act]X)$. $\Inv(\varphi)$ asserts that $\varphi$ is true in \emph{all} reachable states, or, alternatively, it can be read as an assertion that $\varphi$ is common knowledge.
	We  dually define 
	$Eve(\varphi)= \mn X.(\varphi \lor \diam{\act}X)$, which asserts that $\varphi$ is true in \emph{some} reachable state. 
\end{example}

In what follows, we will often use the concept of \textit{bisimilarity}, when arguing that two different Kripke models satisfy the same formulas. 

\begin{defi}[Bisimulation]\label{def:strong_bisim}
A binary relation $S \subseteq W_1 \times W_2$ over the set of states of two Kripke models $\M_1 = (W_1,R_1,V_1)$ and  $\M_2 = (W_2,R_2,V_2)$ is a bisimulation if and only if whenever $s_1 S s_2$, for each agent $\al$:
\begin{itemize}
\item for each state $s_1'$, if $(s_1, s_1') \in R_{1,\al}$, then there exists a state $s_2'$ such that  $(s_2, s_2') \in R_{2,\al}$ and $s_1' S s_2'$, 
\item for each state $s_2'$, if $(s_2, s_2') \in R_{2,\al}$, then there exists a state $s_1'$ such that  $(s_1, s_1') \in R_{1,\al}$ and $s_1' S s_2'$, and
\item $p \in V(s_1)$ iff $p \in V(s_2)$, for each propositional variable $p$.
\end{itemize}
Two pointed models $(\M_1,s_1)$ and $(\M_2,s_2)$ 
are bisimilar, notation $(\M_1,s_1) \sim (\M_2,s_2)$,  if they are related by a bisimulation $S$.
\end{defi}

It is well-known that bisimulation is an equivalence relation. 

\begin{thm}[Hennessy-Milner Theorem \cite{hennessy1985algebraic,mu_calc_model_check,handbook_model_check}]\label{thm:hmt}
    The $\mu$-calculus is bisimulation-invariant. I.e. for every formula  $\varphi \in L$  and pointed models $(\mathcal{M}_1, s_1), (\mathcal{M}_2, s_2)$,
    if $(\mathcal{M}_1, s_1) \sim (\mathcal{M}_2, s_2)$ and  $(\mathcal{M}_1, s_1) \models \varphi$,
    then $(\mathcal{M}_2, s_2) \models \varphi$.

\end{thm}

\subsection{Known Results}\label{sec:multi_ml_known_results}
For logic $\LG$, the satisfiability problem for $\LG$, or $\LG$-satisfiability, is the problem that asks, given a formula $\varphi$, if $\varphi$ is satisfiable. Similarly, the model checking problem for $\LG$ asks if $\varphi$ is true at a given state of a given \change{finite} model. 

 Ladner \cite{ladnermodcomp} established the classical result of \PSPACE-completeness for the satisfiability of \klogic, \tlogic, \dlogic, \kf, \df, and \sr\ as well as \NP-completeness for the satisfiability of \sv.
   Halpern and R\^{e}go later characterized the \NP--\PSPACE\ gap  for one-action logics by the presence or absence of Negative Introspection \cite{Halpern2007Characterizing}, resulting in Theorem \ref{thm:ladhalp} below.
    Later, Rybakov and Shkatov proved the \PSPACE-completeness of \blogic\ and \tlogic\blogic in \cite{RybakovShkatovBComp}.

There is plenty of relevant work on 
the $\mu$-calculus on restricted frames, 
mainly in its single-agent form.
Alberucci and Facchini examine the alternation hierarchy of the $\mu$-calculus over reflexive, symmetric, and transitive frames in \cite{alberucci_facchini_2009}.
D'Agostino and Lenzi have studied the $\mu$-calculus over different classes of frames in great detail.  
In \cite{DAGOSTINO20104273transitive}, they reduce the $\mu$-calculus over finite transitive frames to first-order logic.
In \cite{Dagostino2013S5}, they prove that $\sv^\mu$-satisfiability is \NP-complete, and that the two-agent version of $\sv^\mu$ does not have the finite-model property.
In \cite{Dagostino2015modal}, they consider finite symmetric frames, and they prove that $\blogic^\mu$-satisfiability is in 2\EXP, and \EXP-hard.
They also examine planar frames in \cite{DAGOSTINO201840planar}, where they show that the alternation hierarchy of the $\mu$-calculus over planar frames is infinite.

\begin{thmC}[\cite{ladnermodcomp,Halpern2007Characterizing,RybakovShkatovBComp,Dagostino2013S5}]\label{thm:ladhalp}
	If $\LG\in \{\klogic, \tlogic, \dlogic, \blogic, \tlogic\blogic, \kf, \df, \sr\}$, then 
	$\LG$-satisfiability is \PSPACE-complete; and $\LG+5$-satisfiability and $(\LG+5)^\mu$-satisfiability are \NP-complete.
\end{thmC}

\begin{thmC}[\cite{Halpern1992}]\label{thm:halpKkcomp}
	If $k>1$ and $\LG$ has a combination of constraints from $D, T, 4, 5$ and no recursive operators, then 
	$\LG_k$-satisfiability is \PSPACE-complete.
\end{thmC}

\begin{remark}\label{remark:Halpern-and-Moses-should-have-proved-more}
	We note that Halpern and Moses in \cite{Halpern1992} only prove these bounds for the cases of $\klogic_k, \tlogic_k, \sr_k, \logicize{KD45}_k,$ and $\sv_k$;
	moreover D'Agostino and Lenzi in \cite{Dagostino2013S5} only prove the \NP-completeness of satisfiability for $\sv^\mu$. However, it is not hard to see that their respective methods also work for the rest of the logics of Theorems \ref{thm:ladhalp} and \ref{thm:halpKkcomp}.
	\qed
\end{remark}


\begin{thmC}[\cite{Kozen1983}]\label{prp:muCalc-sat}
	The satisfiability problem for the $\mu$-calculus 
	is \EXP-complete.
\end{thmC}

 When a formula is given in conjunction with a model and one of its states, the question of whether the given formula holds in that fixed model at the given state is called the \textit{model checking problem}.

\begin{thmC}[\cite{emerson2001model}]\label{thm:mu-calc-MC}
	The model checking problem for the $\mu$-calculus is in $\NP \cap \coNP$.\footnote{In fact, the problem is known to be in $\UP \cap \coUP$, and can be solved in pseudo-polynomial time --- see \cite{jurdzinski1998deciding,CaludeJKLS22,JurdzinskiL17,lics_Lehtinen18}.}
\end{thmC}

Finally we have the following results about the complexity of satisfiability, when we have recursive operators. Theorem \ref{prp:fragmantsMu} and Proposition \ref{prp:EXP-hard-k2} have already been observed in \cite{AcetoAFI20}. 

\begin{thm}\label{prp:fragmantsMu}
	The satisfiability problem for the min- and max-fragments of the $\mu$-calculus is \EXP-complete, even 
	when $|\act| = 1$. 
\end{thm}

\begin{proof}[Proof sketch]
	It is known that satisfiability for the min- and max-fragments of the $\mu$-calculus (on one or more action) is \EXP-complete. It is in \EXP\ due to Theorem \ref{prp:muCalc-sat} and as shown in \cite{Pratt_1981}, 
	these fragments  suffice to describe the PDL formula that is constructed by the reduction used in \cite{fischer1979propositional} to prove \EXP-hardness for PDL. Therefore, that reduction can be adjusted to prove that 
	satisfiability for the min- and max-fragments of the $\mu$-calculus  is \EXP-complete.
\end{proof}

As observed in Example \ref{ex:inv_common_K}, one can express  that formula $\varphi$  is common knowledge using the greatest fixed-point operator. 
Since validity for $\LG_k$ with common knowledge (and without recursive operators) and $k > 1$ is \EXP-complete \cite{Halpern1992}\footnote{Similarly to Remark \ref{remark:Halpern-and-Moses-should-have-proved-more}, \cite{Halpern1992} does not explicitly cover all these cases, but the techniques can be adjusted.}, $\LG_k^\mu$ 
\change{is}
$\EXP$-hard.

\begin{prop}\label{prp:EXP-hard-k2}
	Satisfiability for $\LG_k^\mu$, where $k>1$, is \EXP-hard.
\end{prop}

%% file: translations.tex

We examine 
\LG-satisfiability and use formula translations to reduce the satisfiability of one logic to the satisfiability of another.
We investigate the properties of these translations and how they compose with each other, and we achieve complexity bounds for several logics.

In this section,
a formula translation from logic $\LG_1$ to logic $\LG_2$ is a mapping $f$ on formulas such that each formula $\varphi$ is $\LG_1$-satisfiable if and only if $f(\varphi)$ is $\LG_2$-satisfiable.
We only consider translations that can be computed in polynomial time, and therefore, our translations are polynomial-time reductions, transferring complexity bounds between logics.

According to Theorem \ref{prp:muCalc-sat}, $\klogic^\mu_k$-satisfiability, that is, satisfiability for the $\mu$-calculus, is \EXP-complete, and therefore for each logic $\LG$, we aim to connect $\klogic^\mu_k$ and $\LG$ via a sequence of translations in either direction, to prove complexity bounds for $\LG$-satisfiability.

%
%

\subsection{Translating Towards \texorpdfstring{$\klogic_k$}{K\textunderscore k}}
We begin by presenting translations from logics with more frame conditions to logics with fewer ones.
To this end, we study how taking the closure of a frame under one condition affects other frame conditions.

\subsubsection{Composing Frame Conditions}
\label{subsec:compositionality}

We now discuss how the conditions for frames affect each other. 
For example, to construct a transitive frame, one can take the transitive closure of a possibly non-transitive frame. 
The resulting frame will satisfy condition $4$. 
As we will see, taking the closure of a frame under condition $x$ may affect whether that frame maintains condition $y$, depending on $x$ and $y$.
In the following, we observe that one can apply the frame closures in certain orders that preserve the properties one acquires with each application.
We begin by formally defining the closure of a frame with respect to some frame condition. 
\changeJ{\begin{defi}\label{def:closures}
    Let $F = (W,R)$ be a frame, and $x$ a frame condition among \changeJ{$D,$} $ T, B, 4,$ and $ 5$. The closure of $R_\al$ under $x$ is denoted as $R_{\al}^{x}$ and defined by: 
    \begin{itemize}
        \item $R_\al \cup \{(u,u) \mid u \in W \text{ and } \forall v \in W.~(u,v)\notin R_\al\}$, if $x = D$.
        \item $R_\al \cup \{(u,u) \mid u \in W\}$, if $x = T$.
        \item $R_\al \cup \{(u,v) \mid (v,u) \in R_\al\}$, if $x = B$.
        \item $R_\al^{+}$ (i.e., the transitive closure of $R_\al$), if $x = 4$.
        \item $R_\al^{5}$ is the least relation that includes $R_\al$ such that if $(s,u), (s,v) \in R_\al^{5}$, then 
        $(u,v) \in R_\al^{5}$.
    \end{itemize}
\end{defi}
The definition of closure given above can be generalised to sets of agents thus: }
\begin{defi}
Let $F = (W,R)$ be a frame, and $x$ a frame restriction among \changeJ{$D,$} $ T, B, 4,$ and $ 5$.
For \changeJ{a set of } agents \changeJ{$A \subseteq \act$}, 
${R}^{x,A}$ is defined by
\begin{itemize}
\item ${R}^{x,A}_\beta = {R_\beta^{x}}$, if $\beta \in A$, and
\item ${R}^{x,A}_\beta = R_\beta$, otherwise.
\end{itemize}
We define ${F}^{x,A} = (W,{R}^{x,A})$.
\end{defi}

 We make the following observation.



\begin{lem}\label{lem:conditionsarepreserved}
	Let $x$ be a frame condition among $D, T, B, 4,$ and $5$, and $y$ a frame restriction among $T, B, 4,$ and $5$, such that $(x,y) \not\in \{(4,B), (5,T), (5,B)\}$.
	Then, for every frame $F$ that satisfies $x$, for some agent $\al$, ${F}^{y,\al}$ also satisfies $x$.
%
	%
	%
\end{lem}

According to Lemma \ref{lem:conditionsarepreserved}, frame conditions are preserved as seen in Figure \ref{fig:frame_hierarchy}. 
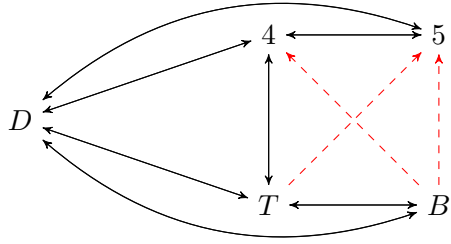
\begin{figure}
	\begin{center}
		\begin{tikzpicture}[scale=0.65,->,>=stealth',node distance=1.6cm, 
			main node/.style={circle,draw,font=\large\bfseries}]
			
			\node (c) { };
			\node (4) [above left of=c] {$4$};
			\node (t) [below left of=c] {$T$};
			\node (b) [below right of=c] {$B$};
			\node (5) [above right of=c] {$5$};
			\node (d) [left= 4cm of c] {$D$};
			
			\path 
			(t) edge (4)
			edge [red,dashed] (5)
			edge (b)
			edge (d)
			

			(d) edge (t) 
			edge [bend right] (b)
			edge (4)
			edge [bend left] (5)
			
			(4) 
			edge (d)
			edge (5)
			edge (t)
			(b) edge [red,dashed] (4)
			edge [red,dashed] (5)
			edge (t)
			edge [bend left] (d)
			(5) edge (4)
			edge [bend right](d)
			;
		\end{tikzpicture}
	\end{center}
	\caption{The frame-property-preservation graph}
	\label{fig:frame_hierarchy}
\end{figure}
%
%
%
In Figure \ref{fig:frame_hierarchy}, an arrow from $x$ to $y$ indicates that property $x$ is preserved through the closure of a frame under $y$. 
Dotted red arrows indicate one-way arrows.

\begin{remark}
	We note that, in general, not all frame conditions are preserved by the closure with respect to some other condition. For example, the accessibility relation $\{(s,t),(t,t)\}$ is euclidean, but its reflexive closure $\{(s,t),(t,t),(s,s)\}$ is not.
\end{remark}

There is at least one linear ordering of the frame conditions $D$, $T$, $B$, $4$, and $5$, such that all preceding conditions are preserved 
by
closures under the conditions following them. We call such an order a closure-preserving order.
We use the linear order $D, T, B, 4, 5$ in the rest of the paper.

For each logic $\LG$, $\al \in \act$, and frame $F = (W,R)$, 
we write $R_\al^\LG$ for the relation that results by applying the closures of the conditions of $\LG(\al)$ following the above-mentioned linear order.
For example, $R_\al^{\Bf_{\al}}$ is the relation we obtain by first applying the symmetric closure, and then the reflexive closure to the relation of agent $\al$. 
We define
$R^\LG = (R_\al^\LG)_{\al \in \act}$, and $F^\LG = (W,R^\LG)$. 


%

\subsubsection{Modal Logics}
\label{subsec:trans-ml}

We start with translations that map logics 
without recursive operators to logics with fewer constraints.
As mentioned in Subsection \ref{sec:multi_ml_known_results}, all of the logics  $\LG\in \{\klogic, \tlogic, \dlogic, \kf, \df, \sr\}$ and $\LG+5$ with one agent have known complexity results, and the complexity of modal logic is well-studied for multi-agent modal logics as well. 
The missing cases are very few and concern the combination of frame conditions (other than $5$) as well as the multi-agent case. However, we take this 
opportunity to present an 
intuitive introduction to our general translation method.
%
%
Each frame condition that we introduced in Section \ref{sec:multi_ml_background} is associated with an axiom for modal logic, such that whenever a model has the condition, every substitution instance of the axiom is satisfied in all states of the model (see \cite{blackburn2006handbook,MLBlackburnRijkeVenema,Fagin1995ReasoningAboutKnowledge}). For each frame condition $x$ and agent $\al$, we define the axiom $\mathsf{ax}^x_\al$ as follows: 
\begin{multicols}{3}
	\begin{description}
		\item[$\ax^D_\al$: ] $\diam{\al}\true$
		\item[$\ax^T_\al$: ] $[\al]p\impl p$
		\item[$\ax^B_\al$: ] $\diam{\al}[\al]p \impl p $
		\item[$\ax^4_\al$: ] $[\al]p \impl [\al][\al]p$
		\item[$\ax^5_\al$: ] $\diam{\al}[\al]p\impl [\al]p$
	\end{description}
\end{multicols}

For each formula $\varphi$ and $d \geq 0$, let $\Inv_d(\varphi) = \bigwedge_{
	i \leq d}[\act]^i\varphi$, where $[\act]^{0} \varphi = \varphi$ and $[\act]^{i+1} \varphi = [\act][\act]^{i} \varphi$.
Our first translations are defined using the above axioms. Intuitively, for each formula $\varphi$, condition $x$, and set of agents $A$, the formula $\transl{A}{x}{\varphi}$ is satisfiable over frames whose accessibility relations for the agents in $A$ meet some set of conditions that precede $x$ in the above-mentioned fixed linear order if, and only, if $\varphi$ is satisfiable over frames where those accessibility relations also satisfy condition $x$. 

%

\begin{translation}[One-step Translation]\label{translation:ml}
	Let $A\subseteq \act$ and let $x$ be one of the frame conditions.
	For every formula $\varphi$, let $d=md(\varphi)$ 
	if $x\neq 4, 5$, and $d=md(\varphi)|\varphi|$, if $x=4$ or $x=5$. We define:
	$$ \transl{A}{x}{\varphi} = \varphi \wedge \Inv_d\Bigg(\bigwedge_{\substack{\psi \in \subc(\varphi) \\ \al \in A}} \ax^x_\al[
	\psi 
	/ p]\Bigg).
	$$
\end{translation}




\begin{thm}\label{thm:no-fixpoint-translations-work}
	Let $A\subseteq \act$, $x$ be one of the frame conditions, and let $\LG_1, \LG_2$ be logics without recursion operators, such that 
	$$\LG_1(\al) =  \begin{cases}
	\LG_2(\al) + x & \text{ when } \al \in A, \text{ and} \\ 
	\LG_2(\al) & \text{ otherwise. }
	\end{cases}$$
	Assume that $\LG_2(\al)$ only includes frame conditions that precede $x$ in the fixed order of frame conditions.
	Then, $\varphi$ is $\LG_1$-satisfiable if and only if $ \transl{A}{x}{\varphi}$ is $\LG_2$-satisfiable.
\end{thm}

Before proving this theorem, we will define the \emph{unfolding} of a model, which 
will be used both in its proof but also repeatedly through the paper. Our definition is a variation of the one in \cite{GORANKO2007249}.

\begin{defi}[Tree unfolding, \cite{SAHLQVIST1975110,GORANKO2007249}]\label{def:unfolding}
Let $\M =( \proc,\tran,V)$, and $w_0 \in \proc$. We define $\M' =( \proc',\tran',V')$, as follows: 
\begin{itemize}
\item $\proc'$ consists of all finite sequences $w_0 \al_1 w_1 \al_2 \cdots \al_n w_n$ ($n \geq 0$), where $w_i \in \proc$, for every $0 \leq i \leq n$, and $w_i\tran_{\al_{i+1}} w_{i+1}$, for every $0 \leq i \leq n-1$.

\item 
$R_{\al} = \{ ( u , u \al w ) \mid  ( u , u \al w ) \in \proc' \times \proc' \text{ and } w \in \proc \}$.
\item For every propositional variable $p$, $w_0 \al_1 w_1 \al_2 \cdots \al_n w_n$ $\in V'(p)$ iff $w_n \in V(p)$.
\end{itemize}
The pair ($\M',w_0$) is called the \emph{tree unfolding} of $(\M,w_0)$.
\end{defi}

\begin{lemC}[\cite{SAHLQVIST1975110}]\label{lem:ufoldings-are-bisimilar}
The tree unfolding $(\M',w)$ of $(\M,w)$ is bisimilar to $(\M,w)$.
Furthermore, if $\M$ is an $\LG$-model, then the closure of 
$(\M',w)$
under the $\LG$ conditions is bisimilar to $(\M,w)$. 
\changeJ{Moreover, the identity mapping is a bisimulation between $(\M',w)$ and its closure under the $\LG$ conditions.}
\end{lemC}
\begin{proof}
	The first statement follows from \cite{SAHLQVIST1975110} -- but see also \cite{GORANKO2007249}. 
        We observe that one bisimulation in particular is the one that maps each state $v$ of $\M$ to all the paths in the unfolding that end with $v$. We call this bisimulation $\R$.
	For the second statement, simply observe that if $R_{\al}$ is serial in $\M$, then it is also serial in the unfolding. Moreover, any transition that is introduced to satisfy any other closure conditions preserves all bisimulations, and $\R$ in particular, thus also yielding the third statement. 
\end{proof}

We also state the following general observation, for the case of transitive frames, which  will be used in the proof of one of the cases of Theorem \ref{thm:no-fixpoint-translations-work}.

\begin{lem}\label{lem:transitive_box_diamond_inclussion}
Let $\M$ be a Kripke structure. 
For each $v \in W$ and $\al \in \act$, let $b_\al(v) = \{ [\al]\psi  \mid (\M,v) \models [\al]\psi \}$, and $d_\al(v) = \{ \diam{\al}\psi\mid (\M,v) \models \diam{\al}\psi \}$.
Assume $ R_\al$ is transitive and $uR_\al v$. Then:
\begin{enumerate}
\item $b_{\al}(u)
\subseteq 
b_\al(v)$, and
\item $d_{\al}(v)
\subseteq 
d_\al(u)$.
\end{enumerate}
\end{lem}

The proof of the above lemma is straightforward, using the proviso that $R_\alpha$ is transitive.
We are now ready to proceed with the proof of Theorem \ref{thm:no-fixpoint-translations-work}. 
Intuitively, this theorem claims that satisfiability over frames with condition $x$ is reducible to satisfiability over frames without condition $x$---that is,
if one is asking whether there exists some frame with condition $x$ that satisfies a formula $\varphi$, it suffices to translate the formula as we have defined in the statement of that theorem, and verify the satisfiability of the resulting formula, over arbitrary frames. 

The main step in the proof of Theorem \ref{thm:no-fixpoint-translations-work} is to show that, from a model of the translation of a formula $\varphi$, we can construct a model for $\varphi$  that has the property $x$. For each property $x$, we show the correctness of our model construction by proving that  the constructed model satisfies all subformulas of $\varphi$ and thus also $\varphi$.

\begin{proof}[Proof of Theorem \ref{thm:no-fixpoint-translations-work}]
	Assume first that $(\M,w)\models\varphi$, where $\M$ is an $\LG_1$-model (and therefore also an $\LG_2$-model).
	For every subformula $\psi$ of $\varphi$ and $\al \in A$, $\ax^x_\al[\psi/ p]$ is an instantiation of the axiom $\ax^x_\al$, and therefore it holds at all states of $\M$ that are reachable from $w$.
	Thus, $\displaystyle(\M,w) \models \Inv_d\Bigg(\bigwedge_{\substack{\psi \in \subc(\varphi) \\ \al \in A}} \ax^x_\al[\psi/ p]\Bigg)$, which yields that $(\M,w) \models \transl{A}{x}{\varphi}$.
	
	For the other direction, assume that $(\M,w)\models\transl{A}{x}{\varphi}$ 
	for some $\LG_2$-model $\M = (W,R,V)$	
	and state $w$.
	We will use $\M$ to construct an $\LG_1$-model for $\varphi$. 
	For convenience and by Lemma \ref{lem:ufoldings-are-bisimilar} and Theorem \ref{thm:hmt}, we assume we are working with 
    the unfolding of $(\M, w)$ which is defined 
	in Definition \ref{def:unfolding}. In particular, that implies 
    that for every $\al,\beta \in \act$, if $\al \neq \beta$, then 
	$R_\al \cap R_\beta = \emptyset$.
	For each $k \geq 0$, we use the notation $W_k = \{ w_0 \al_1 w_1 \al_2 \cdots \al_r w_r \in W \mid  w_0 = w, \text{ and } 0 \leq r \leq k \}$,
 that is, $W_k$ is the collection of paths from $w \in \mathcal{M}$ of length at most $k$.
 Thus, $W_i \subseteq W_j$, for each $i \leq j$.
%
	Let
	$d=md(\varphi)$ 
	if $x\neq 4$ and $x \neq 5$, and $d=md(\varphi)|\varphi|$, if $x=4$ or $x=5$.
We now distinguish the following cases:
\begin{enumerate}
    \item $x \neq 4, 5$, 
    therefore $\LG_2(\al)$ does not have condition $4$ for any $\al$,
    \item $x=4$, and 
    \item $x = 5$.
\end{enumerate}
We tackle each of these cases separately.

\paragraph{First case: $x \neq 4, 5$ and $\LG_2(\al)$ does not have condition $4$ for any $\al$.}
We define the models: 
\begin{itemize}
	\item $\M^c = (W,R^{\LG_2},V)$;
	\item $\M' = (W',R',V')$, where 
            \begin{itemize}
                \item $W' = W_d$, 
            \item 
            for each $\al \in \act$, $R'_\al = (R^{\LG_2} \cap W_d^2)^D$, if $\LG_2(\al)$ has $D$, and 
            $R'_\al = R^{\LG_2} \cap W_d^2$, otherwise,
            and 
            \item $V'$ is the restriction of $V$ on $W'$; and 
            \end{itemize}
	\item $\M^x = (W',R^{'x,A},V')$ is the $\LG_1$-closure of $\M'$.
\end{itemize}
By Lemma \ref{lem:ufoldings-are-bisimilar}, $(\M,w)$ is bisimilar to $(\M^c,w)$.
We also observe that $\M'$ remains an $\LG_2$-model: removing states only affects condition $D$, and the $D$-closure does not affect the other conditions.
Finally, $\M^x$ is an $\LG_1$-model. 
We will now prove that $\M^x$ satisfies $\varphi$. To this end, we first connect satisfaction in the structures $\M^c$ and $\M'$, and then proceed to do the same for $\M'$ and $\M^x$.

\noindent{\textbf{Claim: For every formula $\psi$ with $md(\psi) \leq d$, and every $v \in W_{d-md(\psi)}$, $(\M^c, v ) \models \psi$ if and only if $(\M',v) \models \psi$.}}
We proceed to prove 
the claim 
by induction on $\psi$. 
All propositional cases are straightforward. 
For $\psi = \diam{\al}\psi'$, 
        we have 
        that $md(\psi)>0$, and therefore $v \in W_{d-1}$.
        By Lemma \ref{lem:ufoldings-are-bisimilar} and Theorem \ref{thm:hmt}, 
        $(\M^c, v ) \models \psi$ if and only if $(\M,v) \models \psi$.
        Therefore, 
        $(\M^c, v ) \models \diam{\al}\psi'$ if and only if $(\M,v) \models \diam{\al}\psi'$ if and only if there is some $u \in W_d$, such that $v R_\al u$ and $(\M, u) \models \psi'$.
        We have that $(\M, u) \models \psi'$ if and only if $(\M^c, u) \models \psi'$ if and only if $(\M', u) \models \psi'$, by the inductive hypothesis.
        Therefore, if  $(\M^c, v ) \models \diam{\al}\psi'$, we have that  there is some $u \in W_{d-md(\psi)+1} = W_{d-md(\psi')}$, such that $v R_\al u$ and $(\M', u) \models \psi'$, yielding that $(\M',v) \models \diam{\al}\psi'$.
        On the other hand, if $(\M',v) \models \diam{\al}\psi'$, then there exists 
        some $u \in W_d$, such that $v R'_\al u$ and $(\M', u) \models \psi'$.
        Since $\LG_2(\al)$ does not have condition $4$, $u \in W_{d-md(\psi)+1} = W_{d-md(\psi')}$, and by the inductive hypothesis $(\M^c, u) \models \psi'$. 
        By the construction of $R'$, we have that $v R^{\LG_2}_\al u$, and therefore $(\M^c, v ) \models \diam{\al}\psi'$.
        The case of $\psi = [\al]\psi'$ is reduced to the case of $\diam{\al}\neg\psi'$, thus completing the proof of the claim.
        %
	Specifically, we have that $(\M',w) \models \varphi$.

	We will now conclude the proof of 
        the first case
        by showing that for every subformula $\psi$ of $\varphi$, and every $v \in W_{d-md(\psi)}$, 
	$(\M', v) \models \psi$ implies that $(\M^x,v )\models \psi$.
From this, it follows that $\M^x$ is an $\LG_1$-model  for $\varphi$, and we are done. 

	We proceed by considering the closure conditions in turn, in each case proving the above claim by induction on the structure of $\psi$.
    Observe that all propositional cases for $\psi$ are straightforward.
    Furthermore, for every $\al \in \act$, $R'_\al \subseteq R^{'x,A}_\al$, and therefore all diamond formulas are preserved from $\M'$ to $\M^x$. 
    Therefore, it suffices to consider the case for $\psi = [\al]\psi'$.
    If $\al \notin A$, then $R'_\al = R^{'x,A}_\al$, and therefore $(\M', v) \models [\al]\psi'$ implies that $(\M^x,v )\models [\al]\psi'$ by the inductive hypothesis. 
    To continue, we assume that $\psi = [\al]\psi'$, where $\al \in A$.
	We observe that $md(\psi)>0$, and therefore $v \in W_{d-1}$.
	Specifically, if $md(\psi)=e$, then $v \in W_{d-e}$.
        From $(\M,w) \models \transl{A}{x}{\varphi}$, we have that $(\M,v) \models \ax_\al^x[\psi'/p]$.
        By the inductive hypothesis, it suffices to prove that 
	if there is some $v {R_\al}^{' x,A} u$, but not $v R_\al' u$, then
        $\M',u \models \psi'$.
	We take cases for $x$:
	\begin{description}
		\item[$x=D$] We observe that 
                        from $(\M,v) \models \ax_\al^x[\psi'/p]$  
                        we get $(\M,v) \models \diam{\al}\true$, which with $md(\diam{\al}\true)=1$,
                        $v \in W_{d-1}$, and the claim above, yield that 
                        $(\M',v) \models \diam{\al}\true$, implying that 
                        $v R'_\al u$ for some $u \in W'$. 
                        Therefore, from Definition \ref{def:closures}, 
                        $v {R_\al}^{' x,A} u$ implies $v R_\al' u$, contradicting our assumptions.
		\item[$x=T$] 
                        By Definition $\ref{def:closures}$, if $v {R_\al}^{' x,A} u$, but not $v R_\al' u$, then $u = v$.
                     We observe that 
                        from $(\M,v) \models \ax_\al^x[\psi'/p]$  
                        we get $(\M,v) \models [\al]\psi' \impl \psi'$.
                        Furthermore, $md([\al]\psi' \impl \psi')= md([\al]\psi') =e$.
                        Therefore, 
                        $v \in W_{d-e}$ and the claim above yields that 
                        $(\M',v) \models [\al]\psi' \impl \psi'$.
                        Therefore, $(\M',v) \models \psi$ yields $(\M',v) \models \psi'$ and this case is complete. 
%
%
\item[$x = B$]
%
%
 By Definition $\ref{def:closures}$, if $v {R_\al}^{' x,A} u$, but not $v R_\al' u$, then $u R_\al' v$ and $u \in W_{d-e-1}$.
                     We observe that 
                        from $(\M,u) \models \ax_\al^x[\psi'/p]$  
                        we get 
                        $(\M,u) \models \diam{\al}[\al]\psi' \impl \psi'$.
                        Furthermore, $md(\diam{\al}[\al]\psi' \impl \psi')= md([\al]\psi') + 1 =e + 1$.
                        The claim above yields that 
                        $(\M',u) \models \diam{\al}[\al]\psi' \impl \psi'$.
                        From $u R_\al' v$ we get that 
                         $(\M',u) \models \diam{\al}\psi$, which then yields together with $(\M',u) \models \diam{\al}[\al]\psi' \impl \psi'$ that 
                        $(\M',u) \models \psi'$ and this case is complete. 
\end{description}

\paragraph{We now consider the case for $x = 4$.}

For this case, we construct the models $\M'$ and $\M^x$ in a different way.
Intuitively, we want to ensure that $\M'$ does not prematurely cut off any long paths in the model that are relevant to the evaluation of the formula.
Therefore, in its construction, we want to detect when certain states behave similarly and we (carefully) eliminate certain branches of the unfolding $\M^c$.

We first observe a statement that is 
analogous
to Lemma \ref{lem:transitive_box_diamond_inclussion}. 
For each $v \in W_k$ and $\al \in A$, let 
\[b^\varphi_\al(v) = \{ [\al]\psi \in \subc(\varphi) \mid (\M^c,v) \models [\al]\psi 
\},\] and 
\[ d^\varphi_\al(v) = \{ \diam{\al}\psi  \in \subc(\varphi) \mid (\M^c,v) \models \diam{\al}\psi 
\}.\]
Assume that 
$v \in W_k$, where $k \leq d$ and $uR_\al v$.  Then, 
$(\M^c,v) \models \bigwedge_{\psi \in \subc(\varphi)} [\al]\psi \impl [\al][\al]\psi$
and therefore:
\begin{enumerate}
\item $b^\varphi_{\al}(u)
\subseteq 
b^\varphi_\al(v)$, and
\item $d^\varphi_{\al}(v)
\subseteq 
d^\varphi_\al(u)$.
\end{enumerate}
If, additionally, $\LG_2(\al)$ has constraint $B$, then 
$b^\varphi_{\al}(u)=
b^\varphi_\al(v)$ and
$d^\varphi_{\al}(v)=
d^\varphi_\al(u)$.

We call a state $sv \in W_{k+1}$ \emph{$\al$-stable}, for some $\al\in A$, when 
$s R_\al^c sv \in W_{k+1}$, 
and
$b^\varphi_\al(s)= b^\varphi_\al(sv)$ and $d^\varphi_\al(s) = d^\varphi_\al(sv)$, and write $s \triangleright_\al sv$.
Observe that
if $u R_\al^c v$ and $v$ is not $\al$-stable, then either $b^\varphi_\al(v)$ is a strict superset of $b^\varphi_\al(u)$ or $d^\varphi_\al(v)$ is a strict subset of $d^\varphi_\al(u)$. Therefore, 
by the Pigeonhole Principle,
 for $\al \in A$, in any sequence $s R_\al sv_1 R_\al \cdots R_\al sv_1 \cdots v_l$, 
 where each $sv_1 \cdots v_i \in W_i$ and $l\geq |\varphi|$, there must be an $\al$-stable state. 

%
%


Recall that we assume that $(\M,w)\models\transl{A}{x}{\varphi}$ 
	for some $\LG_2$-model $\M = (W,R,V)$. 
Let 
$\M' = (W',R',V')$, where 
\begin{itemize}
    \item for each $\al \notin A$, $R'_\al = R^c_\al$, and for each $\al \in A$, $(a,b)\in R'_\al$ if one of the following holds:
    \begin{enumerate}
        \item $a R^c_\al b$ and $\not\exists c \triangleright_\al a$,
        \item $\exists c \triangleright_\al a$, such that $ c R^c_\al b$,
        \item $b \triangleright_\al a$ and  $\LG_2(\al)$ has  $B$, or 
        \item $\exists c \triangleright_\al b$ such that $ c R^c_\al a$ and  $\LG_2(\al)$ has  $B$;
    \end{enumerate}
    \item $W' = \{ v \in W \mid v \text{ is reachable from } w \text{ by } R' \}$; and
    \item $V'$ is the restriction of $V$ on $W'$.
\end{itemize}
Observe that if $a \triangleright_\al b R^c_\al c$, then $c$ is not 
reachable from $w$  by  $R' $.
Therefore, for every $\al \in A$, there is no sequence 
$p R'_\al pv_1 R'_\al \cdots R'_\al pv_1 \cdots v_l$ in $\M'$, where each $pv_1 \cdots v_i \in W_{k_i}$, $k_1 < \cdots < k_l$, and $l > |\varphi|$.
Moreover,
if $\LG_2(\al)$ has  $B$, then $a R^c_\al b$ implies $a \triangleright_\al b$, and therefore item $(4)$ is redundant.

$\M'$ is now a $\LG_2$-model: it is not hard to see that conditions $D$ and $T$ are not affected from $R^c$ to $R'$, and the last two items in the construction of $R'_\al$ above ensure that when $\LG_s(\al)$ has condition $B$, then each new pair in $R'_\al$ also introduces its symmetric one.

We prove,
by lexicographic induction on $(|\psi|,k)$, that 
for every formula $\psi \in \sub(\varphi)$ 
and every $v \in W_k \cap W'$, where $k\leq d-md(\psi)|\varphi|$, 
$(\M^c, v) \models \psi$ if and only if 
$(\M',v) \models \psi$. The propositional cases are straightforward. 
%
%
The 
modal cases $\psi=\diam{\al}\psi'$ or $\psi=[\al]\psi'$ ensure that $md(\psi)>0$, and therefore $v \in W_{d-1}$, and there is some $(v,u)\in R'_\al$.
If $\al \notin A$,  then this case follows from the inductive hypothesis and the observation that the accessible states 
from $v$ in $\M^c$ and in $\M'$ are the same.
Similarly, if $\al \in A$ and $v$ is not $\al$-stable, then $\LG_2(\al)$  does not have constraint $B$ and the accessible states 
from $v$ in $\M^c$ and in $\M^s$ are the same.

If $\al \in A$ and $v'\triangleleft_\al v$ for some $v' \in W$, then
$v'$ is not $\al$-stable, because otherwise $v\notin W'$, as it would not be 
reachable from  $w$  by $R'$ in $\M^c$.
From $v'\triangleleft_\al v$, we get that $(\M^c,v') \models \psi$ if and only if 
$(\M^c,v) \models \psi$, and $v' \in W_{k-1}$.
The inductive hypothesis for $(|\psi|,k-1)$ then yields that 
$(\M',v') \models \psi$ if and only if 
$(\M^c,v') \models \psi$.
Let $\psi=\diam{\al}\psi'$. 
From the above, $(\M^c,v) \models \diam{\al}\psi'$ if and only if $(\M^c,v') \models \diam{\al}\psi'$ if and only if $(\M',v') \models \diam{\al}\psi'$.
But if $(\M',v') \models \diam{\al}\psi'$, then there exists some $v' R'_\al u$, such that $(\M',u) \models \psi'$. But by construction we also have that $v R'_\al u$, and therefore $(\M',v) \models \diam{\al}\psi'$.
On the other hand, if $(\M',v) \models \diam{\al}\psi'$, then 
there exists some $v R'_\al u$, such that $(\M',u) \models \psi'$. By 
construction we also have the following two cases:
\begin{itemize}
    \item $\LG_2(\al)$ has constraint $B$ and $u = v'$, and therefore $(\M',v') \models \psi'$, yielding, by the induction hypothesis  that $(\M^c,v') \models \psi'$. Therefore $(\M^c,v) \models \diam{\al}\psi'$.
    \item $v' R'_\al u$, and therefore $(\M',v') \models \diam{\al}\psi'$, and following the biimplications above, we also have $(\M^c,v) \models \diam{\al}\psi'$.
\end{itemize}


%

Let $\M^x = (W',{R'}^{x,A},V')$.
We now
prove that for every subformula $\psi$ of $\varphi$, and every $v \in W_{d-md(\psi)|\varphi|}$, $(\M', v) \models \psi$ if and only if $(\M^x,v) \models \psi$.
As usual, the propositional cases are straightforward.
If $\psi = \diam{\al}\psi'$ and $(\M', v) \models \psi$, then $(\M^x,v) \models \psi$, because diamond formulas are preserved by the introduction of pairs in the accessibility relation.
If $\psi = [\al]\psi'$, $(\M', v) \models \psi$, and $\al \notin A$, then 
for every $(v,u) \in R'_\al$, $(\M', u) \models \psi'$ and $u \in W_{k+1}$.
By the inductive hypothesis, since $v R'_\al u$ iff $v R'^{x,A}_\al u$, 
for every $(v,u) \in R'^{x,A}_\al$, $(\M^x, u) \models \psi'$, yielding 
$(\M^x, v) \models \psi$.
If $\psi = [\al]\psi'$, $(\M', v) \models \psi$, and $\al \in A$, then 
by the induction hypothesis it suffices to prove that 
$(\M', u) \models \psi'$ for every $R'_\al$-reachable state $u$ from $v$.
Let $v R'_\al v_1 R'_\al \cdots R'_\al v_l = u$.
We know that $v \in W_{d-md(\psi)|\varphi|}$ and that 
$v_i \in W_{d-md(\psi)|\varphi| + |\varphi|} = W_{d-(md(\psi) - 1)|\varphi|}$.
From $(\M,w) \models \transl{A}{x}{\varphi}$, we have that $(\M,v_i) \models \ax_\al^x[\psi'/p]$, yielding $(\M^c,v_i) \models \ax_\al^x[\psi'/p]$.
Therefore, $(\M^c,v_i) \models [\al]\psi' \impl [\al][\al] \psi'$, which gives 
$(\M^c,v_i) \models \psi'$ by straightforward induction on $i$.
Therefore, $(\M^c,u) \models \psi'$, yielding $(\M',u) \models \psi'$, as $u \in  W_{d-(md(\psi) - 1)|\varphi|} =  W_{d-md(\psi')|\varphi|}$.
This completes the case for $\psi = [\al]\psi'$.
The other direction of the biimplication is dual, and the case for $x=4$ is complete.


\paragraph{We now consider the case for $x = 5$.}


We construct $\M^c$ and $\M^x$ as in the above cases, but to construct model $\M'$ we need to take into account the agents with constraint $4$ or $5$ in $A$, in a similar way as it was done for the case for $x = 4$.
Similarly to the case for $x=4$, 
we define  $b^\varphi_\al(v)$ and $d^\varphi_\al(v)$.
Assume that 
$uR_\al v$.  Then, 
Lemma \ref{lem:transitive_box_diamond_inclussion} tells us that 
if $\LG_2(\varphi)$ has constraint $4$, then 
$b^\varphi_{\al}(u)
\subseteq 
b^\varphi_\al(v)$, and
$d^\varphi_{\al}(v)
\subseteq 
d^\varphi_\al(u)$.
If, additionally, $\LG_2(\al)$ has constraint $B$, then 
$b_{\al}(u)=
b_\al(v)$ and
$d_{\al}(v)=
d_\al(u)$.

Observe that for every $\chi \in \subc(\varphi)$, $\al \in A$ and $v_1,v_2 \in W_d$, where $v_1 R_\al^c v_2$,  $(\M,w) \models \transl{A}{x}{\varphi}$ yields that $(\M^c,v_1) \models \diam{\al}\chi$ implies $(\M^c,v_2) \models \diam{\al}\chi$, and $(\M^c,v_2) \models [\al]\chi$ implies $(\M^c,v_1) \models [\al]\chi$.
Therefore, 
 if $\al \in A$, $uR_\al v$, and $u,v \in W_d$, then 
$b^\varphi_{\al}(v)
\subseteq 
b^\varphi_\al(u)$, and
$d^\varphi_{\al}(u)
\subseteq 
d^\varphi_\al(v)$.
If, additionally, $\LG_2(\al)$ has constraint $B$ or $4$, then 
$b^\varphi_{\al}(v)
= 
b^\varphi_\al(u)$, and
$d^\varphi_{\al}(u)
= 
d^\varphi_\al(v)$.

We remind to the reader that, as in the case $x=4$, every sufficiently long sequence of states has an $\al$-stable state.
%
%
%
%
%
%
Let 
$\M' = (W',R',V')$, where 
\begin{itemize}
    \item $W' = \{ v \in W_d \mid v \text{ has no stable state as a strict prefix}\}$;
    \item for each $\al \in \act$, 
    if $\LG_2(\al)$ does not have  constraint $4$ 
    and $\al \notin A$, then 
    $R'_\al = R^c_\al$, 
    \item 
    if $\LG_2(\al)$  has constraint $4$ 
    or $\al \in A$, 
    then  $(a,b)\in R'_\al$ if $a,b \in W'$ and one of the following holds:
    \begin{enumerate}
        \item $a R^c_\al b$ and $\not\exists c \triangleright_\al a$,
        \item $\exists c \triangleright_\al a$, such that $ c R^c_\al b$,
        \item $b \triangleright_\al a$ and  $\LG_2(\al)$ has  $B$, or 
        \item $\exists c \triangleright_\al b$ such that $ c R^c_\al a$ and  $\LG_2(\al)$ has  $B$; 
        and
    \end{enumerate}
     %
    \item $V'$ is the restriction of $V$ on $W'$.
\end{itemize}
Observe that if $a \triangleright_\al b R^c_\al c$, then $c$ is not 
reachable from $w$  by  $R' $.
Similarly to the case of $x=4$, if $\LG(\al)$ has constraint $4$ or $\al \in A$, then there is no sequence 
$p R'_\al pv_1 R'_\al \cdots R'_\al pv_1 \cdots v_l$ in $\M'$, where each $pv_1 \cdots v_i \in W_{k_i}$, $k_1 < \cdots < k_l$, and $l > |\varphi|$.
Moreover,
if $\LG_2(\al)$ has  $B$, then $a R^c_\al b$ implies $a \triangleright_\al b$, and therefore item $(4)$ is redundant.

$\M'$ is now a $\LG_2$-model: it is not hard to see that conditions $D$, $T$, and $4$ are not affected from $R^c$ to $R'$, and the last two items in the construction of $R'_\al$ above ensure that when $\LG_s(\al)$ has condition $B$, then each new pair in $R'_\al$ also introduces its symmetric one.
We can prove that for every $v \in W_k$ and $\psi \in \sub(\varphi)$, if $k \leq d - md(\psi)|\varphi|$, then $(\M^c,v) \models \psi$ if and only if $(\M',v) \models \psi$ with similar arguments as for the case of $x=4$, using the observation that 
if $\LG(\al)$ has constraint $4$ or $\al \in A$, then there is no sequence 
$p R'_\al pv_1 R'_\al \cdots R'_\al pv_1 \cdots v_l$ in $\M'$, where each $pv_1 \cdots v_i \in W_{k_i}$, $k_1 < \cdots < k_l$, and $l > |\varphi|$.

What remains is to 
prove that for every subformula $\psi$ of $\varphi$, and every $v \in W_{d-md(\psi)|\varphi|}$, $(\M', v) \models \psi$ if and only if $(\M^x,v) \models \psi$.
As usual, we use induction on $\psi$, and the propositional cases are straightforward.
If $\psi = \diam{\al}\psi'$ and $(\M', v) \models \psi$, then $(\M^x,v) \models \psi$, because diamond formulas are preserved by the introduction of pairs in the accessibility relation.
If $\psi = [\al]\psi'$, $(\M', v) \models \psi$, and $\al \notin A$, then 
for every $(v,u) \in R'_\al$, $(\M', u) \models \psi'$ and $u \in W_{k+1}$.
By the inductive hypothesis, since $v R'_\al u$ iff $v R'^{x,A}_\al u$, 
for every $(v,u) \in R'^{x,A}_\al$, $(\M^x, u) \models \psi'$, yielding 
$(\M^x, v) \models \psi$.



We now prove
 that  $\psi = [\al]\psi'$, $\al \in A$, and
%
$(\M', v) \models \psi$ implies 
$(\M^x, v) \models \psi$, for every $v \in W_{d-md(\psi)}$. To this end, assume that $(\M', v) \models \psi$ and $u$ is a state such that $v {R'_\al}^x u$, but $v R_\al' u$ does not hold.
It suffices to prove that $(\M^c,u)\models \psi'$.
From our assumption that $v {R'_\al}^x u$, but not $v R_\al' u$, and by the euclidean closure condition, 
we can see that there are some $a, b, c \in W'$, such that $a R'_\al b, c$, and $v$ is 
$ R'_\al$-reachable from $b$, and $u$ is  $R'_\al$-reachable from $c$.

Assume that $(\M^c,u) \not \models \psi'$, to reach a contradiction. 
We can see that there is some 
$u'$ that is $R'_\al$-reachable from 
$a$ and
$u' R'_\al u$, such that
$(\M^c,u') \models \diam{\al} \neg \psi'$;
from the above observation, this yields that $(\M^c,a) \models \diam{\al} \neg \psi'$.
But, using the same observation, $(\M', v) \models [\al]\psi'$ yields $(\M', a) \models [\al]\psi'$, which is a contradiction.
Thus, we conclude that $(\M^c,u) \not \models \psi'$, which is what we wanted to prove.

The other direction of the biimplication is dual, and the proof is complete.
\end{proof}


%

%

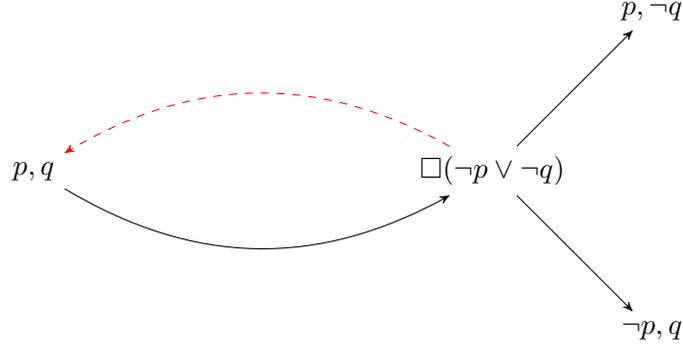
\begin{figure}
\begin{center}
\begin{tikzpicture}[scale=0.6,->,>=stealth',node distance=3cm, auto, main node/.style={draw,font=\large\bfseries}]

\node (1) {$p,q$};
\node (2) [right = 5] {$\Box(\neg p \vee\neg q )$};
\node (3) [above right of=2] {$p,\neg q$};
\node (4) [below right of=2]{$\neg p,q$};

\path 
	(1) edge [bend right](2)
%
    (2) edge [red,dashed,bend right](1)
	(2) edge (3)
		edge (4);
\end{tikzpicture}
\end{center}
\caption{An example where the use of the alternative axiom for symmetricity would yield an incorrect translation.}
\label{fig:notGoodSymetricAxiom}
\end{figure}

\begin{remark}
The attentive reader might have noted that we used the contra-positive version of the axioms for symmetric and euclidean conditions on a frame. Indeed, it turns out that our translation would not be correct if we adopted the usual axioms, namely $ p \rightarrow \Box \Diamond p$ for symmetry and $\Diamond p \rightarrow \Box \Diamond p$ for euclidicity.
To see this, let us assume that there is only one agent and consider the version of a translation defined based on the usual version of the axiom for symmetry, which would map a formula $\varphi$ to
\[
    \varphi \wedge \Inv_d\Bigg(
        \bigwedge_{\substack{\psi \in \subc(\varphi) 
        }} 
            \psi \rightarrow \Box \Diamond \psi 
     \Bigg).
\]
 It is possible in a non-symmetric frame with states $s R t$ 
 to have 
 $t \models \Box \psi$  but $s\not\models  \psi$ for some subformula $\Box \psi$ of $\varphi$.  
 In this case our construction of a symmetric model would not work, as adding symmetric edge of $s R t$ would not preserved the satisfiability of $\Box \psi$ at $t$. 
 For example, Figure \ref{fig:notGoodSymetricAxiom} shows a model where adding edges for the symmetric closure would still satisfy the translation, but it would not preserve the truth of the formula $\Box (\neg p \vee \neg q)$. 
It is then hard to see how such a model would describe a symmetric one, or how that version of the axiom could be used to ``force'' a model of the translation to be easily turned into a symmetric one.
 See Figure \ref{fig:notGoodSymetricAxiom} for such an example.
\end{remark}

%
%
%

Theorem \ref{thm:no-fixpoint-translations-work} together with Ladner's classic result for the complexity of the satisfiability of $\klogic$ \cite{ladnermodcomp} yield an alternative proof of the following statement. 
\begin{cor}\label{cor:modalupper}
	The satisfiability problem for $\LG+x$, with $x \in \tlogic, \dlogic, \blogic, \kf, \kv$, without fixed-point operators
	is in \PSPACE.
\end{cor}



\subsubsection{Modal Logics with Recursion}
\label{sec:transl_mu_cal}

In the remainder of this section we will modify our translations and proof technique, in order to lift our results to 
logics with fixed-point operators.
It is not clear whether the translations of Subsection \ref{subsec:trans-ml} can be extended straightforwardly in the case of logics with recursion, by using unbounded invariance $\Inv$, instead of the bounded  $\Inv_d$. 
To see this, consider the following example. 
\begin{example}\label{ex:all_paths_fin}
	Consider the formula 
$\varphi_f=$ \change{$\mn$}$ X. \Box X$, which requires all paths in the model to be finite, and thus is satisfiable over arbitrary frames, but not over reflexive ones. In the setting of logics with fixed-point operators, formulas may have open subformulas. Therefore, in adapting our translations, it is natural to consider their closures, as defined on page~\pageref{closure-formula}.
The set of the closures of the subformulas of $\varphi_f$ is $\{ \varphi_f, \Box \varphi_f \}$.
The reader may observe that for the case of reflexive frames, we do not need to conjunct over the negations of the subformulas of $\varphi$.
Without those the translation  would give
$$\varphi_t := \varphi_f \land \Inv((\Box \varphi_f \impl \varphi_f) \land (\Box\Box \varphi_f \impl \Box\varphi_f)
).
$$
 \change{Indeed, on reflexive frames,} the formulas $\Box \varphi_f \impl \varphi_f$ and $\Box\Box \varphi_f \impl \Box\varphi_f$ are valid, and therefore $\varphi_t$ is equivalent to $\varphi_f$, which is \klogic-satisfiable.
%
This was not an issue 
in Subsection \ref{subsec:trans-ml}, as
the finiteness of all the paths in a model 
cannot be expressed without recursion.

One would then naturally wonder whether conjoining over $\subc(\varphi_f)$ in the translation
 would make a difference.
The answer is affirmative, as the tranlation 
	$$ \varphi_f \wedge \Inv\Bigg(\bigwedge_{\psi \in \subc(\varphi_f) } 
\Box\psi \impl \psi 
\Bigg)
$$
would then yield a formula that is not satisfiable over reflexive frames.
However, our constructions would not work to prove that such a translation preserves satisfiability.
For example, consider $\mn X.\Box(p \impl(r \land( q \impl X)))$, whose translation is satisfied on a 
 model that satisfies at the same time $p$ and $q$.
We invite the reader to verify the details.
\end{example}

The only case \change{where} the 
approach that we used for the logics without recursion can be applied 
is for 
the case of seriality (condition $D$), as $\Inv(\diam{\al}\true)$ directly ensures the seriality of a model.

\begin{translation}
Let $A\subseteq \act$. We define:
	$$ 
	\transl{A}{D^{\mu}}{\varphi}
	= \varphi  \wedge \Inv\big(\bigwedge_{\al \in A}\diam{\al} \true 
	\big).$$  
\end{translation}

\begin{thm}\label{thm:mu_calc_serial_transl}
	Let $A\subseteq \act$ and $|\act|=k$, 
	and let $\LG$ be a logic, such that 
	$\LG(\al) = \dlogic$ when $\al \in A$, and $\klogic$ otherwise.
	Then, $\varphi$ is $\LG$-satisfiable if and only if 
	$\transl{A}{D^{\mu}}{\varphi}$ is $\klogic_k^\mu$-satisfiable.
\end{thm}

For the cases of reflexivity and transitivity, our simple translations substitute the modal subformulas of a  formula to implicitly enforce the corresponding condition.

\begin{translation}
Let $A\subseteq \act$. The function $\transl{A}{T^{\mu}}{-}$ is defined homomorphically apart from the following two clauses when $\al \in A$.
\begin{itemize}
	\item 
$
\transl{A}{T^{\mu}}{[\al]\varphi}
= \stbox{\al} \transl{A}{T^{\mu}}{\varphi} \wedge \transl{A}{T^{\mu}}{\varphi}$;
\item 
$
\transl{A}{T^{\mu}}{\diam{\al}\varphi}
= \diam{\al} \transl{A}{T^{\mu}}{\varphi} \vee \transl{A}{T^{\mu}}{\varphi}$.
\end{itemize}
\end{translation}

\begin{thm}\label{thm:mu_calc_reflex_translation}
	Let $\LG_1, \LG_2$ be logics, such that 
	$\LG_1(\al) = \LG_2(\al) + T$ when $\al \in A$, and $\LG_1(\al) =\LG_2(\al)$ otherwise, and $\LG_2(\al)$ at most includes frame conditions $D$ or $B$.
	Then, $\varphi$ is $\LG_1$-satisfiable if and only if 
	$\transl{A}{T^{\mu}}{\varphi}$ is $\LG_2$-satisfiable.
\end{thm}

\begin{proof}
	First we assume that $(\M,w)\models \varphi$, 
	where $\M = (W,R,V)$ is an $\LG_1$-model and $w \in W$. 	
	 We can easily verify, by induction on every subformula $\psi$ of $\varphi$ on every environment $\rho$ and state $s$ of $W$, that 
	 $(\M,s)\models_\rho\psi$ if and only if $(\M,s)\models_\rho\transl{A}{T^{\mu}}{\psi}$.
	 We then conclude that
	 $(\M,w)\models\transl{A}{T^{\mu}}{\varphi}$, and thus the translated formula remains satisfiable. 
	 
	 For the converse implication, we assume an $\LG_2$-model $\M= (W,R,V)$ and $w \in W$, such that 
	 $(\M,w)\models\transl{A}{T^{\mu}}{\varphi}$. 
%
%
	We construct an $\LG_1$-model that satisfies ${\varphi}$:
	let $\M_p = (W,R^p,V)$, where for each $\al \in A$, $R^p_\al = R_\al \cup \{(u,u)  \mid u \in W\}$, the reflexive closure of $R_\al$, and for 
		each $\al \notin A$, $R^p_\al = R_\al$. 
By Lemma \ref{lem:conditionsarepreserved}, we have that the symmetry and seriality conditions of $\M$ are preserved in $\M_p$.
It now suffices to prove that for every $v \in W$, environment $\rho$, and $\psi \in \sub(\varphi)$,
$(\M,v) \models_\rho \transl{A}{T^{\mu}}{\psi}$ 
implies 
$(\M_p,v) \models_\rho \psi$.
We proceed by induction on $\psi$.
The cases of propositional and recursion variables and boolean connectives are straightforward.
The cases of fixed-points are also not hard by using the inductive hypothesis.
Finally, 
we treat the modal cases. 
If $\psi = \diam{\al}\psi'$, since we only introduce new pairs to the accessibility relation, it is straightforward to see that $(\M_p,v) \models_\rho \psi$.
If $\psi = [\al]\psi'$ and $(\M,v) \models_\rho \transl{A}{T^{\mu}}{\psi}$, then 
$(\M,v) \models_\rho [\al]\transl{A}{T^{\mu}}{\psi'} \land \transl{A}{T^{\mu}}{\psi'}$, and therefore 
$(\M,v) \models_\rho [\al]\transl{A}{T^{\mu}}{\psi'}$ and $(\M,v) \models_\rho \transl{A}{T^{\mu}}{\psi'}$. 
By the definition of $R^p$, if $v R_\al^p u$ for some $u \in W$, then
$u=v$ or  $v R_\al u$. 
From $(\M,v) \models_\rho [\al]\transl{A}{T^{\mu}}{\psi'}$ and $(\M,v) \models_\rho \transl{A}{T^{\mu}}{\psi'}$ we then get that $(\M,u) \models_\rho \transl{A}{T^{\mu}}{\psi'}$ for both cases, and the inductive hypothesis yields that $(\M_p,u) \models_\rho \transl{A}{T^{\mu}}{\psi'}$, completing the proof.
\end{proof}

For the case of transitivity, we use a restricted form of the invariance operator: for each $\varphi \in L$ and $\al \in \act$, $\Inv^\al(\varphi) = \mx X.(\varphi \land [\al]X)$ and $Eve^\al(\varphi) = \mn X.(\varphi \lor \diam{\al}X)$.

\begin{translation}
Let $A\subseteq \act$.
	The function $\transl{A}{4^{\mu}}{-}$ is defined to be such that
	\begin{itemize}
		\item $\transl{A}{4^{\mu}}{[\al] \psi} = \Inv^\al([\al]\transl{A}{4^{\mu}}{\psi})$ for each $\al \in A$
		\item $\transl{A}{4^{\mu}}{\diam{\al} \psi} = Eve^\al(\diam{\al}\transl{A}{4^{\mu}}{\psi})$ for each $\al \in A$, and
		
		\item $\transl{A}{4^{\mu}}-$ commutes with all other operations.
		
	\end{itemize}
\end{translation}


\begin{thm}\label{thm:mu_calc_transitive_transl}
	Let $A\subseteq \act$, $\LG_1, \LG_2$ be logics, such that 
	$\LG_1(\al) = \LG_2(\al) + 4$ when $\al \in A$, and $\LG_2(\al)$ otherwise, and $\LG_2(\al)$ at most includes frame conditions $D, T, B$.
	Then, $\varphi$ is $\LG_1$-satisfiable if and only if 
	$\transl{A}{4^{\mu}}{\varphi}$ is $\LG_2$-satisfiable.
\end{thm}

\begin{proof}
	If $\transl{A}{4^{\mu}}{\varphi}$ is satisfied in a model $\M = (W,R,V)$, let $\M' = (W,R^+,V)$, where $R^+_\al$ is the transitive closure of $R_\al$, if $\al \in A$, and $R^+_\al = R_\al$, otherwise.
	It is now not hard to use induction on $\psi$ to show that for every (possibly open) subformula $\psi$ of $\varphi$ and for every environment $\rho$,
	$\trueset{\transl{A}{4^{\mu}}{\psi},\rho}_\M = \trueset{\psi,\rho}_{\M'}$.
	The other direction is more straightforward, similarly demonstrating that if $\varphi$ is satisfied in an $\LG_1$-model $\M = (W,R,V)$ (where $R_\al$ is already transitive for $\al \in A$), then for every (possibly open) subformula $\psi$ of $\varphi$ and for every environment $\rho$,
	$\trueset{\transl{A}{4^{\mu}}{\psi},\rho}_\M = \trueset{\psi,\rho}_{\M}$.
\end{proof}

As Example \ref{ex:all_paths_fin} hints, in order to produce a similar translation for symmetric frames, we need to use a more involved construction. Moreover, 
we only prove the correctness of the following translation for formulas without least-fixed-point operators.

\begin{translation}\label{transl:mu_calc_symm}
Let $A\subseteq \act$. Let \change{$p \notin \mycap{PVar}$ be a new propositional variable.} 
The function $\transl{A}{\blogic^{\mu}}{-}$ is defined 
as
 \begin{align*}
     \transl{A}{\blogic^{\mu}}{\varphi} = 
 \varphi ~
 \land ~\bigwedge_{\al \in A} \Inv\big( & \stbox{\al}(\neg p \impl \diam{\al} p) \\
 &\wedge 
 \bigwedge_{\psi \in \subc(\varphi)} (cl(\psi) \rightarrow [\al][\al](p \impl \cl(\psi))) \\
 &\wedge 
 \bigwedge_{\diam{\al}\psi \in \subc(\varphi)} (\diam{\al}cl(\psi) \impl \diam{\al}(\neg p \land \cl(\psi)))
 \big)\change{.}
 \end{align*} 
%
%
\end{translation}
\changeJ{The intuition for this translation is in two parts. In the second conjunct, we force each existing pair $(a,b)$ of an agent's accessibility relation to be given a continuation $(b,c)$ leading to a marked state $c$, which represents $a$, the state from which $b$ is accessible. However, we are not guaranteed that this continuation leads back to the original state -- i.e. that $a=c$, -- and it is not enough to help us construct a truly symmetric model as we would need to do later on. This is why, in the second conjunct, we now require that all subformulas of $\varphi$ will have to be carried through these transitions to be satisfied in all $p$-marked states.
Therefore, $p$-marked states act as accessible clones of a preceding state.
This guarantees that later in the proof we will manage to identify all $p$-marked states with the original one the two transitions originated in (that is, $a$ and $c$ in the above example), thus resulting in a symmetric model. 
The third conjunct can be read in pair with the first one inside the invariant operator, and it ensures its correct behavior. Diamond formulas must be satisfied in state $a$ through a transition to a state $b$ that is not marked with $p$, and therefore, due to the first conjunct, can access a state $c$ that is marked with $p$.}
%

\begin{thm}\label{thm:mu_calc_symm_translation}
	Let $\LG_1, \LG_2$ be logics, such that 
	$\LG_1(\al) = \LG_2(\al) + \blogic$ when $\al \in A$, and $\LG_2(\al)$ otherwise, and $\LG_2(\al)$ at most includes frame condition $D$.
	Then, a  formula $\varphi$ that has no $\mn X$ operators is $\LG_1$-satisfiable if and only if 
	$\transl{A}{\blogic^{\mu}}{\varphi}$ is $\LG_2$-satisfiable.
\end{thm}

\begin{proof}
	First we assume that the \mn-free formula $\varphi$ is satisfied in $(\M,w)$, where $\M = (W,R,V)$ is an $\LG_1$-model and $w \in W$. We assume that $ p \notin V(u)$ for every $ u \in W$. 
	 	We construct an $\LG_2$-model that satisfies $\transl{A}{B^{\mu}}{\varphi}$.
   Let $\M_u = (W_u,R^u,V^u)$ be the unfolding of $(\mathcal{M},w)$, 
   and let $W_p = \{ u_p \mid u \in W_u \}$ be a set of distinct copies of states from $W_u$.
We define \changeJ{$\M_p =  (W_p \cup W_u,R^p,V^p)$}, 
 where:
\begin{itemize}
\item for each $u \in W_u$, $V^p(u) = V(u)$ and $V^p(u_p) = V(u) \cup \{p\}$; and 
\item 
$R_{\al}^p = R_{\al}^u \cup \{(s_p,t), (t,s_p)
\mid (s,t) \in R^u_\al \}$,
if $\al \in A$, and $R_{\al}^p = R_{\al}^u$ otherwise.
\end{itemize}
This construction 
is depicted in
Figure \ref{fig:sym_transl}.

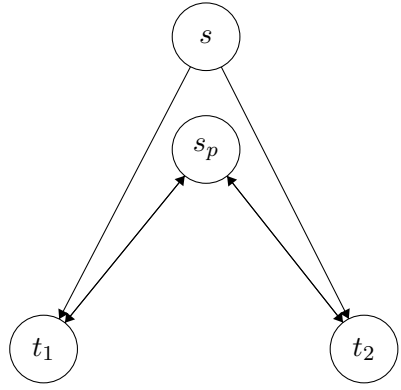
\begin{figure}
\begin{center}
\begin{tikzpicture}[scale=0.15]
\tikzstyle{every node}+=[inner sep=0pt]
\draw [black] (39.1,-5.8) circle (3);
\draw (39.1,-5.8) node {$s$};
\draw [black] (24.7,-33.5) circle (3);
\draw (24.7,-33.5) node {$t_{1}$};
\draw [black] (53.1,-33.5) circle (3);
\draw (53.1,-33.5) node {$t_{2}$};
\draw [black] (39.1,-15.8) circle (3);
\draw (39.1,-15.8) node {$s_p$};
\draw [black] (37.72,-8.46) -- (26.08,-30.84);
\fill [black] (26.08,-30.84) -- (26.9,-30.36) -- (26.01,-29.9);
\draw [black] (40.45,-8.48) -- (51.75,-30.82);
\fill [black] (51.75,-30.82) -- (51.83,-29.88) -- (50.94,-30.33);
\draw [black] (51.24,-31.15) -- (40.96,-18.15);
\fill [black] (40.96,-18.15) -- (41.07,-19.09) -- (41.85,-18.47);
\draw [black] (26.59,-31.17) -- (37.21,-18.13);
\fill [black] (37.21,-18.13) -- (36.31,-18.43) -- (37.09,-19.06);
\draw [black] (40.96,-18.15) -- (51.24,-31.15);
\fill [black] (51.24,-31.15) -- (51.13,-30.21) -- (50.35,-30.83);
\draw [black] (37.21,-18.13) -- (26.59,-31.17);
\fill [black] (26.59,-31.17) -- (27.49,-30.87) -- (26.71,-30.24);
\end{tikzpicture}
\end{center}
\caption{The new model, based on a state $s$ in $\mathcal{M}$, with two neighbours~$t_1,t_2$.}
\label{fig:sym_transl}
\end{figure}
We will now argue that $(\M_p,w) \models \transl{A}{\blogic^{\mu}}{\varphi}$.
To this end, we first 
argue that $(\M_p,w) \models \varphi$. 
To see this, observe that $\R = \{ (s,s), (s,s_p) \mid s \in W_u \}$ is a bisimulation between $(\M_u,w)$ and $(\M_p,w)$, if we omit variable $p$ from the language; Theorem  \ref{thm:hmt} then yields that $(\M_p,w) \models \varphi$. 
More generally, 
let 
$\rho, \rho'$ be environments on $\M_u$ and $\M_p$ respectively, such that 
$\rho(X) = \trueset{\fix(X),\rho}_\M$ and $\rho'(X) = \trueset{\fix(X),\rho'}_{\M_p}$, for every recursion variable $X$ that appears in $\varphi$.
Theorem  \ref{thm:hmt} then yields 
that for every $s \in W_u$ and $\psi \in \sub(\varphi)$, $(\M_u,s) \models_{\rho} \psi$ 
if and only if
$(\M_p,s) \models_{\rho'} \psi$ if and only if $(\M_p,s_p) \models_{\rho'} \psi$.
%

We now prove that for each $\al \in A$,
\[ 
    (\M_p, w) \models \Inv\left( \bigwedge_{\psi \in \subc(\varphi)} cl(\psi) \impl \stbox{\al} \stbox{\al} (p \impl \cl(\psi)) \right) .
\]
It suffices to prove that 
for every $s \in W_u$ and $\psi \in \subc(\varphi)$, 
$(\M_p,s) \models_{\rho'} \psi$ 
implies that 
 $(\M_p,s) \models_{\rho'} [\al][\al] \psi$
 and 
 $(\M_p,s_p) \models_{\rho'} \psi$ 
implies that 
 $(\M_p,s) \models_{\rho'} [\al][\al] \psi$, but the second implication results from the first one.
 Let $(\M_p,s) \models_{\rho'} \psi$ and let $s_1,s_2 \in W_u \cup W_p$, such that $s R_\al^p s_1 R_\al^p s_2$ and $(\M_p,s_2) \models p$. It suffices to prove that  $(\M_p,s_2) \models_{\rho'} \psi$.
 From $(\M_p,s_2) \models p$ we get that $s_2 \in W_p$, and therefore by the construction of $\M_p$, $s_1 \in W_u$.
 Since $\M_u$ is an unfolding of a model, there is a unique $t \in W_u$, such that $t R_\al^p s_1$ and either $s = t$ or $s = t_p$.
 Therefore, $s_2 = t_p$ and $(\M_p,t) \models_{\rho'} \psi$, therefore $(\M_p,s_2) \models_{\rho'} \psi$, which is what we needed to prove.

We now
show that, for each $\al \in A$,
\[ 
    (\M_p, w) \models \Inv\big(\stbox{\al} (\neg p \impl \diam{\al} p)) .
\]
To see this, it suffices to prove that each state $s \in W_u$ that is reachable from $w$ in $\M_p$, 
$(\M_p,s) \models \stbox{\al} (\neg p \impl \diam{\al} p)$ and $(\M_p,s_p) \models \stbox{\al} (\neg p \impl \diam{\al} p)$, where $\alpha \in A$. 
Let $s R_\al^p t$ or $s_p R_\al^p t$.
By the construction of $R^p$, 
if $(\M_p,t) \models \neg p$, then $t \in W_u$, and therefore 
$t R_\al^p s_p$ and $(\M_p,s_p) \models p$.
Therefore, we can conclude that $(\M_p,s) \models \stbox{\al}(\neg p \impl \diam{\al} p)$ and $(\M_p,s_p) \models \stbox{\al} (\neg p \impl \diam{\al} p)$, concluding the proof of the first implication.

Finally, we prove that for each $\al \in A$,
\[ 
    (\M_p, w) \models \Inv\left(\bigwedge_{\diam{\al}\psi \in \subc(\varphi)} (\diam{\al}cl(\psi) \impl \diam{\al}(\neg p \land \cl(\psi)))\right) .
\]
Let $(\M_p,s) \models \diam{\al}cl(\psi) $. Then, there exists some $t$, such that $s R_\al^p t$ and $(\M_p,t) \models \cl(\psi) $. If $s \in W_p$, then by construction, $t \in W_u$ and therefore $(\M_p,t) \models \neg p$, yielding that 
$(\M_p,s) \models \diam{\al}(\neg p \land \cl(\psi)) $.
On the other hand, if $s \in W_u$, then we have that $(\M,s) \models \diam{\al}cl(\psi)$, and therefore there exists some $t' \in W_u$, such that $s R_\al t'$ and $(\M,t') \models \cl(\psi)$. 
As we have observed above, $(\M_u,t')$ and $(\M_p,t')$ are bisimilar, and by the construction of $\M_p$ we have that  $s R_\al^p t'$ and $(\M_p,t') \models cl(\psi)$, yielding that $(\M_p,s) \models \diam{\al}(\neg p \land \cl(\psi)) $. 


For the converse, we assume an $\LG_2(\al)$ model $\M = (W,R,V)$ and $w \in W$, such that 
$(\M,w) \models \transl{A}{B^{\mu}}{\varphi}$, and we construct an $\LG_1$-model $\M'$ that satisfies ${\varphi}$. 
Let $\M' = (W',R',V')$, where 
\begin{itemize}
    \item $W' = W \setminus\{ s \in W \mid (\M,s) \models  p  \text{ and } \exists (t,s) \in R_\al \}$;
    \item $R_\al'$ is the symmetric closure of $R_\al \cap (W')^2$, if $\al \in A$, and  $R'_\al = R_\al  \cap (W')^2$ otherwise;
    \item $V'$ is the restriction of $V$ on $W'$.
\end{itemize}
Let $\rho$ be an environment such that for every $Y$, $\rho(Y)=\trueset{\cl(Y)}$ -- that is, each $X$ is mapped to the interpretation of its recursive closure -- and let $\rho'(Y) = \rho(Y) \cap W'$.
We prove that for every state $v \in W'$ that is reachable from $w$ and $\psi \in \sub(\varphi)$, 
$(\M,v) \models_\rho \psi$ 
implies that $(\M',v)\models_{\rho'} \psi$.
We now proceed by induction on $\psi$. 
\begin{itemize}
	\item 
Propositional cases, the case of logical variables, and the cases of $\psi = [\al]\psi'$ and $\psi = \diam{\al}\psi'$, where $\al \notin A$, are straightforward.
\item 
For the case of $\psi = \diam{\al}\psi'$, 
let $(\M,v) \models_\rho \psi$. 
From $(\M,w) \models \transl{A}{B^{\mu}}{\varphi}$ we get that 
$(\M,v) \models_\rho \diam{\al}(\neg p \land \psi')$, and therefore there exists some $u \in W'$, such that 
$v R_\al u$ and 
$(\M,u) \models_\rho \psi'$.
By the construction of $R'$ and the inductive hypothesis, 
we get that $v R'_\al u$ and 
$(\M',u) \models_{\rho'} \psi'$, yielding that $(\M',v) \models_{\rho'} \psi$.

\item 
For the case of $\psi = [\al]\psi'$, where $\al \in A$, let $(v,v') \in R'_\al$.
It suffices to prove that $(\M',v') \models_{\rho'} \psi'$. If 
$(v,v') \in R_\al$, then we are done by the inductive hypothesis. If not, then $v' R_\al v$, and therefore there is some $(v,v'') \in R_\al$, such that $(\M,v'') \models p$.
Therefore, $(\M,v'' )\models_\rho {\psi'}$. 
From $(\M,w) \models \transl{A}{B^{\mu}}{\varphi}$ we get that 
$(\M,v) \models_\rho \neg \psi' \impl [\al][\al]( p \impl \neg \psi')$.
Therefore, if we assume $(\M,v') \not\models_\rho {\psi'}$, we see that $(\M,v'') \models_\rho \neg \psi'$, which contradicts our observations above.
We conclude that
$(\M,v') \models_\rho {\psi'}$, which by the inductive hypothesis yields $(\M',v') \models_\rho {\psi'}$.
\item For the case of $\psi = \mx Y.\psi'$, notice that 
$\trueset{\psi',\rho}_{\M} 
\subseteq 
\trueset{\psi',\rho'}_{\M'}$ by the inductive hypothesis, and 
$$\trueset{\psi,\rho}_{\M} = \trueset{\psi',\rho[Y \mapsto \trueset{\psi,\rho}_{\M}]}_{\M} = 
\trueset{\psi',\rho}_{\M} 
$$ and $$ 
\trueset{\psi',\rho'}_{\M'} = 
\trueset{\psi',\rho'[Y \mapsto \trueset{\psi,\rho'}_{\M'}]}_{\M} = \trueset{\psi,\rho'}_{\M'},$$ by 
the properties of the fixed points and the definitions of $\rho, \rho'$.
By combining these observations, we get that 
$\trueset{\psi,\rho}_{\M} 
\subseteq 
\trueset{\psi,\rho'}_{\M'}$, which is what we wanted to prove.
\qedhere
\end{itemize}
\end{proof}

Lemma~\ref{lem:conditionsarepreserved} 
assures us that 
closure under some conditions does not affect other conditions.
This observation allows our translations above to remove a frame condition without affecting others, which in turn lets us safely compose our translations.
For example, we can combine Theorems \ref{thm:mu_calc_reflex_translation} and \ref{thm:mu_calc_symm_translation} 
to translate from logics with agents that have both the symmetry and reflexivity condition ($B$ and $T$).
One can first apply the translation of Theorem~\ref{thm:mu_calc_reflex_translation} to reduce the satisfiability problem to one where the reflexivity condition is not in the logic, and then apply the translation of Theorem~\ref{thm:mu_calc_symm_translation} to obtain the following result. 

\begin{cor}
    Let $\LG_1, \LG_2$ be logics, such that 
	$\LG_1(\al) = \LG_2(\al) + B + T$ when $\al \in A$, and $\LG_2(\al)$ otherwise, and $\LG_2(\al)$ at most includes frame condition $D$.
	Then, a  formula $\varphi$ that has no $\mn X$ operators is $\LG_1$-satisfiable if and only if 
	$\transl{A}{\blogic^{\mu}}{\transl{A}{T^{\mu}}{\varphi}}$ is $\LG_2$-satisfiable.
\end{cor}

\begin{remark}\label{remark:finitemodelB5}
	A translation for euclidean frames and for the full logic on symmetric frames would need different approaches.
	 D'Agostino and Lenzi show in \cite{Dagostino2013S5} that $\sv_2^\mu$ does not have a finite-model property, and their result can be easily extended to any logic $\LG$ with fixed-point operators, where there are at least two distinct agents $\al,\alb$, such that $\LG(\al)$ and $\LG(\alb)$ have constraint $B$ or $5$.
	 Therefore, 
	 as our constructions for the translations to $\klogic^\mu_k$ guarantee the finite-model property for the corresponding logics, they do not apply to multimodal logics with $B$ or $5$.
\end{remark}

%% file: translations-lower.tex
\subsection{Embedding \texorpdfstring{$\klogic_n^\mu$}{K\textunderscore n\textasciicircum mu}}
In this subsection, we present translations from logics with fewer frame conditions to ones with more conditions.
These translations will allow us to prove \EXP-completeness results in the following subsection.
	Let
$p,q$ be distinguished propositional variables that do not appear in 
our formulas.
We let $\vec{p}$ range over $p \wedge q$, $p \wedge \neg q$,
and
$\neg p\wedge q$%
.

\begin{defi}[Function $next$]
The function $next$ is defined thus:	$next(p \wedge q) = p \wedge \neg q$, $next(p\wedge \neg q) = \neg p \wedge q$, and $next(\neg p \wedge q) = p \wedge q$. 
\end{defi}


We use a uniform translation from $\klogic^\mu_k$ to any logic with a combination of conditions $D,T,B$. The intuition behind this translation is to describe an unfolded model where all of the unfolded states are iteratively marked with the next value of $\vec{p}$. Then we use the value of $p$ and $q$ in a state to distinguish whether the current state is a duplicate or a ``true'' successor of the original state (a duplicate might be caused due to reflexive or symmetric edges while unfolding). In other words, we want to be able to use modalities to quantify over the transitions in a reflexive or symmetric frame, while excluding any loops introduced by the frame's conditions.

\begin{translation}\label{transl:lower}
Let $A \subseteq \act$.
The function 
$\transl{A}{\klogic^\mu}{-}$
on formulas is defined such that:
\begin{itemize}
	\item $\transl{A}{\klogic^\mu}{\diam{\al} \psi} = \bigwedge_{\vec{p}}(\vec{p} \rightarrow \diam{\al} (next(\vec{p}) \wedge \transl{A}{\klogic^\mu}{\psi}))
	$, if 
	$\al \in A$;
	\item $\transl{A}{\klogic^\mu}{[\al] \psi}=  \bigwedge_{\vec{p}} (\vec{p} \rightarrow [\al] (next(\vec{p}) \rightarrow \transl{A}{\klogic^\mu}{\psi}))
	$, if $\al \in A$;
	\item $\transl{A}{\klogic^\mu}{-}$ commutes with all other operations.
\end{itemize}
\end{translation}

We note that there are simpler translations for the cases of logics with only $D$ or $T$ as a constraint, but the $\transl{A}{\klogic^\mu}{-}$ is uniform for all the logics that we consider in this subsection.

\begin{thm}\label{thm:lower-translations}
	Let $\emptyset \neq A\subseteq \act$, $|\act|=k$, 
	and let $\LG$ be 
	such that 
	$\LG(\al)$ 
	includes only frame conditions from $D,T,B$
	when $\al \in A$, and $\LG(\al)=\klogic$ otherwise.
	Then, $\varphi$ is $\klogic_k^{\mu}$-satisfiable if and only if 
	$\transl{A}{\klogic^{\mu}}{\varphi}$
	is $\LG$-satisfiable.
\end{thm}
\begin{proof}
For the ``only if'' direction, let $\M=(W,R,V)$ be an unfolded model and $w \in W$ be its root, such that $(\M,w) \models \varphi$. Variables $p$ and $q$ do not appear in $\varphi$,
so we can assume that $(\M,w) \models p \land q$, and that, for each $v R_\al v'$ with  $\al \in A$ and each $\vec{p}$, 
 if $(\M,v) \models \vec{p}$ then $(\M,v') \models next(\vec{p})$.
Let $\M'=(W,R',V)$, where $R'$ is the closure of $R$ under the conditions of $\LG$.
We observe that for every $\al \in A$, 
$R_\al = R_\al' \cap \bigcup_{\vec{p}} (\trueset{\vec{p}}\times \trueset{next(\vec{p})})$.
We prove that for every environment $\rho$, every $\psi \in \sub(\varphi)$, and every $v \in W$, $(\M,v) \models_\rho \psi$ iff $(\M',v) \models_\rho \transl{A}{\klogic^\mu}{\psi}$.
The proof proceeds by induction on $\psi$. We fix a $v \in W$ and a $\vec{p}$, such that $(\M,v) \models \vec{p}$.
\begin{itemize}
	\item The propositional cases and the case of $\psi = X$ are immediate, and so
	 are the cases of $\psi = \diam{\al}\psi'$ and $\psi = [\al]\psi'$, where $\al \notin A$.
	\item For the case of $\psi = \mx X.\psi'$ or $\psi = \mn X.\psi'$, note that due to the inductive hypothesis, for any $S \subseteq W$, $\trueset{\psi',\rho[X\mapsto S]}_\M = \trueset{\transl{A}{\klogic^\mu}{\psi'},\rho[X\mapsto S]}_{\M'}$, and therefore $\trueset{\psi,\rho}_\M = \trueset{\transl{A}{\klogic^\mu}{\psi},\rho}_{\M'}$.
	\item For the case of $\psi = \diam{\al}\psi'$, where $\al \in A$, we reason as follows:  
	$(\M,v) \models_\rho \psi$ iff $(\M,v) \models_\rho \diam{\al} \psi'$
	iff there is some $(v,v') \in R_\al$, such that $(\M,v') \models_\rho \psi'$ iff
	there is some $(v,v') \in R'_\al$, where $(\M,v') \models next(\vec{p})$, $(\M,v') \models_\rho \psi'$  iff there is some $(v,v') \in R'_\al$, such that $(\M,v') \models next(\vec{p}) \wedge \psi'$ iff 
	$$(\M,v') \models \transl{A}{\klogic^\mu}{\diam{\al} \psi'} = 
	\bigwedge_{\vec{p}}
	(\vec{p} \rightarrow \diam{\al} 
		(next(\vec{p}) \wedge \transl{A}{\klogic^\mu}{\psi'})
	).$$
	\item For the case of $\psi = [\al]\psi'$, where $\al \in A$, 
	$(\M,v) \models_\rho \psi$ iff $(\M,v) \models_\rho [\al] \psi'$
	iff for every $(v,v') \in R_\al$, $(\M,v') \models_\rho \psi'$ iff
	for every $(v,v') \in R'_\al$, where $(\M,v') \models next(\vec{p})$, $(\M,v') \models_\rho \psi'$ iff for every $(v,v') \in R'_\al$, $(\M,v') \models next(\vec{p}) \impl \psi'$ iff 
	$$(\M,v') \models \transl{A}{\klogic^\mu}{[\al] \psi'} = 
	\bigwedge_{\vec{p}}
	(\vec{p} \rightarrow [\al] 
	(next(\vec{p}) \impl \transl{A}{\klogic^\mu}{\psi'})
	),$$ 
	which completes the proof by induction.
\end{itemize}

For the ``if'' direction, 
let $\M=(W,R,V)$ be an  $\LG$-model and $w \in W$, 
be such that $(\M,w) \models \transl{A}{\klogic^\mu}{\varphi}$.
Let $\M'=(W,R',V)$, where for every $\al \in A$, 
$R'_\al = R_\al \cap \bigcup_{\vec{p}} (\trueset{\vec{p}}\times \trueset{next(\vec{p})})$, and for every $\al \notin A$, $R_\al' = R_\al$.
We can prove that for any environment $\rho$, any $\psi \in \sub(\varphi)$, and any $v \in W$, $(\M,v) \models_\rho \transl{A}{\klogic^\mu}{\psi}$ iff $(\M',v) \models_\rho \psi$.
The proof proceeds by induction on $\psi$ and is very similar to the one for the ``only if'' direction. 
\end{proof}

%% file: translations-complexity.tex
\subsection{Complexity Results}

	\change{We observe that, with the exception of Translations \ref{translation:ml} and \ref{transl:mu_calc_symm}, our translations yield formulas of size at most linear with respect to the input formula. Translations \ref{translation:ml} and \ref{transl:mu_calc_symm} have a quadratic cost.}

\begin{cor}\label{cor:mu-upper1}
		If  $\LG$ only has frame conditions $D,T$, then its satisfiability problem is \EXP-complete;
		if $\LG$ only has frame conditions $D,T,4$, then its satisfiability problem is in \EXP.
\end{cor}

\begin{proof}
	Immediately from Theorems \ref{thm:mu_calc_serial_transl}, \ref{thm:mu_calc_reflex_translation}, \ref{thm:mu_calc_transitive_transl}, and \ref{thm:lower-translations}.
\end{proof}

\begin{cor}\label{cor:mu-upper2}
	If $\LG$ only has frame conditions $D,T,B$, then 
	\begin{enumerate}
		\item 
		$\LG$-satisfiability is \EXP-hard;  and 
		\item 
		the restriction of $\LG$-satisfiability on formulas without $\mn X$operators is \EXP-complete.
	\end{enumerate}
\end{cor}
\begin{proof}
Immediately from Theorems \ref{thm:mu_calc_serial_transl}, \ref{thm:mu_calc_reflex_translation}, \ref{thm:mu_calc_symm_translation}, and \ref{thm:lower-translations}.
\end{proof}

\begin{table}
	\centering
\begin{tabular}{|l|l|l|l|}
\hline
$\#$ agents      & Restrictions on syntax/frames & \multicolumn{1}{l|}{Upper Bound} & Lower Bound \\ \hline
                      & frames with B or 5 & ?                                & \EXP-hard  \ref{prp:EXP-hard-k2}  \\
$\geq 2$ & not B , 5                                  & \EXP~ \ref{cor:mu-upper1} & \EXP-hard \ref{prp:EXP-hard-k2}   \\
                      & not 5, not $\mu$.X                             & \EXP~ \ref{cor:mu-upper2} & \EXP-hard   \ref{prp:fragmantsMu} \\ \hline
                    & with 5 (or B4)             & \NP~ \ref{thm:ladhalp} & \NP-hard   \ref{thm:ladhalp}  \\
       1               & with 4                                      & \EXP~\ref{cor:mu-upper1} & \PSPACE-hard \ref{thm:ladhalp} \\
                      & Any other restrictions                       & \EXP~  \ref{prp:muCalc-sat} & \EXP-hard  \ref{prp:muCalc-sat}  \\ \hline
\end{tabular}
\caption{Summary of the complexity of satisfiability checking for various modal logics with recursion up to section \ref{sec:translations}. 
A logic has the complexity of the first line from the top of the table that contains its description. The respective theorems from either this work or already known are linked. }
\label{tab:summary}
\end{table}

A summary of our results can be seen in Table \ref{tab:summary}. The only mentioned result that was already known is that  the satisfiability problem for the $\mu$-calculus over one agent is \EXP-complete.

We can clearly see from the summary in Table \ref{tab:summary} that the bound we have obtained via the translation-based methods presented in this section is not tight for the satisfiability problem for the single-agent $\mu$-calculus over transitive frames. 
Specifically, the translations established already only guarantee that the problem is in \EXP, but the only available lower bound is \PSPACE-hardness from the corresponding modal logic (without recursion). 
Thus, to have a tight bound, we would need to prove that the satisfiability checking for $\kf^\mu$ is either in \PSPACE, or that it is \EXP-hard. 
We do not have an answer to this question that relies on translations, but we will present some musings on this topic in the discussion section (Sec. \ref{sec:multi_ml_conclusion}).


In what follows, instead of applying the translation method to obtain a tight complexity bound for satisfiability for the single-agent $\mu$-calculus over transitive frames, we will utilize the tableaux we will present in Section  \ref{sec:multi} to produce a \PSPACE~algorithm, showing that the problem is, in fact, PSPACE-complete.


%% file: tableaux-multi.tex

We give a sound and complete tableau system for logic $\LG$.
Furthermore, if $\LG$ has a finite-model property, then we give terminating conditions for its tableau. 
The system that we give in this section is based on Kozen's tableaux for the $\mu$-calculus \cite{Kozen1983} and the tableaux of Fitting \cite{fitting1972tableau} and Massacci \cite{Massacci1994} for \ML.
We can use Kozen's finite model theorem \cite{Kozen1983} to help us ensure
the termination of the tableau for some of these logics.

\begin{thmC}[\cite{Kozen1983}]
	\label{thm:smallKmodel}
	There is a computable $\kappa:\nat \to \nat$, such that every $\logicize{K}_k^\mu$-satisfiable formula $\varphi$ is satisfied in a model with at most $\kappa(|\varphi|)$ states.\footnote{The tableau in 
		\cite{Kozen1983} yields an upper bound of 
		$2^{2^{O(n^3)}}$ for $\kappa_0(n)$, 
		but that bound is not useful to 
		obtain
		a ``good'' decision procedure. The purpose of this section is not to establish any good upper bound for satisfiability testing, which is done in Sections \ref{sec:KfourMu}  and \ref{sec:TtT}.}
\end{thmC}

\begin{cor}\label{cor:smallLmodel}
	If $\LG$ only has frame conditions $D, T, 4$, then 
	there is a computable $\kappa:\nat \to \nat$, such that  every $\LG$-satisfiable formula $\varphi$ is satisfied in a model with at most $\kappa(|\varphi|)$ states.
\end{cor}
\begin{proof}
	Immediately, from Theorems \ref{thm:smallKmodel}, \ref{thm:mu_calc_serial_transl}, \ref{thm:mu_calc_reflex_translation}, and \ref{thm:mu_calc_transitive_transl}, and Lemma \ref{lem:conditionsarepreserved}.
\end{proof}

\begin{remark}
	We note that not all modal logics with recursion have a finite-model property --- see Remark \ref{remark:finitemodelB5}.
\end{remark}

\change{Intuitively, a tableau attempts to build a model that satisfies the given formula. 
When it needs to consider two possible cases, it branches, and thus it may generate several branches. Each branch that satisfies certain consistency conditions, which we define below, represents a corresponding model.}

Our tableaux use \emph{prefixed formulas}, that is, formulas of the form $\sigma~\varphi$, where $\sigma \in (\act\times L)^*$ and $\varphi \in L$; $\sigma$ is the prefix of $\varphi$ in that case, and we say that $\varphi$ is prefixed by $\sigma$.
We note that we separate the elements of $\sigma$ with a dot.
\changeJ{
Furthermore, in the tableau prefixes, we write $\al\diam{\psi}$ to mean the pair $(\al,\psi)\in \act \times L$.
Therefore, for example, $(\al,\phi)(\beta,\chi)(\al,\psi)$ is written $\al\diam{\phi}.\beta\diam{\chi}.\al\diam{\psi}$ as a tableau prefix.}
We say that the prefix $\sigma$ is $\al$-flat when agent $\al$ has axiom $5$ and $\sigma = \sigma'\nStt{\al}{\psi}$ for some $\psi$.
Each prefix possibly represents a state in a corresponding model, and a prefixed formula $\sigma~\varphi$ declares that $\varphi$ is satisfied in the state represented by $\sigma$. As we will see below, the prefixes from $(\act\times L)^*$ allow us to keep track of the diamond formula that generates a prefix through the tableau rules. For 
	agents with condition $5$, this allows us to restrict the generation of new prefixes and avoid certain redundancies, due to the similarity of euclidean binary relations to equivalence relations \cite{nagle_thomason_1985,Halpern2007Characterizing}. For single-agent logics with condition $4$, this notation will further let us avoid redundancies and achieve a more efficient decision procedure --- see Section \ref{sec:KfourMu}.

\begin{table}[t]
	{\centering 
\begin{gather*}
		\AxiomC{$\sigma~\pi X.\varphi$}
		\RightLabel{(\textsf{fix})}
		\UnaryInfC{$\sigma~\varphi$}
		\DisplayProof \quad
		\AxiomC{$\sigma~X$}
		\RightLabel{(\textsf{X})}
		\UnaryInfC{$\sigma~\fix(X)$}
		\DisplayProof 
		\quad
		\AxiomC{$\sigma~\varphi \lor \psi$}
		\RightLabel{(\textsf{or})}
		\UnaryInfC{$\sigma~\varphi\mid\sigma~\psi$}
		\DisplayProof \quad		
		\AxiomC{$\sigma~\varphi \land \psi$}
		\RightLabel{(\textsf{and})}
		\UnaryInfC{$\sigma~\varphi$}
		\noLine
		\UnaryInfC{$\sigma~\psi$}
		\DisplayProof 
		\\[4ex]
	\AxiomC{$\sigma~[\al]\varphi$}
	\RightLabel{(\textsf{B})}
	\UnaryInfC{$\sigma\nStt{\al}{\psi}~\varphi$}
	\DisplayProof \quad
	\AxiomC{$\sigma~\diam{\al}\varphi$}
	\RightLabel{(\textsf{D})}
	\UnaryInfC{$\sigma\nStt{\al}{\varphi}~\varphi$}
	\DisplayProof~ \quad
	\AxiomC{$\sigma~[\al]\varphi$}
	\RightLabel{(\textsf{d})}
	\UnaryInfC{$\sigma\nStt{\al}{\varphi}~\varphi$}
	\DisplayProof~\quad
	\AxiomC{$\sigma~[\al]\varphi$}
	\RightLabel{(\textsf{4})}
	\UnaryInfC{$\sigma\nStt{\al}{\psi}~[\al]\varphi$}
	\DisplayProof
	\\
\end{gather*}
\\[2ex]
\noindent
where, for rules (\textsf{B}) and (\textsf{4}), 
\change{$\sigma\nStt{\al}{\psi}$} has already appeared in the branch;
and for (\textsf{D}), $\sigma$ is not $\al$-flat.
\noindent
\begin{gather*}
	\AxiomC{$\sigma\nStt{\al}{\psi}~[\al]\varphi$}
	\RightLabel{(\textsf{B5})}
	\UnaryInfC{$\sigma~[\al]\varphi$}
\DisplayProof \quad
\AxiomC{$\sigma\nStt{\al}{\psi}~\diam{\al}\varphi$}
\RightLabel{(\textsf{D5})}
\UnaryInfC{$\sigma\nStt{\al}{\psi}\nStt{\al}{\varphi}~\varphi$}
\DisplayProof \quad
\AxiomC{$\sigma\nStt{\al}{\psi}~[\al]\varphi$}
\RightLabel{(\textsf{b})}
\UnaryInfC{$\sigma~\varphi$}
\DisplayProof \quad
\AxiomC{$\sigma~[\al]\varphi$}
\RightLabel{(\textsf{t})}
\UnaryInfC{$\sigma~\varphi$}
\DisplayProof
\\[4ex]
\AxiomC{$\sigma\nStt{\al}{\psi}~[\al]\varphi$}
\RightLabel{(\textsf{B55})}
\UnaryInfC{$\sigma\nStt{\al}{\psi'}~[\al]\varphi$}
\DisplayProof \quad
\AxiomC{$\sigma\nStt{\al}{\psi}\nStt{\al}{\psi'}~\diam{\al}\varphi$}
\RightLabel{(\textsf{D55})}
\UnaryInfC{$\sigma\nStt{\al}{\psi}\nStt{\al}{\varphi}~\varphi$}
\DisplayProof  \quad
\AxiomC{$\sigma\nStt{\al}{\psi}~[\al]\varphi$}
\RightLabel{(\textsf{b4})}
\UnaryInfC{$\sigma~[\al]\varphi$}
\DisplayProof
\\
\end{gather*}
\\[2ex]
\noindent }
 \begin{flushleft} 
where, 
for rule (\textsf{B55}),  
$\sigma\nStt{\al}{\psi'}$ has already appeared in the branch;
for rule (\textsf{D5}),  $\sigma$ is not $\al$-flat, and 
$\sigma~\diam{\al}\varphi$ does not appear in the branch;
for rule (\textsf{D55}),  
$\sigma~\diam{\al}\varphi$ does not appear in the branch.
%
\end{flushleft}
\caption{The tableau rules for $\logicize{L}=\logicize{L}^\mu_n$.}
\label{tab:tableau}
\end{table}

The tableau rules that we use appear in Table \ref{tab:tableau}. These include fixed-point and propositional rules, as well as rules that deal with modalities. Depending on the logic that each agent $\al$ is based on, a different set of rules applies for $\al$: for rule (\textsf{d}), $\logicize{L}(\al)$ must have condition $D$; for rule (\textsf{t}), $\logicize{L}(\al)$ must have condition $T$; for rule (\textsf{4}), $\logicize{L}(\al)$ must have condition $4$; for rules (\textsf{B5}), (\textsf{D5}), and (\textsf{D55}), $\logicize{L}(\al)$ must have condition $5$; for rule (\textsf{b}) $\logicize{L}(\al)$ must have condition $B$; and for rule (\textsf{b4}) $\logicize{L}(\al)$ must have both $B$ and $4$.
Rule (\textsf{or}) is the only rule that splits the current tableau branch into two.
\changeJ{
When an agent $\al$ does not have condition $5$, diamond formulas such as $\sigma~\diam{\al}\psi$ are analysed by extending the current prefix $\sigma$ to $\sigma.\al\diam{\psi}$, using rule $(\textsf{D})$. 
But when $\al$ has condition $5$, we need to take into account that due to the Euclidean condition, formulas prefixed with $\sigma$ may affect prefixes that neither extend nor are prefixes of $\sigma$, in a similar way as for states in a Kripke model. As Lemma \ref{lem:NImodels} demonstrates later on in this section, condition $5$ results in an accessibility relation that is very close to an equivalence relation, and our tableau rules for agents with condition $5$ (\emph{i.e.} rules (\textsf{B5}), (\textsf{B55}), (\textsf{D5}), and (\textsf{D55})) reflect this.}

A tableau branch is propositionally closed when $\sigma~\false$ or both $\sigma~p$ and $\sigma~\neg p$ appear in the branch for some prefix $\sigma$.
For each prefix $\sigma$ that appears in a \change{fixed} tableau \change{branch}, let $\form(\sigma)$ be the set of formulas prefixed by $\sigma$ \change{in that branch}. 
We use the notation $\sigma \prec \sigma'$ to mean that $\sigma'=\sigma.\sigma''$ for some $\sigma''$, 
in which case
$\sigma$ is an ancestor of $\sigma'$.
We denote with $\sigma(b)$ the set of prefixes that appear in a branch $b$.

We define the \changeJ{dependence} relation $\xrightarrow{X}$ on prefixed formulas in a tableau \change{branch} as $\chi_1 \xrightarrow{X} \chi_2$, if 
$\chi_2$ was introduced to the branch by a tableau rule
with $\chi_1$ as its premise, 
and $\chi_1$ is not of the form $\sigma~Y$, where $X<Y$; then, $\xrightarrow{X}^+$ is the transitive closure of $\xrightarrow{X}$ and $\xrightarrow{X}^*$ is its reflexive and transitive closure.
We can also extend this relation to prefixes, so that $\sigma \xrightarrow{X} \sigma'$, if and only if $\sigma~\psi \xrightarrow{X} \sigma'~\psi'$, for some $\psi\in\form(\sigma)$ and $\psi'\in\form(\sigma')$.
If in a branch
there is a
$\xrightarrow{X}$-sequence where $X$ is a least fixed-point and appears infinitely often,
then the branch is called fixed-point-closed.
A branch is closed when it is either fixed-point-closed or propositionally closed; if it is not closed, then it is called open.

Now, assume that there is a 
$\kappa:\nat \to \nat$, such that  every $\LG$-satisfiable formula $\varphi$ is satisfied in a model with at most $\kappa(|\varphi|)$ states.
An open tableau branch is called (\emph{resp. locally}) \emph{maximal} when 
all tableau rules (\emph{resp. the tableau rules that do not produce new prefixes}) have been applied.
A branch is called \emph{sufficient} for $\varphi$ when it is locally maximal and for every $\sigma~\psi$ in the branch, for which a rule can be applied and has not been applied to $\sigma~\psi$, 
$|\sigma| > |\act|\cdot\kappa(|\varphi|)^{|\varphi|^2}\cdot 2^{|\varphi|+1}$.
%
A tableau is called maximal when all of its open branches are maximal, and closed when all of its branches are closed.
It is called sufficiently closed for $\varphi$ if it is propositionally closed, or for some least fixed-point variable $X$, it 
has a $\xrightarrow{X}$-path, where $X$ appears at least $\kappa(|\varphi|)+1$ times.
%
A sufficient branch for $\varphi$ that is not sufficiently closed is called sufficiently open for $\varphi$.

A tableau for $\varphi$ starts from $\varepsilon~\varphi$ and is built using the tableau rules of Table \ref{tab:tableau}.
A tableau proof for $\varphi$ is a closed tableau for the negation of $\varphi$.

The following example demonstrates the taubleau system in action. 

\begin{example}
	\change{Let $\act = \{a,b\}$ and $\LG$ be a logic, such that 
		$\LG(a) = \klogic^\mu$ 
		and $\LG(b) = \kv^\mu$. 
		Let
	\begin{align*}
		\varphi_1 &= (p \land \diam{a}p) \land \mn X.(\neg p \vee [a]X)& &and& &\varphi_2 = \diam{b}p \land \mn X.([b]\neg p \vee [b]X).
	\end{align*}
	As we see in Figure \ref{tab:tableauex}, the tableau for $\varphi_1$ produces an open branch, while the one for $\varphi_2$ has all of its branches closed, the leftmost one due to an infinite $\xrightarrow{X}$-sequence.
}
	
	\alwaysRootAtTop    
	\def\defaultHypSeparation{\hskip .1in}
		

\begin{figure}
\begin{subfigure}[b]{0.4\textwidth}
\scalebox{0.90}{
$	\AxiomC{$a\diam{p}~[a]X$}
	\AxiomC{\closed}
	\noLine
	\UnaryInfC{$a\diam{p}~\neg p$}
	\BinaryInfC{$a\diam{p}~\neg p \vee [a]X$}\RightLabel{\scriptsize(\textsf{fix})}
	\UnaryInfC{$a\diam{p}~\mn X.(\neg p \vee [a]X)$}\RightLabel{\scriptsize(\textsf{X})}
	\UnaryInfC{$a\diam{p}~X$}\RightLabel{\scriptsize(\textsf{B})}
	\UnaryInfC{$a\diam{p}~p$}\RightLabel{\scriptsize(\textsf{D})}
	\UnaryInfC{$\varepsilon~[a]X$}
	\AxiomC{\closed}
	\noLine
	\UnaryInfC{$\varepsilon~\neg p$}
	\BinaryInfC{$\varepsilon~\neg p \vee [a]X$}\RightLabel{\scriptsize(\textsf{fix})}
	\UnaryInfC{$\varepsilon~\diam{a}p$}
	\noLine 
	\UnaryInfC{$\varepsilon~p$}
	\UnaryInfC{$\varepsilon~p \land \diam{a}p$}
	\noLine 
	\UnaryInfC{$\varepsilon~\mn X.(\neg p \vee [a]X)$}
	\UnaryInfC{$\varepsilon~(p \land \diam{a}p) \land \mn X.(\neg p \vee [a]X)$}
	\DisplayProof $ }
\end{subfigure}
\begin{subfigure}[b]{0.55\textwidth}
\scalebox{0.9}{
$ \AxiomC{\vdots}\RightLabel{\scriptsize(\textsf{X})}
	\UnaryInfC{$b\diam{p}~X$}\RightLabel{\scriptsize(\textsf{B})}
	\UnaryInfC{$\varepsilon~[b]X$}\RightLabel{\scriptsize(\textsf{B5})}
	\UnaryInfC{$b\diam{p}~[b]X$} 
	\AxiomC{\closed}\noLine 
	\UnaryInfC{$b\diam{p}~\neg p$}\RightLabel{\scriptsize(\textsf{B})}
	\UnaryInfC{$\varepsilon~[b]\neg p$}\RightLabel{\scriptsize(\textsf{B5})}
	\UnaryInfC{$b\diam{p}~[b]\neg p$} 
	\BinaryInfC{$b\diam{p}~[b]\neg p \vee [b]X$}\RightLabel{\scriptsize(\textsf{fix})}
	\UnaryInfC{$b\diam{p}~\mn X.([b]\neg p \vee [b]X)$}\RightLabel{\scriptsize(\textsf{X})}
	\UnaryInfC{$b\diam{p}~X$}\RightLabel{\scriptsize(\textsf{B})}
	\UnaryInfC{$\varepsilon~[b]X$}
	\AxiomC{\closed}
	\noLine
	\UnaryInfC{$b\diam{p}~\neg p$}\RightLabel{\scriptsize(\textsf{B})}
	\UnaryInfC{$\varepsilon~[b]\neg p$}
	\BinaryInfC{$\varepsilon~[b]\neg p \vee [b]X$}\RightLabel{\scriptsize(\textsf{fix})}
	\UnaryInfC{$b\diam{p}~p$}\RightLabel{\scriptsize(\textsf{D})}
	\UnaryInfC{$\varepsilon~\diam{b}p$}
	\noLine 
	\UnaryInfC{$\varepsilon~\mn X.([b]\neg p \vee [b]X)$}
	\UnaryInfC{$\varepsilon~\diam{b}p \land \mn X.([b]\neg p \vee [b]X)$}
	\DisplayProof $ }
\end{subfigure}
\caption{\change{Tableaux for $\varphi_1$ and $\varphi_2$. The dots represent that the tableau keeps repeating as from the identical node above. The \texttt{x} mark represents a propositionally closed branch.}}
\label{tab:tableauex}
\end{figure}
\end{example}

\begin{thm}[Soundness, Completeness, and Termination of $\LG_k^\mu$-Tableaux]
	\label{thm:tableaux}
 The first two of the following three statements are equivalent for any formula $\varphi\in L$ and any logic $\LG$. Furthermore, if 
	there is a 
	$\kappa:\nat \to \nat$, such that  every $\LG$-satisfiable formula $\varphi$ is satisfied in a model with at most $\kappa(|\varphi|)$ states, then all the following \changeJ{statements} are equivalent.
	\begin{enumerate}
	\itemsep0em 
		\item $\varphi$ has a maximal $\logicize{L}$-tableau with an open branch;
		\item $\varphi$ is $\logicize{L}$-satisfiable; and
		\item $\varphi$ has an $\logicize{L}$-tableau with a sufficiently open branch for $\varphi$.
	\end{enumerate}
\end{thm}

Here we state and prove certain definitions and lemmata that will be used in the proof of Theorem \ref{thm:tableaux}.
For an agent $\al \in \act$ and a state $s$ in a model $(W,R,V)$, let $Reach_\al(s)$ be the states that are reachable in $W$ by $R_\al$.
We also use $R|_S$ for the restriction of a relation $R$ on a set $S$.

For a logic $\LG$,we say that a state $s$ in a model $\M = (W,R,V)$
is $\LG$ \emph{flat} when, for every $\al \in \act$ for which $\LG(\al)$ has constraint $5$, there is a set of states $W_0$, such that:
\begin{itemize}
\itemsep0em 
	\item 
	$Reach_\al(s) = \{s\} \cup W_0$;
	\item
	$R_\al|_{Reach_\al(s)} = E_0 \cup E_1$, where 
			\subitem 
	$E_0 \subseteq \{s\} \times W_0$ and 
			\subitem 
	$E_1 = W_0^2$; and 
	\item if $\LG(\al)$ has constraint $T$, or $E_0\neq \emptyset$ and $\LG(\al)$ has constraint $B$, then $s \in W_0$.
\end{itemize}
%
\begin{lemC}[\cite{nagle_thomason_1985,Halpern2007Characterizing}]\label{lem:NImodels}
	Every pointed $\LG$-model 
	is bisimilar to 
	$\LG$-model whose states are all $\LG$ flat.
\end{lemC}

We are now ready to proceed with our proof of Theorem \ref{thm:tableaux}. 
\begin{proof}[Proof of Theorem \ref{thm:tableaux}]	
This proof is in three parts. 

\textbf{We first prove that statement 1 implies statement 2}
		  To this end, let $b$ be a maximal open branch in the tableau for $\varphi$.
		We construct a $\klogic_k^\mu$-model $\M = (W,R,V)$ for $\varphi$ in the following way.
		Let $W$ be the set of prefixes that appear in the branch, and 
		let, for each $\al \in \act$, 
		$$R_\al^0 = \{ (\sigma,\sigma\nStt{\al}{\psi}) \in W^2 \} \cup \left\{ (\sigma,\sigma) \in W^2 \mid \LG(\al) \text{ has }
		\begin{array}{l}
			\text{ reflexive frames, or } \\
			\text{ serial frames and } \\ ~\forall \psi.~\sigma\nStt{\al}{\psi} \notin W^2
		\end{array} 
		\right\};$$
		$R_\al^1$ is the symmetric closure of $R_\al^0$, if $\LG(\al)$ has symmetric frames, and it is $R_\al^0$ otherwise;
		$R_\al^2$ is the euclidean closure of $R_\al^1$, if $\LG(\al)$ has euclidean frames, and it is $R_\al^1$ otherwise;
		and finally, 
		$R_\al$ is the transitive closure of $R_\al^2$, if $\LG(\al)$ has transitive frames, and it is $R_\al^2$ otherwise.
		By Lemma \ref{lem:conditionsarepreserved}, $R_a$ satisfies all the necessary closure conditions.
		We also set $V(p) = \{ \sigma \in W\mid \sigma~p \text{ appears in the branch} \}$. 
		
		It is now possible to prove, by straightforward induction, that for every subformula $\psi$ of $\varphi$, if $\sigma~\psi$ appears in the branch, then for any environment 
		$\rho$, such that $\{ \sigma' \in W\mid \sigma~\psi \xrightarrow{X}^* \sigma'~X 
		\}\subseteq\rho(X)$, \changeJ{$\sigma \in \trueset{\psi,\rho}$}.
		The only interesting cases are fixed-point formulas, so let $\psi = \mx X.\psi'$. Let $S_X$ be the set of prefixes of $X$ in the branch. We can immediately see that if $\sigma~X$ appears in the branch, then so does $\sigma~\psi'$, and therefore, by the inductive hypothesis,
		\changeJ{$S_X \subseteq \trueset{\psi',\rho[X\mapsto S_X]}$}. 
		From the semantics in Table \ref{table:semantics}, 
		$\sigma \in \trueset{\psi,\rho}$.
		
		On the other hand, if $\psi = \mn X.\psi'$, then 
		we prove that if $\sigma \notin S \subseteq W$, then
		$S\not\supseteq\trueset{\psi',\rho[X\mapsto S]}$.
		Let 
		\changeJ{$\Psi = \{ \sigma'~\chi \mid \sigma~\psi' \xrightarrow{X}^* \sigma'~\chi  \text{ \and } \sigma' \notin S \}$}. 
		We know that $\sigma~\psi' \in \Psi$, so $\Psi\neq \emptyset$.
		There are no infinite $\xrightarrow{X}$-paths in the branch, so there is some $\sigma'~\psi' \in \Psi$, such that $\sigma'~\psi' \not \xrightarrow{X}^+ \sigma''~\psi'$ for any $\sigma''$.
		Then, we see that $\{ \sigma'' \in W\mid \sigma'~\psi \xrightarrow{X}^* \sigma''~X 
		\}\subseteq S $, because if $\sigma'~\psi \xrightarrow{X}^* \sigma''~X$, then $\sigma'~\psi' \xrightarrow{X}^* \sigma''~\psi'$. 
		But then, by the inductive hypothesis, $\sigma' \in \trueset{\psi',\rho[X\mapsto S]}$, and therefore $S\not\supseteq\trueset{\psi',\rho[X\mapsto S]}$, which was what we wanted to prove.
		
\textbf{We now prove that 2 implies 3,} when
	there is a 
	$\kappa:\nat \to \nat$, such that  every $\LG$-satisfiable formula $\varphi$ is satisfied in a model with at most $\kappa(|\varphi|)$ states. 
	       	Assume that $\M = (W,R,V)$ be a $\logicize{L}$-model, $w\in W$, $(\M,w)\models \varphi$, and $W$ has at most $\kappa(|\varphi|)$ states. 
		In the rest of the proof, we fix a lfp-finite dependency relation on $\M$ and $\varphi$, according to Theorem \ref{thm:lfp-finite-dep}, and use it to construct an $\LG$-tableau with a sufficiently open branch for $\varphi$.
		The tableau starts with $\varepsilon~\varphi$ and we can keep expanding this branch to a sufficient one using the tableau rules, such that every prefix is mapped to a state in $W$, whenever $\sigma$ is mapped to $u$, $(\M,u) \models_\rho \psi$ for every $\psi \in \form(\sigma)$, and if $\sigma\nStt{\al}{\psi}$ to $v$, then $u R_\al v$; furthermore, this can be done by following the lfp-finite dependency relation.
		This is by straightforward induction on the application of the tableau rules. A special case are the agents with euclidean accessibility relations, for which we can use 
		Lemma \ref{lem:NImodels}.
		It is not hard to see that in this way we generate a set of branches that are not  propositionally closed.
		Furthermore, since the tableau rule applications follow a lfp-finite dependency relation and $W$ has at most $\kappa(|\varphi|)$ states, it is not hard to see that
		for every least-fixed-point variable $X$, on every
		$\xrightarrow{X}$-path,  $X$ appears at most $\kappa(|\varphi|)$ times.

 \textbf{We now prove that 2 implies 1,} without assuming the finite model property for $\LG$. To do this, we follow exactly the reasoning of the implication from $2$ to $3$, without assuming that the model has $\kappa(|\varphi|)$ states.

		\textbf{Finally, we  prove that 3 implies 1,}  when
	there is a 
	$\kappa:\nat \to \nat$, such that  every $\LG$-satisfiable formula $\varphi$ is satisfied in a model with at most $\kappa(|\varphi|)$ states. 
		To do so, we assume that there is a sufficient open branch $b$ in a tableau for $\varphi$ and we demonstrate that $\varphi$ has a maximal tableau with an open branch --- specifically, we construct such an open branch from $b$.
		
		We call a prefix $\sigma$ a leaf when there is no $\sigma'\neq \sigma$ in the branch, such that $\sigma \prec \sigma'$; we call $\sigma$ productive when a tableau rule on a formula $\sigma~\psi$ in the branch can produce a new prefix.
		We call $\sigma$ ready if it is of the form $\sigma'\nStt{\al}{\psi_1}\nStt{\al}{\psi_2}$, or if $\sigma$ is not $\al$-flat for any $\al \in \act$.
		
		For each $\sigma~\psi$ in $b$ and each least-fixed-point variable $X$, let 
		$$c(\sigma~\psi, X) = \max \left\{ n  ~~\Big|~~ \text{there is a }\xrightarrow{X}
		\begin{array}{l}\text{-path that ends in }\sigma~\psi, \\ \text{where } X \text{ appears }n \text{ times} 
		\end{array} \right\}.$$
		Since $b$ is not sufficiently closed, always $c(\sigma~\psi,X)\leq \kappa(|\varphi|)$.
		We use
		the notation $\sigma \sim \sigma'$ to mean that $\sigma=\sigma_1\nStt{\al}{\psi_1}$, $\sigma'=\sigma_2\nStt{\al}{\psi_2}$ for some $\sigma_1,\sigma_2,\al$, and $\form(\sigma)=\form(\sigma')$;
		and the notation $\sigma \equiv \sigma'$ to mean that 
  \begin{itemize}
      \item $\sigma \sim \sigma'$,  
      \item for every least-fixed-point variable $X$ and $\psi \in \form(\sigma)$,
		$c(\sigma~\psi,X)=c(\sigma'~\psi,X)$, 
		and 
        \item 
		either  both $\sigma$ and $\sigma'$ are ready, 
		or 
		both $\sigma$ and $\sigma'$ are not ready and \changeJ{$\sigma_1\sim\sigma_2$}, where $\sigma=\sigma_1\nStt{\al}{\psi}$ and $\sigma'=\sigma_2\nStt{\al}{\psi'}$.
  \end{itemize}
		We then say that $\sigma$ and $\sigma'$ are equivalent.

            The $\equiv$-equivalence class of a ready prefix $\sigma$ is determined by $c(\sigma~\psi,X)$ for every $\psi$ and $X$; by the last action $\al \in \act$ that appears in $\sigma$; and by $\Phi(\sigma) \in 2^{\sub(\varphi)}$.
            Therefore, there are at most $|\act|\cdot \kappa(|\varphi|)^{|\varphi|^2}\cdot 2^{|\varphi|}$ equivalence classes of ready prefixes.
		
		
		Since $b$ is sufficiently open, every productive leaf $\sigma$ in the branch has 
             more than $|\act|\cdot \kappa(|\varphi|)^{|\varphi|^2}\cdot 2^{|\varphi|+1}$ ancestors, and at least half of these must be ready. 
             Therefore, by the Pigeonhole Principle, every productive leaf $\sigma$ has  
            an ancestor $e(\sigma)$ that has 
		a
		distinct, ready, and $\equiv$-equivalent ancestor; let $s(\sigma)$ be 
		such an ancestor of $e(\sigma)$.
		We further assume that $e(\sigma)$ is the $\prec$-minimal ancestor of $\sigma$ with these properties.
		We note that for any two productive leaves $\sigma_1$ and $\sigma_2$, 
		if  $e(\sigma_1)\prec \sigma_2$, then $e(\sigma_1) = e(\sigma_2)$.
		
		
		Let $b_0$ be the branch that results by removing from $b$ all prefixed formulas of the form $\sigma\nStt{\al}{\psi_1}~\psi_2$, where $e(\sigma') \prec\sigma$ for some productive leaf $\sigma'$.
		Observe that $b_0$ is such that each of its productive leaves has an equivalent ancestor that is not a leaf.
		%
		To complete the proof, it suffices to show how to extend any branch $b_i$, 
		where each of its productive leaves $\sigma$ has an equivalent proper ancestor $s(\sigma)$, to a branch $b_{i+1}$ that preserves this property, and has an increased minimum length of its productive leaves.
		To form $b_{i+1}$, simply add to $b_i$ all formulas of the form 
		$\sigma.\sigma'~\psi$, where $s(\sigma).\sigma'~\psi$ appears in $b_i$.
		%
 \end{proof}

\begin{cor}
	$\logicize{L}$-tableaux are sound and complete for $\logicize{L}$.
\end{cor}

%% file: k4mu_pspace.tex

In this section, we prove that the satisfiability problem for the (single-agent) logics $\kf^\mu$, $\df^\mu$, and $\sr^\mu$ is in \PSPACE, using an optimized decision procedure based on the tableaux from Section \ref{sec:multi}.
To achieve this, we will examine how transitivity allows us to restrict the parts of a branch that we need to examine, in order to verify that a formula is satisfiable.

 
The tableaux rules we use are shown in Table \ref{tab:tableau_for_four}, and are 
the restriction of 
the rules defined in Table \ref{tab:tableau} to $\kf^\mu$, $\df^\mu$, and $\sr^\mu$ \changeJ{for the single-agent setting}.
We note that as the prefixes of the tableau represent states in a satisfying model, the combination of rules (\textsf{B}) and (\textsf{4}) ensures that from $\Box \varphi$ we introduce $\varphi$ to all subsequent prefixes, as one would want in a transitive frame.
\alwaysRootAtBottom
\begin{table}[t]
	{\centering 
\begin{gather*}
		\AxiomC{$\sigma~\pi X.\varphi$}
		\RightLabel{(\textsf{fix})}
		\UnaryInfC{$\sigma~\varphi$}
		\DisplayProof \quad
		\AxiomC{$\sigma~X$}
		\RightLabel{(\textsf{X})}
		\UnaryInfC{$\sigma~\fix(X)$}
		\DisplayProof 
		\quad
		\AxiomC{$\sigma~\varphi \lor \psi$}
		\RightLabel{(\textsf{or})}
		\UnaryInfC{$\sigma~\varphi\mid\sigma~\psi$}
		\DisplayProof \quad		
		\AxiomC{$\sigma~\varphi \land \psi$}
		\RightLabel{(\textsf{and})}
		\UnaryInfC{$\sigma~\varphi$}
		\noLine
		\UnaryInfC{$\sigma~\psi$}
		\DisplayProof 
		\\[4ex]
	\AxiomC{$\sigma~\Box\varphi$}
	\RightLabel{(\textsf{B})}
	\UnaryInfC{$\sigma\nStt{}{\psi}~\varphi$}
	\DisplayProof \quad
	\AxiomC{$\sigma~\Diamond\varphi$}
	\RightLabel{(\textsf{D})}
	\UnaryInfC{$\sigma\nStt{}{\varphi}~\varphi$}
	\DisplayProof~ \quad
	\AxiomC{$\sigma~\Box\varphi$}
	\RightLabel{(\textsf{d})}
	\UnaryInfC{$\sigma\nStt{}{\varphi}~\varphi$}
	\DisplayProof~\quad
	\AxiomC{$\sigma~\Box\varphi$}
	\RightLabel{(\textsf{4})}
	\UnaryInfC{$\sigma\nStt{}{\psi}~\Box\varphi$}
	\DisplayProof~\quad
	\AxiomC{$\sigma~\Box\varphi$}
	\RightLabel{(\textsf{t})}	
	\UnaryInfC{$\sigma~\varphi$}
	\DisplayProof
	\\
\end{gather*}  }
\\[2ex]
\noindent
where, for rules (\textsf{B}) and (\textsf{4}), 
\change{$\sigma\nStt{}{\psi}$} has already appeared in the branch.
\noindent
\caption{The tableau rules for $\logicize{L} \in \{\logicize{\kf}^\mu, \df^\mu, \sr^\mu\}$.}
\label{tab:tableau_for_four}
\end{table}

Given a branch $b$ and prefix $\sigma \in \sigma(b)$, we define $b[\sigma:] = \{ \sigma'~\psi \mid \sigma.\sigma'~\psi \in b \}$. 
A branch $b$ for $\varphi$ is called \emph{repeating} if 
$b[\sigma\diam{\psi}:] = b[\sigma\diam{\psi}\sigma'\diam{\psi}:]$
for every $\sigma\diam{\psi}\sigma'\diam{\psi} \in \sigma(b)$.

\begin{lem}\label{lem:4-repeating-branch}
	Let $\LG \in \{ \kf^\mu, \df^\mu, \sr^\mu \}$.
	A formula $\varphi$ is $\LG$-satisfiable if and only if 
	it has a maximal open, repeating branch $b$.
\end{lem}
\begin{proof}
	Let $\varphi$ be satisfiable. 
	By Corollary \ref{cor:smallLmodel}, $\varphi$ has a finite model $\M = (W,R,V)$.
	We revisit the proof of Theorem \ref{thm:tableaux}, 
	and specifically the construction of the tableau from the finite model.
	We can ensure that when a new prefix $\sigma\diam{\psi}$
	is mapped to a state $u_{\sigma\diam{\psi}}$, that state is in a strongly connected component (SCC) that is maximally deep in the dag of the SCCs in $(W,R)$, such that $(\M,u_{\sigma\diam{\psi}}) \models \psi$.
	So, if an extension $\sigma\diam{\psi}\sigma'\diam{\psi}$ is mapped to some state $v$, then either the SCC of $u_{\sigma\diam{\psi}}$ is not maximal or $v$ and $u_{\sigma\diam{\psi}}$ are in the same SCC, and therefore we can ensure that $ v = u_{\sigma\diam{\psi}}$.
	The converse direction is immediate from Theorem \ref{thm:tableaux}.
\end{proof}

Repeating branches ensure that an algorithm that checks for the satisfiability of $\varphi$ only needs to consider prefixes of length linear with respect to $|\varphi|$. We can go further and consider branches with less redundancy: due to transitivity, the satisfaction of certain diamond formulas in some prefix may be further postponed.

\begin{defi}\label{def:fourfour}
We will use the following terms: 
\begin{enumerate}
    \item  We call a branch $b$ \emph{terse}, if 
        \begin{itemize}
              \item no formula appears twice in a prefix, and 
               \item for all $\psi_1 \neq \psi_2$, if $\sigma.\diam{\psi_1}, \sigma.\diam{\psi_2} \in \sigma(b)$, then $\Diamond\psi_2$ is not prefixed by any extension of $\sigma.\diam{\psi_1}$ in $b$. 
        \end{itemize}
    \item The branch $b$ is \emph{maximally terse}, if
        \begin{itemize}
            \item  it is terse, and
            \item closed under the tableau rules, except $(D)$ and $(d)$, but if $\sigma~\Diamond\psi \in b$ (\resp or $\sigma~\Box\psi \in b$, if $(d)$ is a tableau rule), then $\sigma'\diam{\psi}~\psi \in b$ for some $\sigma'$, where $\sigma'$ is an extension of $\sigma$ or $\sigma$ is an extension of $\sigma'$; and if $\sigma~\Box\psi \in b$ and $\sigma~\Diamond\psi' \in b$, then $\sigma'\diam{\psi'}~\psi \in b$ for any $\sigma$ that is an extension of $\sigma'\diam{\psi'}$.
        \end{itemize}
    \item A maximally terse branch $b$ is called $4$-open when $b$ remains open, after 
extending the definition of $\xrightarrow{X}$, so that additionally: 
\begin{itemize}
	\item $\sigma~\Diamond\psi \xrightarrow{X} \sigma'\diam{\psi}~\psi$ for every $\sigma'$, where $\sigma'$ is an extension of $\sigma$ or $\sigma$ is an extension of $\sigma'$; 
	\item $\sigma~\Box\psi \xrightarrow{X} \sigma'\diam{\psi}~\psi$ for every $\sigma'$, where $\sigma'$ is an extension of $\sigma$ or $\sigma$ is an extension of $\sigma'$, if $(d)$ is included in the tableau rules; and
	\item $\sigma~\Box\psi \xrightarrow{X} \sigma'\diam{\psi'}~\psi$ for every $\sigma'$, where $\sigma$ is an extension of $\sigma'\diam{\psi}$ and $\sigma~\Diamond\psi' \in b$.
\end{itemize}
\end{enumerate}
   
\end{defi}


\begin{lem}\label{lem:4-repeating-restricted-branch}
	Let $\LG \in \{ \kf^\mu, \df^\mu, \sr^\mu \}$.
	A formula $\varphi$ is $\LG$-satisfiable if and only if 
	it has a maximally terse, $4$-open branch.
\end{lem}
\begin{proof}
	If $\varphi$ has a maximally terse, $4$-open repeating branch, we can construct a model for $\varphi$, similarly as in the proof of Theorem \ref{thm:tableaux}. 

	Conversely, Let $\varphi$ be satisfiable and let $b$ be a maximal, open, repeating branch for $\varphi$ by Lemma \ref{lem:4-repeating-branch}. 
	Let $b^-$ be a maximal 
	terse
	sub-branch of $b$.
	It is immediate that $b^-$ is closed under the tableau rules, except $(D)$ and $(d)$. 
	Due to the maximality of $b^-$, 
	if $\sigma~\Diamond\psi \in b^-$ (\resp or $\sigma~\Box\psi \in b$, if $(d)$ is a tableau rule), then $\sigma'\diam{\psi}~\psi \in b$ for some $\sigma'$, where $\sigma'$ is an extension of $\sigma$ or $\sigma$ is an extension of $\sigma'$ -- otherwise, $b^-$ can be extended to a larger sub-branch of $b$.
	Since $b$ is repeating, if $\sigma~\Box\psi \in b^-$ and $\sigma~\Diamond\psi' \in b^-$, then $\sigma'\diam{\psi'}~\psi \in b^-$ for any $\sigma$ that is an extension of $\sigma'\diam{\psi'}^-$.
	Therefore, $b^-$ is maximally terse.

	We use the extended  definition of $\xrightarrow{X}$
 as in the definition of $4$-open branches in Definition \ref{def:fourfour}.
	By the transitivity of $R$ and following the proof of Theorem \ref{thm:tableaux}, 
	since $b$ is repeating, and due to rule $(4)$,
        every $\xrightarrow{X}$-cycle in $b^-$ corresponds to an infinite $\xrightarrow{X}$-path in $b$, and therefore
	$b^-$ is $4$-open.
\end{proof}

\begin{cor}\label{cor:4-small-model}
	Let $\LG \in \{ \kf^\mu, \df^\mu, \sr^\mu \}$.
	A formula $\varphi$ is $\LG$-satisfiable if and only if 
	it is satisfied in a model of at most $2|\varphi|!-1$ states.
\end{cor}

\begin{lem}\label{lem:4-terse-is-good}
	Let $b$ be a terse branch. If a $\xrightarrow{X}$-path in $b$ visits two prefixes, then one is an extension of the other.
\end{lem}

\begin{thm}\label{thm:4-PSPACE}
	The satisfiability problem for $\kf^\mu$, $\df^\mu$, and $\sr^\mu$, is \PSPACE-complete.
\end{thm}

\begin{proof}
	\PSPACE-hardness results from Theorem \ref{thm:ladhalp}. We prove that $\kf^\mu$-, $\df^\mu$-, and $\sr^\mu$-satisfiability is in \PSPACE.
	We describe a non-deterministic polynomial-space algorithm that, given a formula $\varphi$, constructs and explores a maximally terse, $4$-open $\kf^\mu$-branch for $\varphi$, if one exists, and it verifies that the branch is $4$-open. By Lemma \ref{lem:4-repeating-restricted-branch}, this suffices for checking the satisfiability of $\varphi$.

	Our algorithm can generate and explore a maximally terse branch $b$ with a depth-first-search on its prefixes. 
	Note that the maximality and $4$-openness conditions for $b$ require only checking pairs of prefixes where one is an extension of the other.
	The algorithm keeps track of the diamond formulas (or also the box formulas, in case of rule (D)) that have generated a prefix.
	When the algorithm has visited a maximally long prefix $\sigma$, it can verify that it contains no $\xrightarrow{X}$-cycles where $X$ appears. It also verifies that no $\sigma'~\Diamond\psi$ is marked as having generated $\sigma'\diam{\psi}$, where $\sigma$ extends $\sigma'$ but not $\sigma'\diam{\psi}$.
	If either of those verifications fail, the algorithm rejects.
	The algorithm then backtracks until all these formulas have generated a prefix.
	If there is some maximally terse, $4$-open branch for $\varphi$, then the algorithm will succeed by following that branch.
	If, on the other hand, the algorithm succeeds, then the prefixed formulas it generates form a branch that is terse, as it marks the diamond formulas that generate a prefix, and verifies that the terseness condition is satisfied. 
	It is also not hard to see that the generated branch is maximally terse and $4$-open.
	The algorithm uses at most polynomial space to store a maximal prefix, the prefixed formulas, and the markings of these formulas, while it performs the depth-first-search exploration. The verifications it performs at each maximal prefix only require polynomial space.
\end{proof}

%% file: translations-tableaux.tex

In this section, we give a uniform translation from every modal logic with fixed-points that we have defined in this paper to the $\mu$-calculus.
This translation thus establishes the decidability and an upper bound for the complexity of satisfiability for all the cases that were not provided an upper bound via the methods used in Section~\ref{sec:translations}.
The key observation in this section is that we already have prefixed tableau proof systems for our logics, with the following properties:
\begin{itemize}
    \item the subformula property, in that all formulas that appear in a tableau are prefixed subformulas of the starting formula;
    \item the tableau rules are local, in that they relate formulas prefixed either with the same prefix or with prefixes that differ very little from one another;
    \item the tableau proof systems produce a complete open branch when the starting formula is satisfiable. 
\end{itemize}
An open branch has a bounded-branching tree structure with respect to the prefixes, but that tree can be infinite.
The fact that a model or a branch can be infinite appears to cause difficulties in deciding the satisfiability of a logic, especially when symmetry or negative introspection is involved \cite{Dagostino2013S5}.
However, we observe that the fixed-points in the $\mu$-calculus, the logic that is the target of our translations, are explicitly designed to describe infinite models.
Our approach is to view an open branch as a model and to describe the property that the branch is open via a $\mu$-calculus formula.
The caveat is that our translation results in a formula that can be exponentially larger than the original, and thus yields a double-exponential upper bound for the satisfiability problem.
For the sake of simplicity, in what follows we first assume that axiom $5$ is not present.  
Afterward, we explain how to adjust our construction to handle agents with negative introspection.

\subsection{The case of logics without \NI}

For a given formula $\varphi$, we construct a new formula that describes an open tableau branch for $\varphi$. 
To that end, we introduce exponentially many new propositional variables, each a directed graph of subformulas of $\varphi$ and their negations, with edges that describe the tableau productions. 
We then describe in the $\mu$-calculus 
a propositionally open branch that is closed under the tableau rules, and then that the branch is open. 

Let $\varphi$ be a formula where all fixed-point variables are distinct, and let $mV$ be the set of least-fixed-point variables in $\varphi$.
Let $\G$ be the set of labelled directed graphs of the form 
$g = (\Phi,(\xrightarrow{X})_{X \in mV},\ell)$, where 
\begin{itemize}
    \item $\Phi \subseteq \sub(\varphi)$ is its set of vertices, 
    \item each edge is labelled with a variable from $mV$, and
    \item each node $\psi$ is additionally labelled with a set $\ell(\psi) \subseteq \{ in, out \}\times \{ top, bot \}$.
\end{itemize}  
We use the notation $\Phi_g$, $\xrightarrow{X}_g$, and $\ell_g$ for $\Phi$, $\xrightarrow{X}$, and $\ell$, respectively, when $g = (\Phi,{(\xrightarrow{X})}_{X \in mV},\ell)$.
We also abuse notation and write that $(\psi,u,v) \in g$ to mean that $\psi \in \Phi_g$ and $(u,v) \in \ell_g(\psi)$. 

\changeJ{
We use graphs from $\G$ to represent a part of a tableau branch that is prefixed with a single tableau prefix. Intuitively, we want a $g \in \G$ to represent a tableau prefix $\sigma$, such that $\Phi_g$ is the set of formulas prefixed with $\sigma$; $\xrightarrow{X}_g$ represents the relation $\xrightarrow{X}$ as defined in Section \ref{sec:multi}, restricted in the $\sigma$-prefixed part of the branch; and $(in,bot) \in \ell_g(\psi)$ (\resp $(out,bot) \in \ell_g(\psi)$) represents that $\psi$ may be affected by (\resp may affect) a subsequent prefix in the branch through some rule.
This intuition is reflected in the following definition, where we describe the graphs that are suitable to represent prefixes of a complete branch.
}

\begin{defi}\label{def:compatible}
We say that $g = (\Phi,(\xrightarrow{X})_{X \in mV},\ell)\in \G$ is \emph{\change{locally} compatible with the tableau rules} and write $com(g)$ when for every $\psi \in \Phi$ and $X \in mV$, where $\psi \neq Y$ for any $X<Y$, 
\changeJ{
the following conditions are satisfied. 
For every pair $\sigma, \sigma'$ of possible tableau prefixes and  $\psi'\in \sub(\varphi)$
}:
\begin{enumerate}
	\item 
 if 
 $\frac{\sigma~\psi}{\sigma~\psi'}$
  is a tableau rule, then 
  $\psi' \in \Phi$ and $\psi \xrightarrow{X} \psi'$;
	\item
 if
	$\frac{\sigma~\psi}{\genfrac{}{}{0pt}{1}{\sigma~\psi_1}{\sigma~\psi_2}}$
 is a tableau rule, then  
	$\psi_1 \in \Phi$, $\psi \xrightarrow{X} \psi_1$, $\psi_2 \in \Phi$, and $\psi \xrightarrow{X} \psi_2$;
	\item 
 if 
 $\frac{\sigma~\psi}{\sigma~\psi_1~\mid~ \sigma~\psi_2}$
 is a tableau rule, then 
 $\psi_1 \in \Phi$ and $\psi \xrightarrow{X} \psi_1$, or $\psi_2 \in \Phi$ and $\psi \xrightarrow{X} \psi_2$;
	\item 
 if 
 $\frac{\sigma~\psi}{\sigma'~\psi'}$
 is a tableau rule 
	and $\sigma'$ is a strict extension of $\sigma$, then 
	$(out,bot) \in \ell(\psi)$;
	\item 
 if 
 $\frac{\sigma~\psi}{\sigma'~\psi'}$
 is a tableau rule and
	$\sigma$ is a strict extension of $\sigma'$,  then 
	$(out,top) \in \ell(\psi)$; and
	\item 
	$\neg \psi \notin \Phi$ and $\false \notin \Phi$.
\end{enumerate}

\changeJ{Let $\G^t = \{ g \in \G \mid com(g) \}$. }
\end{defi}

\begin{defi}\label{def:achild}
Let $\al \in \act$, $\chi \in \sub(\varphi)$, and $g = (\Phi_g,(\xrightarrow{X}_g)_{X \in mV},\ell_g)$ and $h =(\Phi_h,{(\xrightarrow{X}_h)}_{X \in mV},\ell_h) \in G^t$.
We say that $g $ is \changeJ{a prospective} $\al\diam{\chi}$-child of $h$, and write $c_{\al\diam{\chi}}(h,g)$ when:
\begin{enumerate}
	\item $\diam{\al}\chi \in \Phi_h$, or  $[\al]\chi \in \Phi_h$ if $\mathsf{(d)}$ is a rule,  and $\chi \in \Phi_g$ and $(in,top) \in \ell_g(\chi)$;
	\item if $[\al]\psi \in \Phi_h$ and $\frac{\sigma~[\al]\psi}{\sigma.\al\diam{\psi_1}~\psi_2}$ is a tableau rule, then $\psi_2 \in \Phi_g$ and $(in,top) \in \ell_g(\psi_2)$; and 
	\item if $[\al]\psi \in \Phi_g$ and $\frac{\sigma.\al\diam{\psi_1}~[\al]\psi}{\sigma~\psi_2}$ is a tableau rule, then $\psi_2 \in \Phi_h$ and $(in,bot) \in \ell_h(\psi_2)$.
\end{enumerate}
We can similarly define that $g $ is \changeJ{a prospective} $\al$-child of $h$, and write $c_{\al}(h,g)$, by omitting the first condition.
\end{defi}

\begin{remark}
\changeJ{
We note that if one applies the conditions of Definition \ref{def:compatible} 
for $com(g)$ to a set of tableau formulas with a common prefix, they do not suffice to declare that 
the tableau rules have been applied correctly. 
As such, $com(g)$ ensures that if $g$ represents a prefix $\sigma$, then $\Phi_g$ is closed under the first line of Table \ref{tab:tableau}, and possibly rule $(\logicize{t})$.
The labeling $\ell_g$ marks the formulas that interact with other prefixes in the tableau -- the immediate extensions or prefixes of $\sigma$.
Similarly, in the context of $com(g)$ and $com(h)$,  $c_{\al\diam{\chi}}(h,g)$ ensures that 
if $h$ represents a tableau prefix $\sigma$, then the branch can be extended in a way such that $\sigma\al\diam{\chi}$ is represented by $g$.
If one views a tableau branch as the tree of its prefixes, labelled with sets of formulas, then the collection of pairs $h,g$, such that $c_{\al\diam{\chi}}(h,g)$, gives us the allowed $\al\diam{\chi}$-children from a node labelled with $\Phi_h$.
%
}

\changeJ{
Naturally, these predicates do not account for the conditions for applying the box tableau rules, such as (B) and (4). For that, we use the formula $\logicize{rules}$, introduced below (see also Lemma \ref{lem:b-to-model-representing}).}
\end{remark}

Observe that the size of $\G$ is at most exponential with respect to $\varphi$ and that it takes at most polynomial time to verify $com(g)$, $c_{\al\diam{\chi}}(h,g)$, and $c_{\al}(h,g)$.

\changeJ{We now start to describe a tableau branch in $\klogic^\mu_{|\sub(\varphi)|\cdot |\act|}$. We construct formulas that use each $g \in G^t$ }
as a propositional variable \changeJ{
and each $\al\diam{\psi}$, where $\al \in \act$ and $\psi \in \sub(\varphi)$, as an agent.
We use the notation $\act_{\varphi} = \{\al\diam{\psi} \mid \al \in \act \text{ and } \psi \in \sub(\varphi)\}$,  $[\al]\psi$ as a shorthand for $\bigwedge_{\chi \in \sub(\varphi)}[\al\diam{\chi}]\psi$, and dually, 
$\diam{\al}\psi$ as a shorthand for $\bigvee_{\chi \in \sub(\varphi)}\diam{\al\diam{\chi}}\psi$.
}
Let 
\begin{align*}
	\logicize{rules} := 
	\Inv
	\left(\bigvee_{g \in G^t} g \land \bigwedge_{\substack{g,h \in G^t \\ g \neq h}} \neg (g \land h) 
	~~\land~ 
 \begin{array}{l}
        \displaystyle
      \bigwedge_{\substack{g \in G^t \\ \diam{\al}\psi \in \Phi_g}} 
      \left(
      g \impl \diam{\al\diam{\psi}}\bigvee_{c_{\al\diam{\psi}}(g,h)} h
      \right)
	~~\land \\
       \displaystyle
        \bigwedge_{\substack{g \in G^t \\ [\al]\psi \in \Phi_g \\ \mathsf{(d)} \text{ is a rule}}} 
        \left(
        g \impl \diam{\al\diam{\psi}}\bigvee_{c_{\al\diam{\psi}}(g,h)} h
        \right)
      ~~\land 
      \\
       \displaystyle
      \bigwedge_{\substack{g \in G^t\\ \al \in \act}}
      \left(
      g \impl [\al]\bigvee_{c_{\al}(g,h)} h
      \right)
 \end{array}
	\right)
\end{align*}

\changeJ{
Formula $\logicize{rules}$ ensures that a model behaves in a way that its states correspond to prefixes in a maximal branch that is not propositionally closed.
We now make explicit the correspondence between the graphs in $G^t$ and the prefixes in a branch.}

Let $b$ be a tableau branch \changeJ{with dependence relations $\xrightarrow{X}$ for each $X \in mV$,} and $\sigma$ a prefix in $b$.
We define $g_b(\sigma) = (\Phi,(\xrightarrow{X}_g)_{X \in mV},\ell)$, where 
\begin{itemize} 
\item 
$\Phi = \{\psi \mid \sigma~\psi \in b\}$; 
\item 
for each $X \in mV$, 
\changeJ{
$\xrightarrow{X}_g = \{ (\psi_1,\psi_2) \in \Phi^2 \mid (\sigma~\psi_1,\sigma~\psi_2) \in \xrightarrow{X} \}$;} 
and 
\item 
for each $\psi \in \Phi$, $\ell(\psi)$ is such that:
\subitem 
$(in,top) \in \ell(\psi)$ (\resp $(in,bot) \in \ell(\psi)$) iff $\sigma'~\psi' \xrightarrow{X} \sigma~\psi$ for some $\sigma'~\psi' \in b$, where $\sigma'$ is  a strict prefix of $\sigma$ (\resp $\sigma$ is  a strict prefix of $\sigma'$); and
\subitem 
$(out,top) \in \ell(\psi)$ (\resp $(out,bot) \in \ell(\psi)$) iff $\sigma~\psi \xrightarrow{X} \sigma'~\psi'$ for some $\sigma'~\psi' \in b$, where $\sigma'$ is  a strict prefix of $\sigma$ (\resp $\sigma$ is  a strict prefix of $\sigma'$).
\end{itemize}
We may omit $b$ from $g_b(\sigma)$ and write $g(\sigma)$ instead, if $b$ is clear from the context.

\begin{defi}\label{def:model-represents-branch}
	Let $b$ be a tableau branch, 
	$\M = (W,R,V)$ be a model, using $\G^t$ as the set of propositional variables and $\act_\varphi$ as the set of agents, and $s \in W$.
	We say that $\M,s$ \emph{represent} $b$ when there is a 
	mapping $f: \sigma(b) \to W$, such that:
	\begin{enumerate}
		\item $f(\epsilon) = s$; 
		\item \label{item:states-describe-prefixes}
		for every $\sigma \in \sigma(b)$, 
		$g_b(\sigma) \in \G^t$ and 
		$(\M,f(\sigma)) \models 
            g_b(\sigma)$; 
		\item \label{item:states-describe-one-prefix}
		for every $\sigma \in \sigma(b)$ and 
		$g \in \G^t$, if 
		$(\M,f(\sigma)) \models 
            g$, then $g = g_b(\sigma)$; 
		\item \label{item:every-prefix-to-a-state}
		for every $\sigma.\al\diam{\psi} \in \sigma(b)$, $f(\sigma) R_{\al\diam{\psi}} f(\sigma.\al\diam{\psi})$; and
		\item \label{item:every-transition-to-a-prefix}
		for every $\sigma \in \sigma(b)$ and $t \in W$, $f(\sigma) R_{\al\diam{\psi}} t$ implies that $\sigma.\al\diam{\psi} \in \sigma(b)$ and $f(\sigma.\al\diam{\psi}) = t$.
	\end{enumerate}
	If the above conditions are satisfied, then we use the notation $f:b \xrightarrow{r} \M,s$.
\end{defi}

\begin{lem}\label{lem:rules-iff-ex-b--rep-impl-max-not-closed}
	For every model $\M = (W,R,V)$ and state $s \in W$:
	\begin{enumerate}
		\item 
		$(\M, s) \models \logicize{rules}$ if and only if $\M,s$  represent some branch that is maximal and not propositionally closed; and 
		\item if
		$(\M, s) \models \logicize{rules}$
		and 
		$\M,s$  represent $b$, then $b$ is maximal and not propositionally closed.
	\end{enumerate}
\end{lem}

\begin{proof}
	We first prove the second statement of the lemma.
	Let $b$ be represented by $(\M,s)\models \logicize{rules}$.
	As every prefix is mapped to a state that is reachable from $s$ (item \ref{item:every-prefix-to-a-state} of Definition \ref{def:model-represents-branch}), because of item \ref{item:states-describe-prefixes}, $b$ is locally maximal.
	We further observe that item \ref{item:every-prefix-to-a-state} of Definition \ref{def:model-represents-branch} and the conjunct  
	$$\bigwedge_{g \in \G^t} g \impl [\al]\bigvee_{c_{\al}(g,h)} h 
    $$ in $\logicize{rules}$
	ensures that $b$ is closed under all box rules except possibly for $\mathsf{(d)}$; and 
	item \ref{item:every-transition-to-a-prefix} of Definition \ref{def:model-represents-branch} and the conjuncts  
	$$\bigwedge_{\substack{g \in \G^t \\ \diam{\al}\psi \in \Phi_g}} g \impl \diam{\al}\bigvee_{c_{\al\diam{\psi}}(g,h)} h
        \land 
        \bigwedge_{\substack{g \in \G^t \\ [\al]\psi \in \Phi_g \\ \mathsf{(d)} \text{ is a rule}}} g \impl \diam{\al\diam{\psi}}\bigvee_{c_{\al\diam{\psi}}(g,h)} h
        $$ in $\logicize{rules}$
	ensures that $b$ is closed under all diamond rules and $\mathsf{(d)}$, if it is a rule of the logic for the tableau.
	Therefore, $b$ is maximal and not propositionally closed.

	We now prove the first statement of the lemma.
	If $(\M,s) \models \logicize{rules}$, then we demonstrate how to construct a branch $b$ that is represented by $\M,s$, together with the mapping $f$ of the prefixes to the states that are reachable from $s$, as in Definition \ref{def:model-represents-branch}. 
	As we have already proven the second part of the lemma, $b$ is maximal and not propositionally closed.
	Note that at every state $t$ that is reachable from $s$, exactly one variable from $\G^t$ is true, and we call it $g(t)$.
	The construction of $b$ is by straightforward recursion: 
	\begin{itemize}
		\item The branch 
		$b$ includes 
		the prefix $\varepsilon$, $f(\varepsilon) = s$, and $g(\varepsilon) = g(s)$.
		\item 
		For every $\sigma \in \sigma(b)$ and $\diam{\al}\psi \in \Phi_{g(f(\sigma))}$, $\sigma.\al\diam{\psi}$ is a prefix in $b$; from the conjunct 
		$$\bigwedge_{\substack{g \in \G^t \\ \diam{\al}\psi \in \Phi_g}} g \impl \diam{\al}\bigvee_{c_{\al\diam{\psi}}(g,h)} h$$  in $\logicize{rules}$, there is at least one state $t$, such that $f(\sigma) R_\al t$ and $c_{\al\diam{\psi}}(g(f(\sigma)),g(t))$.
		Let $f(\sigma.\al\diam{\psi}) = t$ for one of these states $t$, and $g(\sigma.\al\diam{\psi}) = g(t)$.
        \item 
		For every $\sigma \in \sigma(b)$ and $[\al]\psi \in \Phi_{g(f(\sigma))}$, if 
        $\mathsf{(d)}$ is a tableau rule, then
        $\sigma.\al\diam{\psi}$ is a prefix in $b$; from the conjunct 
		$$\bigwedge_{\substack{g \in \G^t \\ [\al]\psi \in \Phi_g \\ \mathsf{(d)} \text{ is a rule}}} g \impl \diam{\al\diam{\psi}}\bigvee_{c_{\al\diam{\psi}}(g,h)} h$$  in $\logicize{rules}$, there is at least one state $t$, such that $f(\sigma) R_\al t$ and $c_{\al\diam{\psi}}(g(f(\sigma)),g(t))$.
		Let $f(\sigma.\al\diam{\psi}) = t$ for one of these states $t$, and $g(\sigma.\al\diam{\psi}) = g(t)$.
	\end{itemize}
	It is now straightforward to verify that $\M,s$ represent $b$.

	We are now left to prove that if $\M,s$  represent some branch $b$ that is maximal and not propositionally closed, then $(\M,s) \models \logicize{rules}$.
	First observe that for every prefix $\sigma \in \sigma(b)$, $g(\sigma) \in \G^t$.
	Together with items \ref{item:states-describe-prefixes} and \ref{item:every-transition-to-a-prefix} of Definition \ref{def:model-represents-branch}, this means that 
	$(\M,s) \models \Inv
	\left(\bigvee_{g \in \G^t} g 
	\right)$.
	Item \ref{item:states-describe-one-prefix} of Definition \ref{def:model-represents-branch} yields that 
	\[(\M, s) \models 
	\Inv
	\left(\bigwedge_{\substack{g,h \in G^t \\ g \neq h}} \neg (g \land h) 
	\right); \]
    while items \ref{item:every-prefix-to-a-state} and \ref{item:every-transition-to-a-prefix} yield that 
	$$(\M,s )\models 
	\Inv
	\left(
	\begin{array}{l}
        \displaystyle
      \bigwedge_{\substack{g \in \G^t \\ 
      \diam{\al}\psi \in \Phi_g}} 
      \left(
      g \impl \diam{\al\diam{\psi}}\bigvee_{c_{\al\diam{\psi}}(g,h)} h
      \right)
	~~\land \\
       \displaystyle
        \bigwedge_{\substack{g \in \G^t \\ [\al]\psi \in \Phi_g \\ \mathsf{(d)} \text{ is a rule}}} 
        \left( g \impl \diam{\al\diam{\psi}}\bigvee_{c_{\al\diam{\psi}}(g,h)} h \right)
      ~~\land 
      \\
       \displaystyle
      \bigwedge_{\substack{g \in \G^t\\ \al \in \act}} 
      \left(
      g \impl [\al]\bigvee_{c_{\al}(g,h)} h
      \right)
 \end{array}
	\right).
	$$
    Since $Inc(\varphi_1) \land \Inv(\varphi_2)$ is logically equivalent to $\Inv(\varphi_1 \land \varphi_2)$, it follows that $(\M, s) \models \logicize{rules}$, which was what we wanted to demonstrate.
\end{proof}

\begin{lem}\label{lem:b-to-model-representing}
	For every propositionally open, maximal branch $b$, there is a pointed model $(\M,s)$, such that $\M,s$ represent $b$.
\end{lem}

\begin{proof}
	We define model $\M = (W,R,V)$, such that $W = \sigma(b)$, $R_{\al\diam{\psi}} = \{ (\sigma,\sigma.\al\diam{\psi}) \mid (\sigma,\sigma.\al\diam{\psi}) \in W^2 \}$ for each $\al\diam{\psi} \in \act_\varphi$, and $V(\sigma) = \{g(\sigma)\}$ for each $\sigma \in W$; and $s = \varepsilon$.
	It is now straightforward to verify that $\M,s$ represent $b$ with $f:b \xrightarrow{r} \M,s$, such that $f(\sigma) = \sigma$ for every $\sigma$.
\end{proof}

Let $X \in mV$. For a given prefix $\sigma$ of the tableau, we want to describe that there is no $\xrightarrow{X}$-path where $X$ appears infinitely often. 
Infinite $\xrightarrow{X}$-paths may include infinitely many prefixed formulas, or a cycle with finitely many formulas.
From a local view on a tableau prefix, it can be hard to tell whether two paths are part of the same infinite path of concern.
Thus, we need to rule out all possible combinations of paths that, combined with paths inside other prefixes, can give us infinitely many occurrences of $X$.

\changeJ{
We need to describe that there is a finite $\xrightarrow{X}$-path from one formula to another, and we do so by defining recursive predicates on pairs of formulas.
For this, it is useful to define
additional notation: 
\begin{itemize}
    \item we write $(\psi_1,\psi_2) \in g$ to mean that $(\psi_1,in,top) \in g$ and $ (\psi_2,out,top) \in g$;
    \item for $(\diam{\al}\psi_1,out,bot)$, $(\psi_2,in,bot) \in g$, we define $nx(\diam{\al}\psi_1,\psi_2) = (\psi_1,[\al]\psi_2)$, $nx(\diam{\al}\psi_1) = \psi_1$, and $\al(\diam{\al}\psi_1,\psi_2) = \al(\diam{\al}\psi_1) = \al\diam{\psi_1}$;
    \item for $([\al]\psi_1,out,bot)$, $(\psi_2,in,bot) \in g$, we define $nx([\al]\psi_1,\psi_2) = (\psi_1,[\al]\psi_2)$, $nx([\al]\psi_1) = \psi_1$,  and $\al([\al]\psi_1,\psi_2) = \al([\al]\psi_1) = \al$.
\end{itemize} 
Intuitively, some of the formulas that we construct search through a model to find all the pieces of a finite path from a formula $\psi_1$ to $\psi_2$ that starts and ends in $g$. 
Definition \ref{def:searchTtoS} below defines variations of $(\psi_1,\psi_2) \xrightarrow{g} S$ and  $\psi_1 \xrightarrow{g} S,\psi_2$, which mean that the search should continue with all the pairs in $S$, as illustrated in Figure \ref{fig:finite-paths}. In the case of  $(\psi_1,\psi_2) \xrightarrow{g} S$, the path is part of a path that comes from and returns to a parent prefix; in the case of $\psi_1 \xrightarrow{g} S,\psi_2$, the finite path is part of an infinite path that comes from the parent prefix and continues with an infinite path from $\psi_2$ to a child prefix.
}

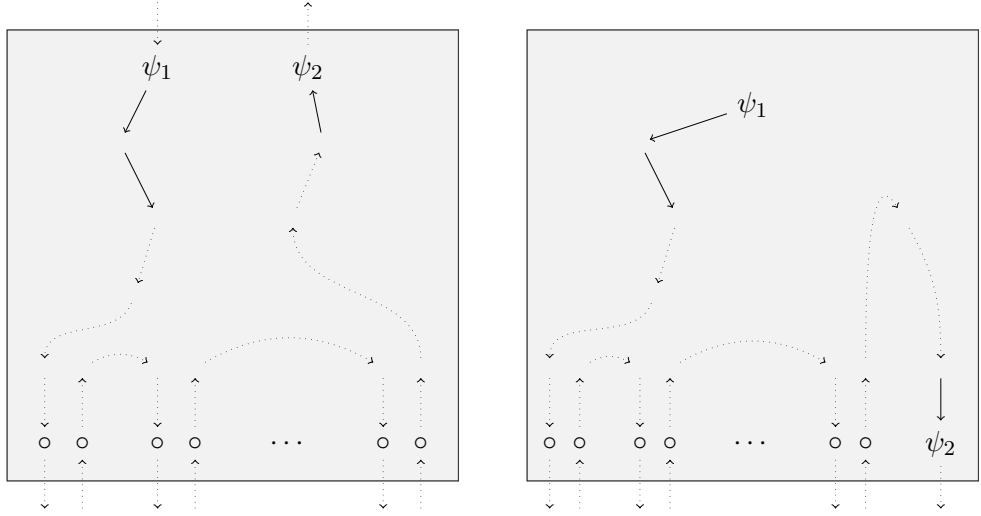
\begin{figure}
    \centering
    \begin{tikzpicture}

  \filldraw[gray!10, draw=gray!80] (0,0) -- (6,0) -- (6,6) -- (0,6) -- cycle;

  \draw[draw=black!80] (0,0) rectangle (6,6);

    \node (start) at (2,5.5) {$\psi_1$}; 
    \node (end) at (4,5.5) {$\psi_2$};
    \node (in) at (2,6.5) {}; 
    \node (out) at (4,6.5) {};

    \node (p1) at (1.5,4.5) {}; 
    \node (q1) at (4.2,4.5) {};

    \node (p2) at (2,3.5) {}; 
    \node (q2) at (3.8,3.5) {};

    \node (p3) at (1.7,2.5) {}; 
    \node (q3) at (5,2.5) {};

    \node (s1) at (0.5,0.5) {$\circ$};
    \node (s2) at (1,0.5) {$\circ$};
    
    \node (s3) at (2,0.5) {$\circ$};
    \node (s4) at (2.5,0.5) {$\circ$};

    \node (Sdots) at (3.75,0.5) {$\cdots$};
    
    \node (s5) at (5,0.5) {$\circ$};
    \node (s6) at (5.5,0.5) {$\circ$};

    \node (t1) at (0.5,-0.5) {};
    \node (t2) at (1,-0.5) {};
    
    \node (t3) at (2,-0.5) {};
    \node (t4) at (2.5,-0.5) {};
    
    \node (t5) at (5,-0.5) {};
    \node (t6) at (5.5,-0.5) {};

    \node (r1) at (0.5,1.5) {};
    \node (r2) at (1,1.5) {};
    
    \node (r3) at (2,1.5) {};
    \node (r4) at (2.5,1.5) {};
    
    \node (r5) at (5,1.5) {};
    \node (r6) at (5.5,1.5) {};


    \draw[line width=0.2pt,->,dotted] (r2) to [out=30,in=150]  (r3);
    \draw[line width=0.2pt,->,dotted] (r4) to [out=30,in=150]  (r5);

   \draw[line width=0.2pt,->,dotted] (p3) to [out=250,in=90]  (r1);
   \draw[line width=0.2pt,->,dotted] (r6) to [out=90,in=270]  (q2);

    \draw[line width=0.2pt,->,dotted] (in) -- (start);
    \draw[line width=0.2pt,->,dotted] (end) -- (out);

  \draw[line width=0.2pt,->] (start) -- (p1);
    \draw[line width=0.2pt,->] (p1) -- (p2);
  \draw[line width=0.2pt,->,dotted] (p2) -- (p3);
  \draw[line width=0.2pt,->] (q1) -- (end);
  \draw[line width=0.2pt,->,dotted] (q2) -- (q1);

 \draw[line width=0.2pt,->,dotted] (s1) -- (t1);
 \draw[line width=0.2pt,->,dotted] (t2) -- (s2);
 \draw[line width=0.2pt,->,dotted] (s3) -- (t3);
 \draw[line width=0.2pt,->,dotted] (t4) -- (s4);
 \draw[line width=0.2pt,->,dotted] (s5) -- (t5);
 \draw[line width=0.2pt,->,dotted] (t6) -- (s6);  

  \draw[line width=0.2pt,->,dotted] (r1) -- (s1);
 \draw[line width=0.2pt,->,dotted] (s2) -- (r2);
 \draw[line width=0.2pt,->,dotted] (r3) -- (s3);
 \draw[line width=0.2pt,->,dotted] (s4) -- (r4);
 \draw[line width=0.2pt,->,dotted] (r5) -- (s5);
 \draw[line width=0.2pt,->,dotted] (s6) -- (r6);  


\end{tikzpicture}
\qquad
    \begin{tikzpicture}

  \filldraw[gray!10, draw=gray!80] (0,0) -- (6,0) -- (6,6) -- (0,6) -- cycle;

  \draw[draw=black!80] (0,0) rectangle (6,6);

    \node (start) at (3,5) {$\psi_1$}; 
    \node (end) at (5.5,0.5) {$\psi_2$};
    
    \node (out) at (5.5,-0.5) {};

    \node (p1) at (1.5,4.5) {}; 
    \node (q1) at (5.5,1.5) {};

    \node (p2) at (2,3.5) {}; 
    \node (q2) at (5,3.5) {};

    \node (p3) at (1.7,2.5) {}; 
    \node (q3) at (5,2.5) {};

    \node (s1) at (0.3,0.5) {$\circ$};
    \node (s2) at (0.7,0.5) {$\circ$};
    
    \node (s3) at (1.5,0.5) {$\circ$};
    \node (s4) at (1.9,0.5) {$\circ$};

    \node (Sdots) at (3,0.5) {$\cdots$};
    
    \node (s5) at (4.1,0.5) {$\circ$};
    \node (s6) at (4.5,0.5) {$\circ$};

    \node (t1) at (0.3,-0.5) {};
    \node (t2) at (0.7,-0.5) {};
    
    \node (t3) at (1.5,-0.5) {};
    \node (t4) at (1.9,-0.5) {};
    
    \node (t5) at (4.1,-0.5) {};
    \node (t6) at (4.5,-0.5) {};

    \node (r1) at (0.3,1.5) {};
    \node (r2) at (0.7,1.5) {};
    
    \node (r3) at (1.5,1.5) {};
    \node (r4) at (1.9,1.5) {};
    
    \node (r5) at (4.1,1.5) {};
    \node (r6) at (4.5,1.5) {};


    \draw[line width=0.2pt,->,dotted] (r2) to [out=30,in=150]  (r3);
    \draw[line width=0.2pt,->,dotted] (r4) to [out=30,in=150]  (r5);

   \draw[line width=0.2pt,->,dotted] (p3) to [out=250,in=90]  (r1);
   \draw[line width=0.2pt,->,dotted] (r6) to [out=90,in=120]  (q2);

    \draw[line width=0.2pt,->,dotted] (end) -- (out);

  \draw[line width=0.2pt,->] (start) -- (p1);
    \draw[line width=0.2pt,->] (p1) -- (p2);
  \draw[line width=0.2pt,->,dotted] (p2) -- (p3);
  \draw[line width=0.2pt,->] (q1) -- (end);
  \draw[line width=0.2pt,->,dotted] (q2) to [out=300,in=90] (q1);

 \draw[line width=0.2pt,->,dotted] (s1) -- (t1);
 \draw[line width=0.2pt,->,dotted] (t2) -- (s2);
 \draw[line width=0.2pt,->,dotted] (s3) -- (t3);
 \draw[line width=0.2pt,->,dotted] (t4) -- (s4);
 \draw[line width=0.2pt,->,dotted] (s5) -- (t5);
 \draw[line width=0.2pt,->,dotted] (t6) -- (s6);  

  \draw[line width=0.2pt,->,dotted] (r1) -- (s1);
 \draw[line width=0.2pt,->,dotted] (s2) -- (r2);
 \draw[line width=0.2pt,->,dotted] (r3) -- (s3);
 \draw[line width=0.2pt,->,dotted] (s4) -- (r4);
 \draw[line width=0.2pt,->,dotted] (r5) -- (s5);
 \draw[line width=0.2pt,->,dotted] (s6) -- (r6);  


\end{tikzpicture}
%
    
    
    
    \caption{A finite $\xrightarrow{X}$-path from $\sigma~\psi_1$ to $\sigma~\psi_2$ may visit other tableau prefixes, and an infinite $\xrightarrow{X}$-path from $\sigma~\psi_1$ may include finite segments that visit other prefixes, before continuing with an infinite path from a formula $\psi_2$. Each square area represents a tableau prefix, or a graph $g \in G^t$. For $S$ the set of pairs of formulas that appear near the bottom, the left figure illustrates the relation $(\psi_1,\psi_2) \xrightarrow{g} S$ and the right one the relation $\psi_1 \xrightarrow{g} S, \psi_2$ (Definition \ref{def:searchTtoS}).}
    \label{fig:finite-paths}
\end{figure}

\begin{defi}\label{def:searchTtoS}
Given a graph $g = (\Phi,(\xrightarrow{X})_{X \in mV},\ell) \in \G$, $T =(\psi_1,\psi_2) \in (\sub(\varphi))^2$, $S \subseteq (\sub(\varphi))^2$, and $X \in mV$, we write
$T \xrightarrow{g,X} S$ to mean that 
\begin{enumerate}
	\item 
$\psi_1,\psi_2 \in \Phi$, $(in,top) \in \ell(\psi_1)$, $(out,top) \in \ell(\psi_2)$, 
\item 
for each $(\psi_1',\psi_2') \in S$, 
$\psi_1',\psi_2' \in \Phi$, $(out,bot) \in \ell(\psi_1')$, $(in,bot) \in \ell(\psi_2')$, 
and 
\item 
$\{ \psi_1',\psi_2', X \mid (\psi_1',\psi_2') \in S \}$ 
are on a path from $\psi_1$ to $\psi_2$ in 
$(\Phi,\xrightarrow{X} \cup \{(\psi_2,\psi_1)\} \cup S)$.
\end{enumerate}
We write
$\emptyset \xrightarrow{g,X} S$ to mean that Condition 2 above is satisfied, and 
$\{ \psi_1',\psi_2', X \mid (\psi_1',\psi_2') \in S \}$ is part of a 
cycle
in 
$(\Phi,\xrightarrow{X} \cup  S)$.
We can similarly define 
$T \xrightarrow{g} S$  and $\emptyset \xrightarrow{g} S$ by not requiring that $X$ is part of the path.

We 
write
$\psi_1 \xrightarrow{g,X} S,\psi_2$ to mean that 
\begin{enumerate}
	\item 
$\psi_1,\psi_2 \in \Phi$, 
$(out,bot) \in \ell(\psi_2)$, 
\item 
for each $(\psi_1',\psi_2') \in S$, 
$\psi_1',\psi_2' \in \Phi$, $(out,bot) \in \ell(\psi_1')$, $(in,bot) \in \ell(\psi_2')$, 
and 
\item 
$\{ \psi_1',\psi_2', X \mid (\psi_1',\psi_2') \in S \}$ 
are on a path from $\psi_1$ to $\psi_2$ in 
$(\Phi,\xrightarrow{X} \cup \{(\psi_2,\psi_1)\} \cup S)$.
\end{enumerate}
We can similarly define 
$\psi_1 \xrightarrow{g} S, \psi_2$  by not requiring that $X$ is part of path.
\end{defi}

To describe that there are no $\xrightarrow{X}$-paths where $X$ appears infinitely often, we first describe that \emph{there exists} some 
$\xrightarrow{X}$-path where $X$ appears infinitely often.
We do so, by defining the following formulas, recursively on two parameters $\varpi$ and $\varpi_X$, that are 
\changeJ{\emph{strings from}} 
$\{\emptyset\} \cup \sub(\varphi)^2$ or from $\sub(\varphi)$, \emph{i.e.}~$(\{\emptyset\} \cup \sub(\varphi)^2)^* \cup (\sub(\varphi))^*$.
These strings keep track of the path from the root of the syntax tree of the constructed formula, to the current subformula, allowing the construction to detect when to use a recursive call.
For $T \in \sub(\varphi)^2$ or $T = \emptyset$:

\begingroup
\allowdisplaybreaks
\begin{align*}
	\finb^X(T,\varpi) &:= 
	\begin{cases}
 \displaystyle
		\bigvee_{\emptyset \xrightarrow{g,X} \emptyset } g 
		\lor 
		\bigvee_{\emptyset \centernot{\xrightarrow{g,X}} \emptyset } 
		(g \land 
		\finb^X(\emptyset,\varepsilon,g))
		& \text{ if $T = \emptyset$}
		\\
  \displaystyle
		Z^X_T &\text{if $T$ appears in }\varpi \\
  \displaystyle
		\mn Z_T^X. 
	\bigvee_{\substack{T \in g\\ T \xrightarrow{g,X} \emptyset }} g 
	\lor 
	\bigvee_{\substack{T \in g\\ T \centernot{\xrightarrow{g,X}} \emptyset }} 
	(g \land 
	\finb^X(T,\varpi,g)),
	&\text{ otherwise}
	\end{cases}\\
	\finb(T,\varpi) &:= 
	\begin{cases}
		Z_T &\text{ if $T$ appears in }\varpi \\
  \displaystyle
		\mn Z_T. 
	\bigvee_{\substack{T \in g\\ T \xrightarrow{g} \emptyset }} g 
	\lor 
	\bigvee_{\substack{T \in g\\ T \centernot{\xrightarrow{g}} \emptyset }} 
	(g \land 
	\finb(T,\varpi,g)),
	&\text{ otherwise}
	\end{cases}
	\\
	\finb^X(T,\varpi,g) &:=
		\bigvee_{\substack{S \subseteq \Phi_g^2\\
						T \xrightarrow{g} S}}
						\bigvee_{s \in S}
						\diam{\al(s)} \finb^X(nx(s),\varpi T)
						\land 
						\bigwedge_{s \in S} 
						\diam{\al(s)} \finb(nx(s),\varepsilon)
			\\
			&	\qquad\qquad	\lor
			\bigvee_{\substack{S \subseteq \Phi_g^2\\
			T \xrightarrow{g,X} S}}
			\bigwedge_{s \in S} 
			\diam{\al(s)} \finb(nx(s),\varepsilon)
	\\
	\finb(T,\varpi,g) &:=
			\bigvee_{\substack{S \subseteq \Phi_g^2\\
			T \xrightarrow{g} S}}
			\bigwedge_{s \in S} 
			\diam{\al(s)} \finb(nx(s),\varpi T)
\end{align*}
\endgroup

\changeJ{
The formula $\finb((\psi_1,\psi_2),\varpi)$ (\resp $\finb^X((\psi_1,\psi_2),\varpi)$) describes that there is a finite $\xrightarrow{X}$-path from $\psi_1$ to $\psi_2$ that does not visit prefixes smaller than the current one (\resp where $X$ appears).
Formula $\finb(\emptyset,\varpi)$ describes that there is a finite cycle that does not visit prefixes smaller than the current one.
We note an asymmetry between $\al(\diam{\al}\psi_1,\psi_2)$ and $\al([\al]\psi_1,\psi_2)$, used in the above formulas. That is the case, because in a tableau branch, $\sigma~[\al]\psi$ has a dependency edge to all $\sigma.\al\diam{\psi'}~\psi$, but $\sigma~\diam{\al}\psi$ only has a dependency edge to $\sigma.\al\diam{\psi}~\psi$.
}

\changeJ{
It is also worth noting that the way these formulas are constructed, we do not need to consider any strings $\varpi$ or $\varpi_X$ where the same element appears twice. Therefore, in effect, $\varpi$ and $\varpi_X$ are permutations of subsets of $\{\emptyset\} \cup \sub(\varphi)^2$ and $\sub(\varphi)$.
}

\begin{lem}\label{lem:finb-equivalences-variables}
	For every $T,T' \in \sub(\varphi)^2$,
	\[
	\trueset{\finb^X(T,\varepsilon)}  = 
	 \trueset{\finb^X(T,T')[\finb^X(T',\varepsilon)/Z^X_{T'}]} 
	= 
	\trueset{\finb^X(T,T'),\rho\left[Z^X_{T'}\mapsto\trueset{\finb^X(T',\varepsilon)}\right]} 
	.\]
\end{lem}

\begin{proof}
	Specifically, we prove the stronger statement, that for any 
 string
 $\varpi$ where $T'$ does not appear, 
	\[
	\trueset{\finb^X(T,\varpi), \rho}  = 
	 \trueset{\finb^X(T,T'\varpi)[\finb^X(T',\varpi)/Z^X_{T'}], \rho} 
	\]
	and 
	\[
		\trueset{\finb^X(T,T'\varpi)[\finb^X(T',\varpi)/Z^X_{T'}], \rho} 
		= 
		\trueset{\finb^X(T,T'\varpi),\rho\left[Z^X_{T'}\mapsto\trueset{\finb^X(T',\varpi)}\right], \rho} 
		.\]	
	The second equality is a straightforward property of fixed-point operators. We prove the first one by induction on 
	$\finb^X(T,T'\varpi)$.
	\begin{description}
		\item[If $\finb^X(T,T'\varpi) = Z^X_T$] then $T$ appears in $T'\varpi$. If $T = T'$, then $$\finb^X(T,\varpi) = \finb^X(T',\varpi) = Z^X_T[\finb^X(T',\varpi)/Z^X_{T'}] = \finb^X(T,T'\varpi)[\finb^X(T',\varpi)/Z^X_{T'}],$$ yielding the required semantic equality.
		On the other hand, if $T \neq T'$, then $T$ appears in $\varpi$, and therefore 
		$\finb^X(T,\varpi) = 
		\finb^X(T,T'\varpi)[\finb^X(T',\varpi)/Z^X_{T'}] = Z^X_T$.
		\item[Otherwise] the inductive step is straightforward from the definition of the formulas.
		\qedhere 
	\end{description}
\end{proof}

\changeJ{
The following Definition \ref{def:cyc} defines the predicates $cyc((\psi_1,\psi_2),\sigma)$ and $cyc_X((\psi_1,\psi_2),\sigma)$, which are the counterparts of $\finb(T,\varpi)$ and  $\finb^X(T,\varpi)$ for tableau branches.
}

\begin{defi}\label{def:cyc}
	Let $X \in mV$, $\psi_1,\psi_2\in \sub(\varphi)$,
	$f:b \xrightarrow{r} \M,s$,
	and $\sigma \in \sigma(b)$.
	We write $cyc((\psi_1,\psi_2),\sigma)$ to denote that the following hold:
	\begin{enumerate}
		\item $\sigma ~ \psi_1, \sigma~ \psi_2 \in b$,
		\item there is a finite $\xrightarrow{X}$-path in $b$, from $\sigma~\psi_1$ to $\sigma~\psi_2$, 
            which only goes through prefixes that extend $\sigma$.
	\end{enumerate}
	We write $cyc_X((\psi_1,\psi_2),\sigma)$ to denote that the same conditions as for $cyc((\psi_1,\psi_2),\sigma)$ are satisfied, with the addition that $X$ appears on the finite $\xrightarrow{X}$-path from $\sigma~\psi_1$ to $\sigma~\psi_2$.
\end{defi}

\begin{lem}\label{lem:finb-correct}
	Let $X \in mV$, $\psi_1,\psi_2\in \sub(\varphi)$, 
	$f:b \xrightarrow{r} \M,s$,
	and $\sigma \in \sigma(b)$.
	\begin{enumerate}
		\item $(\M,f(\sigma)) \models \finb((\psi_1,\psi_2),\varepsilon)$  if and only if 
		$cyc((\psi_1,\psi_2),\sigma)$ and $(\psi_1,\psi_2) \in f(\sigma)$, 
		\item 
		$(\M,f(\sigma)) \models \finb^X((\psi_1,\psi_2),\varepsilon)$ if and only if 
		$cyc_X((\psi_1,\psi_2),\sigma)$ and $(\psi_1,\psi_2) \in f(\sigma)$,
		\item $(\M,f(\sigma) )\models \finb(\emptyset,\varepsilon)$  if and only if 
		$cyc((\psi,\psi),\sigma)$ for some $\psi \in \Phi_{f(\sigma)}$, and
		\item 
		$(\M,f(\sigma)) \models \finb^X(\emptyset,\varepsilon)$ if and only if 
		$cyc_X((\psi,\psi),\sigma)$ for some $\psi \in \Phi_{f(\sigma)}$.
	\end{enumerate}
\end{lem}

\begin{proof}
	We use the first biimplication of the lemma to prove the second 
	one.
	All other cases can be proven with similar reasoning.
	We first observe that
        $(\M,f(\sigma)) \models \finb^X((\psi_1,\psi_2),\varepsilon)$
	implies  $(\psi_1,\psi_2) \in f(\sigma)$, and 
        prove that  
	$(\M,f(\sigma)) \models \finb^X((\psi_1,\psi_2),\varepsilon)$
	implies  
	$cyc_X((\psi_1,\psi_2),\sigma)$.

	For each $T$, let $C_{Z_{T}^X} = \{ f(\sigma')  \mid \sigma' \in \sigma(b) \text{ and }  cyc_X(T,\sigma') \}$.
	Let $T = (\psi_1,\psi_2)$.
	To prove the implication, we proceed to prove the more general statement that 
	$\neg cyc_X(T,\sigma)$ implies that 
	$(\M,f(\sigma)) \not\models_\rho \finb^X(T,\varpi)$
	for every 
	$\varpi \in (\sub(\varphi)^2)^*
    $ and environment $\rho$, such that $\rho(Z_{T'}^X) = C_{Z_{T'}^X}$ for each $Z_{T'}^X$. 
	We assume that 
	$\neg cyc_X(T,\sigma)$ and proceed by induction on the structure of $\finb^X(T,\varpi)$.
		\textbf{If $T$ appears in $\varpi$} or, equivalently, 
	$\finb^X(T,\varpi) = Z_T^X$, by the definition of $\rho$, $f(\sigma) \notin \rho(Z^X_T)$, so 
	$(\M,f(\sigma)) \not\models_\rho \finb^X(T,\varpi)$.
	\textbf{We now assume that $T$
	does not appear in $\varpi$}
	If 
    Condition 1
    for $cyc_X(T,\sigma)$ does not hold, then 
	it is evident from the definitions of $T \xrightarrow{g} S$, $T \xrightarrow{g,X} S$, and $T \xrightarrow{g,X} \emptyset$ that $(\M,f(\sigma)) \not\models \finb^X(T,\varpi)$, and we are done. 
	We now assume that Condition 1 
    holds, but there is no finite 
	$\xrightarrow{X}$-path in $b$, from $\sigma~\psi_1$ to $\sigma~\psi_2$, that only visits extensions of $\sigma$, where $X$ appears.
	From the fixed-point semantics of Table \ref{table:semantics}, it suffices to prove two statements:
	\begin{enumerate}
		\item for every $\sigma' \in \sigma(b)$, if 
		$$(\M, f(\sigma')) \models \bigvee_{\substack{T \in g\\ T \xrightarrow{g,X} \emptyset }} g 
		\lor 
		\bigvee_{\substack{T \in g\\ T \centernot{\xrightarrow{g,X}} \emptyset }} 
		(g \land 
		\finb^X(T,\varpi,g)),$$ then 
		$f(\sigma') \in C_{Z_T^X}$; and 
		\item $f(\sigma) \notin C_{Z_T^X}$. 
	\end{enumerate}
	The second statement is immediate from our assumption that $\neg cyc_X(T,\sigma)$.
	We now prove the first statement.
	Clearly, if $(\M,f(\sigma'))  \models \bigvee_{\substack{T \in g\\ T \xrightarrow{g,X} \emptyset }} g $, we have $cyc_X(T,\sigma')$.
	Otherwise, we have that there is some $g \in G^t$, such that 
    $T \in g$,  
		$(\M,f(\sigma'))  \models g$, and $(\M,f(\sigma') ) \models \finb^X(T,\varpi,g)$.
		Therefore, there is some $S \subseteq \sub(\varphi)^2$, such that either 
		\begin{description}
			\item[$T \xrightarrow{g,X} S$] and for each $s \in S$, there is some 
            $\varsigma = \al(s)$ or $\varsigma = \al(s)\diam{\psi}$, such that 
            $\sigma'.\varsigma \in \sigma(b)$, 
            $f(\sigma') R_{\varsigma} f(\sigma'.\varsigma)$ and $(\M, f(\sigma'.\varsigma)) \models \finb(nx(s),\varepsilon)$, or 
			\item[$T \xrightarrow{g} S$] and for each $s \in S$, there is some 
            $\varsigma = \al(s)$ or $\varsigma = \al(s)\diam{\psi}$, 
            such that 
            $\sigma'.\varsigma \in \sigma(b)$, 
            $f(\sigma') R_{\varsigma} f(\sigma'.\varsigma)$ and $(\M, f(\sigma'.\varsigma)) \models \finb(nx(s),\varepsilon)$ and 
			for some
			$s \in S$, there is some 
            $\varsigma = \al(s)$ or $\varsigma = \al(s)\diam{\psi}$, 
            such that $\sigma'.\varsigma \in \sigma(b)$, 
            $f(\sigma') R_{\varsigma} f(\sigma'.\varsigma)$ and $(\M, f(\sigma'.\varsigma)) \models \finb^X(nx(s),\varpi T)$.
		\end{description}
        In the first case, 
        for each $s =(\chi_1,\chi_2) \in S$, where $nx(s) = (\chi_1',\chi_2')$, by the 
        first item of the lemma, 
        there is a
		finite 
		$\xrightarrow{X}$-path in $b$, from 
        $\sigma'.\varsigma~\chi_1'$ to $\sigma'.\varsigma~\chi_2'$, that only visits extensions of $\sigma'.\varsigma$.
        We know that $\chi_2' = [\al]\chi_2''$, and 
        either $\chi_1 = [\al]\chi_1'$ or $\chi_1 = \diam{\al}\chi_1'$.
        By the conditions of $\xrightarrow{X}$ for box formulas, 
        $\sigma'.\varsigma~\chi_2' \xrightarrow{X} \sigma'~\chi_2$, and if $\chi_1 = [\al]\chi_1'$, then 
        $\sigma'~\chi_1 \xrightarrow{X} \sigma'.\varsigma~\chi_1'$.
        If $\chi_1 = \diam{\al}\chi_1'$, then we know that ($\varsigma = $) $\al(s)=\al\diam{\chi_1'}$, and therefore 
        $\frac{\sigma'~\chi_1}{\sigma'.\al\diam{\chi_1'}~\chi_1'}$ is a tableau rule, and 
        by the conditions of $\xrightarrow{X}$, 
        $\sigma'~\chi_1 \xrightarrow{X} \sigma'.\al(s)~\chi_1'$.
        Therefore, 
        there is a
		finite 
		$\xrightarrow{X}$-path in $b$, from 
        $\sigma'.~\chi_1$ to $\sigma'.~\chi_2$ that only visits extensions of $\sigma'$, where $\sigma~X$ appears.

        In the second case, 
        we can similarly conclude that 
        there is a
		finite 
		$\xrightarrow{X}$-path in $b$, from 
        $\sigma'.~\chi_1$ to $\sigma'.~\chi_2$ that only visits extensions of $\sigma'$.
        From $f(\sigma') R_{\varsigma} f(\sigma'.\varsigma)$ and $(\M, f(\sigma'.\varsigma)) \models \finb^X(nx(s),\varpi T)$ and the inductive hypothesis, we see that some $\sigma''~X$ occurs on that path.
        %
		Therefore, 
        both 
        conditions for $cyc_X(T,\sigma')$ are satisfied and 
		$f(\sigma') \in C_{Z_T^X}$, completing the proof of the first statement and the induction.

	We now prove that $cyc_X((\psi_1,\psi_2),\sigma)$ implies 
	$(\M,f(\sigma)) \models \finb^X((\psi_1,\psi_2),\varepsilon)$. 
	By Condition 
    2
    from $cyc_X((\psi_1,\psi_2),\sigma)$, there is a finite $\xrightarrow{X}$-path from $\psi_1$ to $\psi_2$ that only visits prefixes that extend $\sigma$, including $\sigma$.
	Let $k$ be the length of the longest $\sigma'$, such that the path goes through $\sigma.\sigma'$.
	We proceed by induction on $k$. Let $T = (\psi_1,\psi_2).$
		\textbf{If $k=0$} then the finite path remains in $\sigma$, and therefore $(\psi_1,\psi_2) \xrightarrow{g(\sigma),X} \emptyset$. As $(\M,f(\sigma)) \models g(\sigma)$, we conclude that $(\M,f(\sigma)) \models_\rho \finb^X((\psi_1,\psi_2),\varepsilon)$. 
		\textbf{If $k>0$} we demonstrate that 
        $$(\M,f(\sigma)) \models \finb^X((\psi_1,\psi_2),\varepsilon)
        .$$
		Let $p$ be the finite $\xrightarrow{X}$-path from $\psi_1$ to $\psi_2$. Since $k>0$, $p$ will leave $\sigma$ and return, perhaps multiple times. As $p$ leaves and returns to $\sigma$ $r$ times, in $\sigma$ it will visit exactly the formulas $\sigma~\chi_1,\sigma~\chi_2,\ldots,\sigma~\chi_{2r} \in b$, in that order.
		We note that for every $1 \leq i \leq r$, $(out,bot) \in \ell_{g(\sigma)}(\chi_{2i-1})$ and $(in,bot) \in \ell_{g(\sigma)}(\chi_{2i})$.
		Let $S = \{ s_i \mid  1 \leq i \leq r \}$, where $s_i= (\chi_{2i-1},\chi_{2i})$ for every $i$; we see that $(\psi_1,\psi_2) \xrightarrow{g(\sigma)} S$.
		Furthermore, if $p$ goes through $\sigma~X$, then $(\psi_1,\psi_2) \xrightarrow{g(\sigma),X} S$, and otherwise $p$ goes through some $\sigma'~X$ between some $\sigma~\chi_{2j-1}$ and $\sigma~\chi_{2j}$.
		For every $i \leq r$, 
        let $\varsigma_i = \al(s_i)$ or $\al(s_i)\diam{\psi_i'}$,
        depending on whether $\chi_{2i-1}$ is a box or a diamond formula; 
        there are some $\sigma.\varsigma_i~\chi'_{2i-1} \in b$ and $\sigma.\varsigma_i~\chi'_{2i} \in b$, such that 
		$\sigma~\chi_{2i-1} \xrightarrow{X} \sigma.\varsigma_i~\chi'_{2i}$ and $\sigma.\varsigma_i~\chi'_{2i} \xrightarrow{X} \sigma~\chi_{2i}$.
		Then, the part of $p$ from  $\sigma.\varsigma_i~\chi'_{2i-1}$ to $\sigma.\varsigma_i~\chi'_{2i}$ is a finite $\xrightarrow{X}$-path that only goes through extensions of $\sigma.\varsigma_i$ that extend the prefix by a length at most $k-1$.
		Therefore, $cyc(nx(s_i),\sigma.\varsigma_i)$ 
        and by the inductive hypothesis, 
		$(\M,f(\sigma.\varsigma_i)) \models \finb(nx(s_i),\varepsilon)$, and if it is not the case that $(\psi_1,\psi_2) \xrightarrow{g(\sigma),X} S$, 
        then
        $(\M,f(\sigma.\varsigma_j)) \models \finb^X(nx(s_j),\varepsilon)$.
		Using Lemma \ref{lem:finb-equivalences-variables}, we see that
        $$(\M,f(\sigma.\varsigma_j)) \models \finb^X(nx(s_j),(\psi_1,\psi_2))[\finb^X((\psi_1,\psi_2),\varepsilon)/Z^X_T].$$
        From  the fact that $f(\sigma) R_{\al(s_i)} f(\sigma.\varsigma_i)$, we conclude that 
		$$(\M,f(\sigma)) \models \finb^X((\psi_1,\psi_2),\varepsilon,g(\sigma))[\finb^X((\psi_1,\psi_2),\varepsilon)/Z^X_T].$$
        As $\finb^X((\psi_1,\psi_2),\varepsilon$ is a fixed-point formula, 
  this yields that 
		\begin{align*}
			(\M,f(\sigma)) \models \finb^X((\psi_1,\psi_2),\varepsilon).
			\tag*{\qedhere}
		\end{align*}
\end{proof}

We now express that there is an infinite path in $b$, where $X$ appears infinitely often, starting from $\sigma~\psi$.
\begingroup
\allowdisplaybreaks
\begin{align*}
	\infb^X(\psi,\varpi,\varpi_X) := &
	\begin{cases}
		Z_\psi^X &\text{ if $\psi$ appears in }\varpi_X \\
		Z_\psi &\text{ if }\psi  \in \varpi \text{ and }\not\in \varpi_X \\
	\mx Z_\psi.
	\mn Z_\psi^X.
	\displaystyle
	\bigvee_{
		\psi 	
	\in \Phi_g} g \land 
	\infb^X(\psi,\varpi,\varpi_X,g)
	&\text{ otherwise}
	\end{cases}
 \\ 
	\infb^X(\psi,\varpi,\varpi_X,g) :=&
		\bigvee_{
			\psi \xrightarrow{g,X} S,\psi'
			}
						\bigwedge_{s \in S} 
						\diam{\al(s)} \finb(nx(s),\varepsilon)
						\land 
						\diam{\al(\psi')} \infb^X(nx(\psi'),\varpi \psi,\varepsilon)
			\\
			&\lor
			\\
			&\bigvee_{
			\psi \xrightarrow{g} S,\psi'
			}
			\left(
			\begin{array}{l}
			\left(
			\begin{array}{l}\displaystyle
						\bigvee_{s \in S}
						\diam{\al(s)} \finb^X(nx(s),\varepsilon)
						\\ 
						\qquad 
						\land 
						\\
						\displaystyle
						\bigwedge_{s \in S} 
						\diam{\al(s)} \finb(nx(s),\varepsilon)
						\land 
						\diam{\al(\psi')} \infb^X(nx(\psi'),\varpi \psi,\varepsilon)
			\end{array}
			\right)
			\\
			\qquad \qquad \qquad \lor
			 \\
			\displaystyle 
			\bigwedge_{s \in S} 
						\diam{\al(s)} \finb(nx(s),\varepsilon)
						\land 
						\diam{\al(s)} \infb^X(nx(\psi'),\varpi \psi,\varpi_X\psi)
					\end{array}
					\right)
	\\
	\logicize{inf\_path}_X := &
		\finb^X(\emptyset,\epsilon)
		\lor 
		\bigvee_{\psi \in g}
	\left(
		g \land  
		\infb^X(\psi,\epsilon,\epsilon)
	\right)
	\\
	\logicize{inf\_path} := & \bigvee_{X \in mV}
	\logicize{inf\_path}_X 
\end{align*}
\endgroup 

\changeJ{
The formula $\infb^X(\psi,\varpi,\varpi_X,g)$ describes that there is an \emph{infinite} path in the branch that does not visit prefixes smaller than the current one, where $X$ appears infinitely often, according to the intuition in Figure \ref{fig:finite-paths}. We use the formulas $\finb$ to express that two formulas in the current prefix are connected by a finite path (that does not revisit the current prefix), and we use $\pi_X$ to keep track of whether we have encountered $X$ along the path, after the latest recursive call of $\infb^X(\psi,\varpi,\varpi_X,g)$. 
In the definition of $\infb^X(\psi,\varpi,\varpi_X,g)$, if $\psi$ appears in $\varpi_X$, then we activate the least-fixed-point operator, 
because $\psi$ appearing in $\varpi_X$ marks that $X$ was not encountered since the last activation of $\infb^X(\psi,-,-,g)$,
and otherwise the greatest-fixed-point operator. 
We do that to ensure that as the formula evaluates along a $\xrightarrow{X}$-path, $X$ is encountered infinitely often, because if it is not, then only the least-fixed-point operator would activate infinitely often.
}

\begin{lem}\label{lem:infb-equivalences-variables}
	For every $\psi,\psi' \in \sub(\varphi)$ and $\varpi \in (\sub(\varphi))^*$, such that $\psi'$ does not appear in $\varpi$,
	\[
	\trueset{\infb^X(\psi,\varpi,\varepsilon)}  = 
	 \trueset{\infb^X(\psi,\varpi\psi',\psi')[\infb^X(\psi',\varpi,\varepsilon)/Z^X_{\psi'}]} \]

  \[= 
	\trueset{\infb^X(\psi,\varpi\psi',\psi'),\rho\left[Z^X_{\psi'}\mapsto\trueset{\infb^X(\psi',\varpi,\varepsilon)}\right]} 
	.\]
\end{lem}

\begin{proof}
	The reasoning to prove this lemma is as in the proof of Lemma \ref{lem:finb-equivalences-variables}.
\end{proof}

\begin{lem}\label{lem:infb-correct}
	Let $X \in mV$, 
	$f:b \xrightarrow{r} \M,s$,
	and $\sigma \in \sigma(b)$. 
        Then
	$(\M,f(\sigma)) \models \infb^X(\psi,\varepsilon,\varepsilon)$ 
	if and only if 
	\begin{enumerate}
		\item $\sigma ~ \psi \in b$, and
		\item there is an infinite $\xrightarrow{X}$-path in $b$, from $\sigma~\psi$, where $X$ appears infinitely often, that only visits prefixes that extend $\sigma$.
	\end{enumerate}
\end{lem}

\begin{proof}
	We first assume that $\sigma ~ \psi \in b$,
	and that there is an infinite $\xrightarrow{X}$-path in $b$, from $\sigma~\psi$, where $X$ appears infinitely often, that only visits prefixes that extend $\sigma$, and prove that 
	$(\M ,f(\sigma))\models \infb^X(\psi,\varepsilon,\varepsilon)$.
	For every $\psi' \in \sub(\varphi)$, let $IP_{\psi'}$ be the set of prefixes $\sigma' \in \sigma(b)$, such that 
	$\sigma' ~ \psi' \in b$,
	and that there is an infinite $\xrightarrow{X}$-path in $b$, from $\sigma'~\psi'$, where $X$ appears infinitely often, that only visits prefixes that extend $\sigma'$.
	To prove the implication, we prove the stronger statement that for every $\varpi$, 
	$(\M ,f(\sigma))\models_\rho \infb^X(\psi,\varpi,\varepsilon)$, where for every $\psi' \in \sub(\varphi)$, $\rho(Z_{\psi'}) = \{f(\sigma') \mid \sigma' \in IP_{\psi'}\}$. 
	We proceed by induction on the structure of the formula.

	If $\psi$ appears in $\varpi$, 
	then $\trueset{\infb^X(\psi,\varpi,\varepsilon),\rho} = \trueset{Z_\psi} = \{f(\sigma') \mid \sigma' \in IP_{\psi}\} \ni f(\sigma)$, and we are done.
	Otherwise, using the semantics of Table \ref{table:semantics}, since 
	$$f(\sigma) \in \{f(\sigma') \mid \sigma' \in IP_{\psi'}\} = \rho(Z_\psi),$$ it suffices to prove that 
	\[
		\{f(\sigma') \mid \sigma' \in IP_{\psi'}\} \subseteq 	
		\trueset{\mn Z_\psi^X. \bigvee_{\psi \in \Phi_g} \infb^X(\psi,\varpi,\varepsilon,g),\rho}.
	\]
	Let $\sigma' \in IP_{\psi'}$. 
	There must be an infinite $\xrightarrow{X}$-path in $b$ from $\sigma'~\psi$ that visits $X$ infinitely often, and 
	avoids strict prefixes of $\sigma$.
	That path may visit certain prefixes more than once.
	Let $\sigma'.\sigma''$ be minimal, such that the path goes through some $\sigma'.\sigma''~\psi' \in b$ after visiting $X$ at least once after $\sigma'~\psi$.
	In particular, this means that $X$ is visited only in extensions of $\sigma'.\sigma''$.
	Furthermore, since the path is infinite, we can assume that 
	the path visits a strict extension of $\sigma'.\sigma''$ right after 
	$\sigma'.\sigma''~\psi'$. 
	Let $k$ be the length of $\sigma''$.
	
 We prove that $(\M,f(\sigma')) \models_\rho \mn Z_\psi^X. \bigvee_{\psi \in \Phi_g} \infb^X(\psi,\varpi,\varepsilon,g)$ and we do so by induction on~$k$.
\begin{description}
    \item[Case of $k=0$]
		in this case, the path starts from $\sigma'~\psi$ and visits $\sigma'~\psi'$ after visiting $X$ at least once.
		Therefore, since we have assumed that 
		the path visits a strict extension of $\sigma'$ right after 
		$\sigma'~\psi'$, 
		we have that for some $S \subseteq \sub^2(\varphi)$,
		where 
		$cyc(s,\sigma')$ through a path that only visits extensions of some $\sigma'.a_s$, for each $s \in S$,
		$\psi \xrightarrow{g,X} S \psi'$, 
		or $\psi \xrightarrow{g} S \psi'$ 
		and 
		$cyc_X(s_X,\sigma')$ through a path that only visits extensions of $\sigma'.a_{s_X}$, for some $s_X \in S$.
		For each $s\in S$, notice that $f(\sigma') R_{\al(s)} f(\sigma'.a_s)$, and that $cyc(nx(s),\sigma'.a_s)$ and $cyc_X(nx(s_X),\sigma'.a_{s_X})$.
		Therefore, for $g = g(\sigma')$, using  Lemma \ref{lem:finb-correct} and the inductive hypothesis for $\infb^X(nx(\psi'),\varpi \psi,\varepsilon)$, 
		either $\psi \xrightarrow{g,X} S \psi'$ and 
		\[
			(\M,f(\sigma')) \models_\rho \bigwedge_{s \in S} 
			\diam{\al(s)} \finb(nx(s),\varepsilon)
			\land 
			\diam{\al(\psi')} \infb^X(nx(\psi'),\varpi \psi,\varepsilon);
		\]
		or $\psi \xrightarrow{g} S \psi'$, 
		\begin{align*}
			(\M,f(\sigma')) \models_\rho& \diam{\al(s_X)} \finb^X(nx(s_X),\varepsilon); \text{ and}\\ 
			(\M,f(\sigma')) \models_\rho& \bigwedge_{s \in S} 
			\diam{\al(s)} \finb(nx(s),\varepsilon)
			\land 
			\diam{\al(\psi')} \infb^X(nx(\psi'),\varpi \psi,\varepsilon) . 
		\end{align*}
		For both cases, we 
		observe that $Z_\psi^X$ is not free in $\infb^X(nx(\psi'),\varpi \psi,\varepsilon)$, and therefore 
		either $\psi \xrightarrow{g,X} S \psi'$ and 
		\[
			(\M,f(\sigma')) \models_{\rho'} \bigwedge_{s \in S} 
			\diam{\al(s)} \finb(nx(s),\varepsilon)
			\land 
			\diam{\al(\psi')} \infb^X(nx(\psi'),\varpi \psi,\varepsilon);
		\]
		or $\psi \xrightarrow{g} S \psi'$, 
		\begin{align*}
			(\M,f(\sigma')) \models_{\rho'}& \diam{\al(s_X)} \finb^X(nx(s_X),\varepsilon); \text{ and}\\ 
			(\M,f(\sigma')) \models_{\rho'}& \bigwedge_{s \in S} 
			\diam{\al(s)} \finb(nx(s),\varepsilon)
			\land 
			\diam{\al(\psi')} \infb^X(nx(\psi'),\varpi \psi,\varepsilon) , 
		\end{align*}
		where $\displaystyle \rho' = \rho\left[Z_\psi^X \mapsto \trueset{\mn Z_\psi^X. \bigvee_{\psi \in \Phi_g} \infb^X(\psi,\varpi,\varepsilon,g)}\right]$.
		We 
		conclude therefore that 
		$(\M,f(\sigma')) \models_{\rho'} \bigvee_{\psi \in \Phi_g} \infb^X(\psi,\varpi,\varepsilon,g)$, and using the fact that \\ $\mn Z_\psi^X. \bigvee_{\psi \in \Phi_g} \infb^X(\psi,\varpi,\varepsilon,g)$ is a fixed point formula, we conclude that
		\[(\M,f(\sigma')) \models_\rho \mn Z_\psi^X. \bigvee_{\psi \in \Phi_g} \infb^X(\psi,\varpi,\varepsilon,g).\]
    \item[Case of $k>0$]
		let $\sigma'~\psi'$ be the last formula that the path visits in $\sigma'$.
		Therefore, 
		the path visits a strict extension of $\sigma'$ right after 
		$\sigma'~\psi'$, and 
		we have that for some $S \subseteq \sub^2(\varphi)$,
		where 
		$cyc(s,\sigma')$ through a path that only visits extensions of some $\sigma'.a_s$, for each $s \in S$,
		$\psi \xrightarrow{g} S \psi'$. 
		For each $s\in S$, notice that $f(\sigma') R_{\al(s)} f(\sigma'.a_s)$.
		Therefore, for $g = g(\sigma')$, using  Lemma \ref{lem:finb-correct} and the inductive hypothesis for $k$, 
		$\psi \xrightarrow{g} S \psi'$ and 
		\[
			(\M,f(\sigma')) \models_\rho \bigwedge_{s \in S} 
			\diam{\al(s)} \finb(nx(s),\varepsilon)
			\land 
			\diam{\al(\psi')} \infb^X(nx(\psi'),\varpi,\varepsilon).
		\]
		From this, Lemma \ref{lem:infb-equivalences-variables} yields that 
		\[
			(\M,f(\sigma')) \models_\rho \bigwedge_{s \in S} 
			\diam{\al(s)} \finb(nx(s),\varepsilon)
			\land 
			\diam{\al(\psi')} \infb^X(nx(\psi'),\varpi \psi,\psi)[\infb^X(nx(\psi),\varpi,\varepsilon)/Z^X_\psi].
		\]
		Therefore, $(\M,f(\sigma')) \models_\rho  \infb^X(\psi,\varpi,\varpi_X,g(\sigma'))[\infb^X(nx(\psi),\varpi,\varepsilon)/Z^X_\psi]$, and therefore 
		\[(\M,f(\sigma')) \models_\rho  \bigvee_{\psi \in \Phi_g} \infb^X(\psi,\varpi,\varpi_X,g)[\infb^X(nx(\psi),\varpi,\varepsilon)/Z^X_\psi].\]
		But $\bigvee_{\psi \in \Phi_g} \infb^X(\psi,\varpi,\varpi_X,g)[\infb^X(nx(\psi),\varpi,\varepsilon)/Z^X_\psi]$ is an unfolding of $\infb^X(nx(\psi),\varpi,\varepsilon)$, and we conclude that 
		\[(\M,f(\sigma')) \models_\rho  \infb^X(nx(\psi),\varpi,\varepsilon).\]
\end{description}	
	%

	We now assume that $(\M,f(\sigma)) \models \infb^X(\psi,\varepsilon,\varepsilon)$ 
	and prove that 
	\begin{enumerate}
		\item $\sigma ~ \psi \in b$, and
		\item there is an infinite $\xrightarrow{X}$-path in $b$, from $\sigma~\psi$, where $X$ appears infinitely often, that only visits prefixes that extend $\sigma$.
	\end{enumerate}
	It is straightforward to see that for every prefix $\sigma \in \sigma(b)$, $(\M,f(\sigma)) \models \infb^X(\psi,\varpi,\varpi_X)$ implies that $\sigma~\psi \in b$.

	Let $\rightarrow$ be an lfp-finite dependency relation, as ensured by Theorem \ref{thm:lfp-finite-dep}. We demonstrate how to recursively construct an infinite $\xrightarrow{X}$-path in $b$, starting from $\sigma~\psi$, where $X$ appears infinitely often. 
	Let $p_0 = \sigma~\psi = \sigma_0~\psi_0$. For a finite $\xrightarrow{X}$-path $p_i$ in $b$ that ends in $\sigma_i~\psi_i$, such that 
	$(\M,f(\sigma_i)) \models \infb^X(\psi_i,\varpi^i,\varpi^i_X)$, we show how to extend the path to $p_{i+1}$ 
	that ends in $\sigma_{i+1}~\psi_{i+1}$, such that 
	$(\M,f(\sigma_{i+1})) \models \infb^X(\psi_{i+1},\varpi^{i+1},\varpi^{i+1}_X)$, and $$(\sigma_i,\infb^X(\psi_i,\varpi^i,\varpi^i_X)) \rightarrow^+ (\sigma_{i+1},\infb^X(\psi_{i+1},\varpi^{i+1},\varpi^{i+1}_X))~.$$

	By observing the definition of  
	$
	\infb^X(\psi_i,\varpi^i,\varpi^i_X)
	$, 
	and using Lemma \ref{lem:finb-correct},
	it
	is not hard to observe that 
	$$(\sigma_i,\infb^X(\psi_i,\varpi^i,\varpi^i_X)) \rightarrow^+ (\sigma_{i+1},\infb^X(\psi_{i+1},\varpi^{i+1},\varpi^{i+1}_X))$$ for some $(\sigma_{i+1},\infb^X(\psi_{i+1},\varpi^{i+1},\varpi^{i+1}_X))$, where 
	\begin{enumerate}
		\item $\sigma_i R_a \sigma_{i+1}$ for some $a \in \act_\varphi$;
		\item there is a finite $\xrightarrow{X}$-path $\varkappa_i$ in $b$ from $\sigma_i~\psi_i$ to $\sigma_{i+1}~\psi_{i+1}$; and 
		\item $\varkappa_i$ goes through some $\sigma'~X$ if and only if $\varpi_X^{i+1} = \varepsilon$.
	\end{enumerate}
	We define $p_{i+1} = p_i\varkappa_i$. The infinite path $p$ is the path that has every $p_i$ as a prefix.
	If $X$ does not appear infinitely often in $p$, then there is some $i > 0$, such that every $\varkappa_j$, where $j >i$, does not go through any $\sigma'~X$. The dependency relation $\rightarrow$ goes through finitely many formulas, and therefore it goes through some infinitely many times. Let $\infb^X(\psi',\varpi,\varpi_X)$ be the largest subformula of $\infb^X(\psi,\varepsilon,\varepsilon)$ that appears infinitely often as a $\infb^X(\psi_j,\varpi^j,\varpi^j_X)$. 
	Then, we observe that there is an infinite $\xrightarrow{Z_{\psi'}^X}$-path from some $(f(\sigma_j),\infb^X(\psi',\varpi,\varpi_X))$ that visits $Z_{\psi'}^X$ infinitely often. Therefore, the dependency relation is not lfp-finite, contradicting our assumption that $(\M,f(\sigma)) \models \infb^X(\psi,\varepsilon,\varepsilon)$. We conclude that $X$ appears infinitely often in $p$, which completes our proof.
\end{proof}

\begin{thm}\label{thm:tableau-reduction}
	The formula $\logicize{rules} \land \Inv(\neg \logicize{inf\_path}) \land \bigvee_{\varphi \in g \in G^t} g$ is satisfiable if and only if $\varphi$ is $\logicize{L}$-satisfiable.
\end{thm}
\begin{proof}
	If $\logicize{rules} \land \Inv(\neg \logicize{inf\_path}) \land \bigvee_{\varphi \in g \in G^t} g$ is satisfiable, then there is a model $(\M, s) \models \logicize{rules} \land \Inv(\neg \logicize{inf\_path}) \land \bigvee_{\varphi \in g \in G^t} g$, and therefore $(\M, s) \models \logicize{rules}$. By Lemma \ref{lem:rules-iff-ex-b--rep-impl-max-not-closed}, $\M,s$ represent a maximal branch $b$ that is not locally closed. 
	From $(\M, s) \models \bigvee_{\varphi \in g \in G^t} g$, we have that $b$ is a branch of a tableau for $\varphi$, and from 
	$(\M, s) \models  \Inv(\neg \logicize{inf\_path})$ and Lemma \ref{lem:infb-correct}, that $b$ is open.
    Theorem \ref{thm:tableaux} now yields that $\varphi$ is $\LG$ satisfiable and we are done.

	On the other hand, if $\varphi$ is $\logicize{L}$-satisfiable, then its tableau has an open  branch, and similarly, by Lemmata \ref{lem:b-to-model-representing}, \ref{lem:rules-iff-ex-b--rep-impl-max-not-closed}, and \ref{lem:infb-correct}, 
	$\logicize{rules} \land \Inv(\neg \logicize{inf\_path}) \land \bigvee_{\varphi \in g \in G^t} g$ is satisfiable.
\end{proof}

\begin{cor}\label{cor:no5-in-eexp}
	$\logicize{L}$-satisfiability is in $\EEXP$, for every logic $\logicize{L}$, where no agent has constraint~$5$.
\end{cor} 

\begin{proof}
	Given a formula $\varphi$, the height of the syntax tree of 
	each $\finb(T,\varepsilon),\finb^X(T,\varepsilon)$, and $\infb^X(\psi,\varepsilon)$, due to the dependency on the strings  $\varpi$, is $O(|\varphi|^2)$. 
	 The syntax tree furthermore branches with at most $2^{|\varphi|}$ disjunctions or conjunctions at each node. Therefore, the formula $\logicize{rules} \land \Inv(\neg \logicize{inf\_path}) \land \bigvee_{\varphi \in g \in G^t} g$ has size at most $2^{O(|\varphi|)}$. 
	Each step of the construction only requires up to double-exponential time with respect to $|\varphi|$.
	Theorem \ref{prp:muCalc-sat} then yields the resulting double-exponential upper bound.
\end{proof}

\subsection{The case of logics with \NI}

We now explain how to adjust the construction of $\logicize{rules}$ and $\logicize{inf\_path}$, to also handle the case of agents with negative introspection. 
What differs in this case, is that the tableau rules for these agents relate prefixes that are not necessarily consecutive. 
These prefixes are, however, close, in that their closest common ancestor is at worst their parent prefix.
Therefore, the tableau rules are still, in that sense, local.
We adjust the contents of the graphs that we use as variables, to take into account these extended formula interactions through the tableau rules.
In effect, in the following construction, graphs are equipped with limited lookahead capabilities. Still, a graph represents one tableau prefix, but it also contains certain prefixed formulas, so it can describe certain extensions of the current prefix. This allows our formulas to describe $\xrightarrow{X}$-paths that traverse non-consecutive tableau prefixes, resulting from the tableau rules in Table \ref{tab:tableau} for agents with constraint $5$.

\changeJ{
We extend the set of formulas that appear in a graph, to include certain prefixed formulas. We also keep track of any agent with condition $5$ that appears at the end of the prefix. 
Let $\act_5 = \{ \al \in \act \mid \logicize{L}(\al) \text{ has condition } 5 \}$.
For every $\al \in \act_5$, 
let:
}
%
%
\changeJ{
\begin{align*}
\sub_5^{\al\al}(\varphi) &=
\{ \varepsilon~\psi, ~ \beta\diam{\psi_1}~\psi, ~ \beta\diam{\psi_1}.\beta\diam{\psi_2}.\psi
\mid \psi, \psi_1,\psi_2 \in \sub(\varphi),~ \al \neq \beta \in \act_5
\}
\\ 
\sub_5^{\al}(\varphi) &= \sub_5^{\al\al}(\varphi) \cup 
\{ \al\diam{\psi_1}~\psi \mid \psi, \psi_1 \in \sub(\varphi)  \}
\\
\sub_5^\varepsilon(\varphi) &= \sub_5^{\al}(\varphi) \cup 
\{ \al\diam{\psi_1}.\al\diam{\psi_2}.\psi \mid \psi, \psi_1,\psi_2 \in \sub(\varphi)  \}
.
\end{align*}
}

\changeJ{
We modify the definition of $\G$, so that the graphs in $\G$ are of the form 
\[g = (\Phi,(\xrightarrow{X})_{X \in mV},\ell,b),\] 
where $b \in \{ \varepsilon, \al, \al\al \mid \al \in \act_5 \}$, and 
$\Phi \subseteq \sub_5^b(\varphi)$.
We also define $\G^b = \{ (\Phi,{(\xrightarrow{X})}_{X \in mV},\ell,b) \in \G \}$, for each  
$b \in \{ \varepsilon, \al, \al\al \mid \al \in \act_5 \}$.
For a tableau prefix $\sigma$ and $\delta \in \sub_5(\varphi)$, we say that $\sigma$ is $\delta$-flat when $\sigma = \sigma'.\al\diam{\psi}$ and $\delta = \al.\diam{\psi'}.\delta'$ or $\delta = \diam{\al}\psi'$ for some $\al, \psi, \psi'$, such that $\logicize{L}(\al)$ has condition $\logicize{5}$.
We say that $\al \in \act$ is $b$-inactive when $b = \al \diam{\psi} b'$ for some $\psi, b'$.
A tableau rule $\frac{\sigma~\psi}{\sigma'~\psi'}$ is $b$-active when $\psi$ is not of the form $[\al]\psi''$ or $\diam{\al}\psi''$, where $\al$ is $b$-inactive.}

\changeJ{
The purpose of the additional prefixes in the graph formulas is to control the application of the tableau rules that relate to condition $5$. 
A graph $g$ that represents a prefix $\sigma$ may also describe the branch at later prefixes.
This ensures that our constructed formula has models that adhere to the more involved tableau rules and the conditions for applying these rules.
However, this alteration does not affect the construction of the formulas for the infinite paths.
In this case, however, we need a more detailed treatment of the tableau rules, due to their more complex behaviour. Figure \ref{fig:lookahead} (top) illustrates the structure of these new graphs.
}

\begin{figure}
    \centering
\begin{tikzpicture}[->,>=stealth',shorten >=1pt,auto,node distance=2.5cm,thin]
  \tikzstyle{area} = [draw, dotted, rectangle, minimum height=3cm, minimum width=3cm]
  \tikzstyle{edge} = [line width=0.15,densely dashed]

  \filldraw[gray!5, draw=gray!80,label=right:$g$] (-2,-4) -- (8,-4) -- (8,6.8) -- (-2,6.8) -- cycle;
  \filldraw[gray!5, draw=gray!80,label=right:$h$] (1,-11.5) -- (10.5,-11.5) -- (10.5,-4.5) -- (1,-4.5) -- cycle;

  \node (g) at (7.5,6.3) {$g$};
  \node (h) at (10,-5) {$h$};


  
  \node[area,label=150:$\varepsilon$] (area11) at (3,5) {};
  \node[area,label=150:$\al\diam{\psi_1}$] (area21) at (1,1.5) {};
  \node[area,fill=gray!15,label=150:$\al\diam{\psi_2}$] (area31) at (6,1.5) {};
  \node[area,fill=gray!15,label=150:$\al\diam{\psi_2}.\al\diam{\psi_3}$] (area41) at (5,-2) {};

  \node[area,fill=gray!15,label=150:$\varepsilon$] (area32) at (5,-6.3) {};
  \node[area,fill=gray!15,label=150:$\al\diam{\psi_3}$] (area42) at (4,-9.8) {};
  \node[area,label=150:$\beta\diam{\psi_4}$] (area12) at (8.5,-9.8) {};

  \def\marg{0.1}
  \def\margc{0.9}

  \foreach \x in {1,2,3,4} {
    \foreach \y in {1,2,3,4} {
      \if\y1
        \def\xx{\marg}
        \def\yy{1.5+\marg}
      \else 
        \if\y2
        \def\xx{\marg}
        \def\yy{\marg}
      \else 
        \if\y3
        \def\xx{1.5+\marg}
        \def\yy{\marg}
      \else 
        \if\y4
        \def\xx{1.5+\marg}
        \def\yy{1.5+\marg}        
      \fi
      \fi
      \fi  
      \fi 
      \pgfmathsetmacro{\rndx}{rand}
      \pgfmathsetmacro{\rndy}{rand}
      \def\rx{\rndx*\rndx*\margc+\xx}
      \def\ry{\rndy*\rndy*\margc+\yy}
      \if\x1
        \def\prf{\varepsilon}
      \else 
        \if\x4
        \def\prf{\al\diam{\psi_2}.\al\diam{\psi_\x}}
      \else 
        \def\prf{\al\diam{\psi_\x}}
      \fi 
      \fi 
      \if\x1
      \node (\x-\y) at ($ (area\x1.south west) + (0.3,0.3) + (\rx, \ry) $) {$\circ$};
      \else 
      \if\x2
      \node (\x-\y) at ($ (area\x1.south west) + (0.3,0.3) + (\rx, \ry) $) {$\circ$};
      \else 
      \foreach \g in {1,2}  {
      \node (\x-\y-\g) at ($ (area\x\g.south west) + (0.3,0.3) + (\rx, \ry) $) {$\circ$};
      }
      \fi \fi 
    }
  }

  \node (4-4-3) at ($ (area12.south west) + (0.3,0.3) + (1.3, 2) $) {$\circ$};
  \draw[edge] (3-3-2) -> (4-4-3);
  

  \draw[edge] (1-1) -> (1-2);
  \draw[edge] (1-1) -> (1-3);
  \draw[edge] (1-2) -> (1-4);
  \draw[edge] (1-4) -> (1-3);
  
  \draw[edge] (1-2) -> (2-1);
  \draw[edge] (1-2) -> (2-4);
  \draw[edge] (1-3) -> (2-4);

  \draw[edge] (2-1) -> (2-2);
  \draw[edge] (2-2) -> (2-3);
  \draw[edge] (2-2) -> (2-4);

  \draw[edge] (1-3) -> (3-1-1);
  \draw[edge] (2-4) -> (3-3-1);
  \draw[edge] (2-3) -> (4-1-1);
  \draw[edge] (3-2-1) -> (3-4-1);
  \draw[edge] (3-3-1) -> (4-4-1);
  \draw[edge] (4-2-1) -> (3-2-1);
  \draw[edge] (3-1-1) -> (2-2);
  \draw[edge] (3-4-1) -> (1-2);
  \draw[edge] (4-1-1) -> (4-3-1);
  \draw[edge] (3-4-1) -> (3-3-1);
  \draw[edge] (4-4-1) -> (4-2-1);
  \draw[edge] (3-2-2) -> (3-4-2);
  \draw[edge] (3-3-2) -> (4-4-2);
  \draw[edge] (4-2-2) -> (3-2-2);
  \draw[edge] (4-1-2) -> (4-3-2);
  \draw[edge] (3-4-2) -> (3-3-2);
  \draw[edge] (4-4-2) -> (4-2-2);
  
\end{tikzpicture}
%
    
    
    
    \caption{Graphs 
    for logics 
    with Condition $5$. Each rectangle represents a graph and 
    each dotted square represents the set of formulas that are prefixed by 
    its label.
    Graph $h$ is a prospective $\al\diam{\psi_2}$-child of $g$ and we see that the darker squares of the two graphs match, as in Definition \ref{def:achild5}.}
    \label{fig:lookahead}
\end{figure}
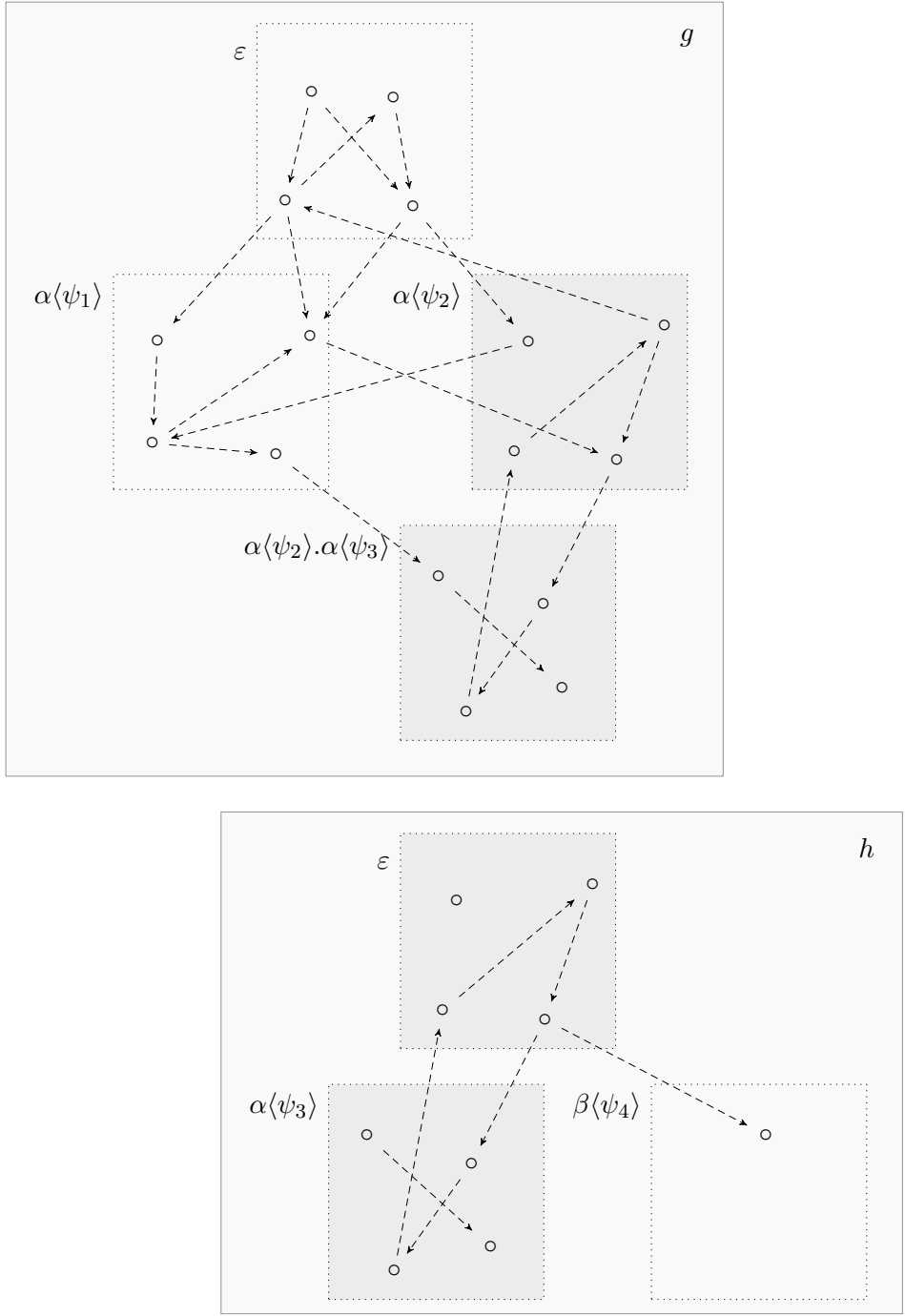

\begin{defi}\label{def:compatible5}
We  say that $g = (\Phi,(\xrightarrow{X})_{X \in mV},\ell,b)\in \G$ is \emph{locally 
compatible with the tableau rules} and write $com(g)$ when for every $\delta = \varsigma~\psi \in \Phi$, $\beta \neq \al$, and $X \in mV$, where $\psi \neq Y$ for any $X<Y$, the following conditions are satisfied. For every pair $\sigma, \sigma'$ of possible tableau prefixes and $\delta',\delta_1,\delta_2 \in \sub_5(\varphi)$, where $\sigma$ is not $\delta$-flat:
\begin{enumerate}[align=left] 
	\item 
	if $\frac{\sigma.\varsigma~\psi}{\sigma.
    \varsigma~
    \psi'}$ is a $b$-active tableau rule, 
    then 
    $\varsigma~\psi' \in \Phi$ and $\delta \xrightarrow{X} \varsigma~\psi'$;
    \item 
	if $\frac{\sigma.\delta}{\sigma.\delta'}$ 
    is a $b$-active instance of tableau rule ($\mathsf{D}$), ($\mathsf{d}$), 
    ($\mathsf{4}$), ($\mathsf{B5}$), ($\mathsf{D5}$), ($\mathsf{b}$), or ($\mathsf{b4}$), 
    then 
    $\delta' \in \Phi$ and $\delta \xrightarrow{X} \delta'$;
    \item 
	if $\frac{\sigma.\delta}{\sigma.\delta'}$ 
    is a $b$-active instance of tableau rule ($\mathsf{B}$), 
    $\delta = \varsigma~[\al]\psi_1$, and $\varsigma.\al\diam{\psi_2}~\psi_3 \in \Phi$ for some $\psi_1,\psi_2,\psi_3 \in \sub(\varphi)$,
    then 
    $\delta' \in \Phi$ and $\delta \xrightarrow{X} \delta'$;
    \item 
	if $\frac{\sigma.\delta}{\sigma.\delta'}$ 
    is a $b$-active instance of tableau rule ($\mathsf{B55}$), 
    $\delta = \varsigma~[\al] \psi_1$, and $\sigma.\al\diam{\psi_2}~ \psi_3 \in \Phi$ for some $\psi_1,\psi_2,\psi_3 \in \sub(\varphi)$,
    then 
    $\delta' \in \Phi$ and $\delta \xrightarrow{X} \delta'$;
	\item 
	if $\frac{\sigma.\delta}{\sigma.\delta'}$ 
    is a $b$-active instance of tableau rule ($\mathsf{D55}$), 
    $\delta = \varsigma~\diam{\al} \psi'$, and $\sigma~\diam{\al} \psi' \notin \Phi$ for some $\psi' \in \sub(\varphi)$,
    then 
    $\delta' \in \Phi$ and $\delta \xrightarrow{X} \delta'$;
	\item 
	if 
	$\frac{\sigma.\delta}{\genfrac{}{}{0pt}{1}{\sigma.\delta_1}{\sigma.\delta_2}}$ is an instance of a $b$-active tableau rule, 
    then 
	$\delta_1, \delta_2 \in \Phi$, $\delta \xrightarrow{X} \delta_1$, 
    and $\delta \xrightarrow{X} \delta_2$;
	\item 
	if $\frac{\sigma.\delta}{\sigma.\delta_1~\mid~ \sigma.\delta_2}$ is a $b$-active tableau rule, then $\delta_1 \in \Phi$ and $\delta \xrightarrow{X} \delta_1$, or $\delta_2 \in \Phi$ and $\delta \xrightarrow{X} \delta_2$;
	\item 
	if $\frac{\sigma.\varsigma~\psi}{\sigma.\varsigma'~\psi'}$, where 
    $\varsigma'~\psi' \notin \sub_5(\varphi)$ and
	$\varsigma'$ is a strict extension of $\varsigma$,  is a $b$-active tableau rule, then 
	$(out,bot) \in \ell(\delta)$;
	\item 
	if $\frac{\sigma.\varsigma.\psi}{\sigma'~\psi'}$ is a $b$-active tableau rule, 
    where 
        $\sigma.\varsigma$ is a strict extension of $\sigma'$,  
    then
	$(out,top) \in \ell(\delta)$; 
        and
	\item 
	$\varsigma~\neg \psi \notin \Phi$ and $\varsigma~\false \notin \Phi$.
\end{enumerate}
Let $\G^t = \{ g \in G \mid com(g) \}$. 
\end{defi}

\changeJ{
Let $\sigma = \varepsilon$ or $\sigma = \al\diam{\psi}$, where $\al \in \act_5$ and $\psi \in \sub(\varphi)$, 
and $g = (\Phi,(\xrightarrow{X})_{X \in mV},\ell,b)\in \G$. 
We define $g^\sigma = (\Phi^\sigma,(\xrightarrow{X}^\sigma)_{X \in mV},\ell^\sigma)$, where 
$\Phi^\sigma = \{ \delta \mid \sigma.\delta \in \Phi \}$; 
$\xrightarrow{X}^\sigma = \{ (\delta_1,\delta_2) \in (\Phi^\sigma)^2 \mid \sigma.\delta_1 \xrightarrow{X} \sigma.\delta_2 \}$; and 
$\ell^\sigma(\delta) = \ell(\sigma.\delta) \cap \{(out,bot)\}$ for every $\delta \in \Phi^\sigma$.}

\begin{defi}\label{def:achild5}
Let $\al \in \act$, $\chi \in \sub(\varphi)$, and 
$$g = (\Phi_g,(\xrightarrow{X}_g)_{X \in mV},\ell_g,b_g), ~h =(\Phi_h,(\xrightarrow{X}_h)_{X \in mV},\ell_h,b_h) \in G^t.$$
If $\al \in \act_5$, then we say that $g $ is a prospective $\al\diam{\chi}$-child of $h$, and write $c_{\al\diam{\chi}}(h,g)$ when
\begin{enumerate}
    \item $\diam{\al}\chi \in \Phi_h$, $\diam{\al}\chi \xrightarrow{X}_h \al\diam{\chi}~\chi$ for some $X$, $\chi \in \Phi_g$, and $(in,top) \in \ell_g(\chi)$;
    \item $b_h \neq \al\al$; 
    \item $g^\varepsilon = h^{\al\diam{\chi}}$; and 
    \item if $b_h = \al$, then $b_g = \al\al$, otherwise $b_g = \al$.
\end{enumerate}

If $\al \notin \act_5$, then we say that $g $ is a prospective $\al\diam{\chi}$-child of $h$, and write $c_{\al\diam{\chi}}(h,g)$ when:
\begin{enumerate}
	\item $\diam{\al}\chi \in \Phi_h$ and $\chi \in \Phi_g$ and $(in,top) \in \ell_g(\chi)$;
	\item if $[\al]\psi \in \Phi_h$ and $\frac{\sigma.~[\al]\psi}{\sigma.\al\diam{\psi_1}~\psi_2}$ is a tableau rule, then $\psi_2 \in \Phi_g$ and $(in,top) \in \ell_g(\psi_2)$;  
	\item if $[\al]\psi \in \Phi_g$ and $\frac{\sigma.\al\diam{\psi_1}~[\al]\psi}{\sigma~\psi_2}$ is a tableau rule, then $\psi_2 \in \Phi_h$ and $(in,bot) \in \ell_h(\psi_2)$; and
    \item $b_g = \varepsilon$.
\end{enumerate}
We can similarly define that $g $ is a prospective $\al$-child of $h$, and write $c_{\al}(h,g)$, by omitting the first condition in each case.
\end{defi}

Figure \ref{fig:lookahead} depicts two graphs, $g$ and $h$, where $c_{a\diam{\psi_2}}(g,h)$ for some formula $\psi_2$.

The size of $\G$ is still at most exponential with respect to $\varphi$ and it takes at most polynomial time to verify $com(g)$ and $c(g,h)$; furthermore, $|\act_5^\varphi|\leq 2|\act|\cdot |\sub(\varphi)|^2$.


Let $b$ be a tableau branch. 
For the remaining construction, to express when there exists a path that visits some $a~X$ infinitely often, it does not matter what objects the vertices of the graphs are.
We can then continue with the construction, similarly as for the case of logics with no agents with negative introspection.

\begin{cor}\label{cor:all-in-eexp}
	$\logicize{L}$-satisfiability is in $\EEXP$.
\end{cor}

%% file: conclusions.tex

In this work we studied the complexity of the satisfiability problem for multi-modal logics with recursion.
 These logics mix the frame conditions from epistemic modal logic and the recursion operators of the $\mu$-calculus.
 
In the first part we gave translations among several of these logics that connect their satisfiability problems.
This approach allowed us to offer complexity bounds for
satisfiability and to prove finite-model results. We also presented a sound and complete tableau
that has termination guarantees, conditional on a logic’s finite-model property.
Using our translations, we obtained several hardness and completeness results for individual logics and their combinations, which are summarized in Table \ref{tab:summary}, with a few isolated cases missing.

In the second part of this work we aimed to fill these gaps in complexity results, specifically for the case of the multi-agent $\mu$-calculus over frames containing the $\blogic$ and $\kv$ restrictions and the case of single-agent $\mu$-calculus over transitive frames. 
To tackle the multi-agent case, we gave a translation from tableaux to $\mu$-calculus formulas, which allowed us to use the known algorithm for testing satisfiability over the  $\mu$-calculus  in that setting. However, the translation we provided was not polynomial in size. Thus, it yielded an upper bound on the complexity of satisfiability that did not match the \PSPACE\ hardness result stemming from classic work on modal logics. To address this complexity gap, we customized the tableaux exploration algorithm we developed to utilize the transitivity condition on frames. Doing so allowed us to establish a tight \PSPACE\ bound for the satisfiability problem in that case. 

Table \ref{tab:updated_summary} summarises the results. 

\subsection{Discussion on the existence of a translation from \texorpdfstring{$\kf^{\mu}$ to $\kf$}{K4\textasciicircum mu to K4}}

\begin{table}
	\centering
\begin{tabular}{|l|l|l|l|}
\hline
$\#$ agents      & Restrictions on syntax/frames & \multicolumn{1}{l|}{Upper Bound} & Lower Bound \\ \hline
                      & frames with B or 5 & {\color{purple}\textbf{\EEXP  }} \ref{cor:all-in-eexp} & \EXP-hard  \ref{prp:EXP-hard-k2}  \\
$\geq 2$ & not B , 5                                  & \EXP~ \ref{cor:mu-upper1} & \EXP-hard \ref{prp:EXP-hard-k2}   \\
                      & not 5, not $\mu$. X                             & \EXP~ \ref{cor:mu-upper2} & \EXP-hard   \ref{prp:fragmantsMu} \\ \hline
                    & with 5 (or B4)             & \NP~ \ref{thm:ladhalp} & \NP-hard   \ref{thm:ladhalp}  \\
       1               & with 4                                      & {\color{purple}\textbf{\PSPACE  }} \ref{thm:4-PSPACE} & \PSPACE-hard \ref{thm:ladhalp} \\
                      & Any other restrictions                       & \EXP~  \ref{prp:muCalc-sat} & \EXP-hard  \ref{prp:muCalc-sat}  \\ \hline
\end{tabular}
\caption{The updated summary of the complexity of satisfiability checking for various modal logics with recursion. Our additional contributions are highlighted.}
\label{tab:updated_summary}
\end{table}

The reader might wonder whether some of the complexity results we obtained with the use of tableaux might have also been obtained via translations.
For example, Section \ref{sec:KfourMu} is dedicated to proving the tight bound for the complexity of satisfiability for formulas in $\kf^\mu$. 
We will discuss here why we think this would be very hard to do via translations. 
To apply our translation technique over $\kf^\mu$, we would need to translate formulas in $\kf^\mu$ to formulas of some logic whose satisfiability problem is in \PSPACE. A natural choice would be modal logics without recursion.
 Amongst those, a prime candidate would be $\kf$ as the transitivity of the models would be easier to tackle.
Therefore we would need to focus our efforts on developing ``efficient'' translations that remove fixed-point operators from formulas. 

The above task seems relatively easy for $\Box$-free formulas since we can remove consecutive diamonds from such formulas due to the transitivity restriction on the frames. 
This is straightforward over transitive frames as any subformula starting with two or more consecutive diamonds is exactly satisfied when the similar formula, where all the consecutive diamonds are merged into one, is satisfiable. 

At the same time, we know that we can express the invariance operator $\Inv(-)$ in $\kf$ by using a box, and therefore we can
simulate some form of recursion. 
We would hope that an intuition like this would help us design a translation. 
Unfortunately, it turns out that tackling arbitrary formulas with boxes is less easy.
We now discuss the use of our translation formulas for the case of $\kf^{\mu}$. 

Till now, all translations we have presented have a simple characteristic. 
Namely, given a formula $\varphi$ over a logic $\LG$ and its translation $f(\varphi)$ for a logic $\LG'$, we have that 
all models that model the formula $f(\varphi)$, will also model the formula $\varphi$, or, in some sense, describe a model of $\varphi$. 
This can be seen by observing the way we have defined translations so far, and the proofs of their correctness. 
Being able to extract a model for $\varphi$ from a model of $f(\varphi)$ is an important aspect of our translations, because it is at the core of our arguments regarding compositionality (see Section \ref{subsec:compositionality}). 
 We demonstrate here our argument on why such a (well-behaved) translation may not be possible.
 
To do so, assume that such a translation exists and consider again the formula $\varphi_f=$ \change{$\mn$}$ X. \Box X$, which requires all paths in the model to be finite, from Example \ref{ex:all_paths_fin}.
The formula is both recursive and its satisfiability is not affected by the transitivity restriction.
Therefore, due to the translation, there should exist a satisfiable property $\phi'$ in $\kf$ that is the translation of $\varphi_f$. 
Let $k$ be the number of subformulas of $\phi'$.
We now write a formula $\psi \in \kf^{\mu}$ such that $\psi$ is satisfied on a (transitive) model having $2k$-many states connected in a path.
Said formula can even be expressed without recursion and a model $\mathcal{M}$ for it is shown in Figure \ref{fig:big_path_transitive_model} (before the addition of the red edge connecting $n_{i+1}$ to itself).

\begin{figure}
\begin{center}
\begin{tikzpicture}[scale=0.6,node distance=2.5cm]

\tikzstyle{every node}+=[inner sep=3pt]
\node (1) {$n_1$};
\node (2) [right of =1 ] {$n_2$};
\node (3) [right of =2] {$\ldots$};

\node (5) [right of =3]{$n_i$};
\node (6) [right of =5] {$n_{i+1}$};

\node (4) [right of=6]{$n_{2*k}$};

\path 
	(1) edge [->](2)
		edge [bend left, ->] (3)
		edge [bend left, ->] (4)
		edge [bend left, ->] (5)
		edge [bend left, ->] (6)

%
    (2) edge [->](3)
    		edge [bend right, ->] (4)
    		edge [bend right, ->] (5)
		edge [bend right, ->] (6)
	(3) edge [dashed,->] (5)
	(5) edge [->] (6)
	
	(6) edge [->,dashed] (4)
		edge [->,red, loop above] (6);
\end{tikzpicture}
\end{center}
\caption{The model $\mathcal{M}$ of property $\psi$, before and after adding the loop edge.}
\label{fig:big_path_transitive_model}
\end{figure}
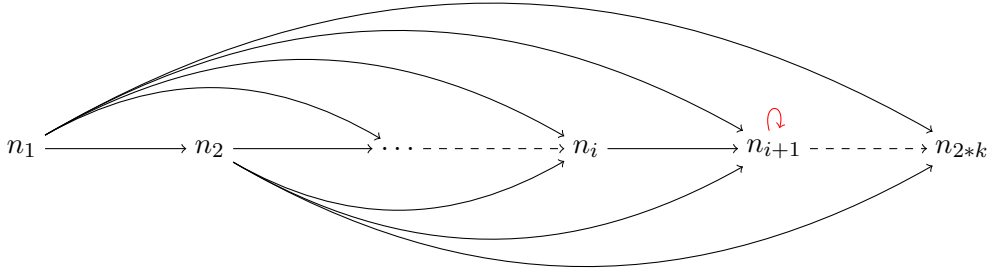

We now study $\mathcal{M}$ and its relation to $\phi'$. We see that we should have $\mathcal{M}\models \phi'$ as all paths of $\mathcal{M}$ are indeed finite.
 However, the length of the longest path of $\mathcal{M}$ is $2*k$ which is larger than the modal depth of $\phi'$.
  Moreover, since  $\mathcal{M}$ is transitive we know that as we ``move away'' from the root node, 
  more box-formulas are satisfied, due to the transitivity axiom $\Box p \implies \Box \Box p$. 
  Thus, the number of subformulas of $\phi'$ satisfiable in the states of $\mathcal{M}$ increases the further we move away from the root.
   Finally, we know the number of $\Box$ prefixed sub-formulas of $\phi'$ is finite and bounded by $k$. 

Thus, there must exist nodes $n_i$, $n_{i+1}$, as shown in Figure \ref{fig:big_path_transitive_model}, that satisfy exactly the same number of box-prefixed subformulas of $\phi'$.
 We now add the red loop-edge on $n_{i+1}$ as seen in the figure.
  The addition of said edge cannot affect the validity of any diamond-prefixed subformula of $\phi'$, since it only extends the reachable states of \changeJ{$n_{i+i}$}.
   Regarding the boxes, we inspect whether the addition of the loop edge could affect the boxes of $n_{i+1}$.
    However, the boxes of $n_{i}$ are the same as those of $n_{i+i}$ (before the addition of the loop), and due to transitivity and the form of $\mathcal{M}$ the states reachable by $n_i$ are exactly $n_{i+1}$ and the states reachable by \changeJ{$n_{i+1}$}.
     Therefore, all box-prefixed formulas of $n_i$ are also applied in $n_{i+1}$, and thus the addition of the loop 
     preserves the satisfaction
     of the box-formulas of $n_{i+1}$. 

We now have constructed the model $\mathcal{M}'$ (after the addition of the red loop-edge), which satisfies identical subformulas of $\phi'$ in the identical states as $\mathcal{M}$. 
 \changeJ{Note here that the assumption was that a translation that preserves satisfiability exists, and we used such a translation on the formula $\varphi$.}
 Thus, it must be the case that $\mathcal{M}' \models \phi'$. 
 \changeJ{However, 
 $\varphi_f$ states that all paths are finite, which means that $\mathcal{M} \models \varphi_f$, and $\mathcal{M}' \not \models \varphi_f$, while $\mathcal{M},\mathcal{M}' \models \varphi'$. 
 Therefore, it is hard to see how any translated formula would ``agree'' with the initial one on certain models, or describe models for $\varphi$. 
 So far the translations we have defined always work, with the proof producing a model that satisfies both the old formula and the new. 
 It is possible that more complicated translations exist that do not have this property, but we were not able to think of one. 
}

\paragraph{Conjectures and Future Work}
We currently do not have a tight complexity bound for the case of the multi-agent $\mu$-calculus over symmetric or euclidean frames. 
What is interesting is that we also do not have a counterexample to prove that the translations of Section \ref{sec:transl_mu_cal} that we already have, as well as other ones we attempted, do not preserve satisfiability. 

\change{The complexity of the model checking problem for the $\mu$-calculus is an important open problem, known to have a quasi-polynomial time solution, but not known whether it is in \P~\cite{DBLP:journals/lmcs/LehtinenPSW22,DBLP:conf/lics/Lehtinen18,jurdzinski2017succinct,calude2017deciding,fearnley2017ordered}.
The problem does not depend on the frame restrictions of the particular logic, though one may wonder whether additional frame restrictions would help solve the problem more efficiently.
Currently, we are not aware of a way to use our translations to obtain such an improvement.}



As, to the best of our knowledge, most of the logics described in this chapter have not been explicitly defined before, with notable exceptions such as \cite{DAGOSTINO20104273transitive,Dagostino2013S5,alberucci_facchini_2009}, they also lack any axiomatizations and completeness theorems.
We do expect the classical methods from \cite{Kozen1983,ladnermodcomp,Halpern1992} and others to work in these cases as well. However, it would be interesting and desirable to flesh out the details and see if there are any unexpected situations that arise.

Given the importance of common knowledge for epistemic logic and the fact that it has been known that common knowledge can be thought of as a (greatest) fixed point already from \cite{harman1977review,Barwise:1988:TVC:1029718.1029753}, we consider the logics that
we presented to be natural extensions of \ML.
Besides the examples given in Section \ref{sec:multi_ml_background}, we are interested in exploring what other natural concepts can be defined with this enlarged language.
Finally, it would be interesting to prove a result showing that, under some natural assumptions, there is no ``efficient'' satisfiability-preserving translation from $\kf^{\mu}$ to $\kf$ or other modal logics.